\documentclass[11pt]{article}
\usepackage{cherian}
\usepackage{amsmath, amsthm, amssymb}
\usepackage{todonotes}
\usepackage{subcaption}
\usepackage{algorithm}
\usepackage{algpseudocode}
\usepackage{booktabs}
\usepackage{multirow}
\usepackage[numbers,square]{natbib}
\usepackage[colorlinks = true, citecolor = blue]{hyperref}

\usepackage{cleveref}

\newcommand{\cX}{\mathcal{X}} 
\newcommand{\cY}{\mathcal{Y}} 
\newcommand{\cD}{\mathcal{D}} 
\newcommand{\cG}{\mathcal{G}} 
\newcommand{\cH}{\mathcal{H}} 

\renewcommand{\l}{\left}
\renewcommand{\r}{\right}

\newcommand{\indG}{\mathbf{1}_G}

\begin{document}

\title{Statistical Inference for Fairness Auditing}
\author{John J. Cherian\footnote{Department of Statistics, Stanford University}$\ $ and Emmanuel J. Cand\`{e}s\footnote{Departments of Statistics and Mathematics, Stanford University}}

\maketitle

\begin{abstract}
Before deploying a black-box model in high-stakes problems, it is important to evaluate the model’s performance on sensitive subpopulations. For example, in a recidivism prediction task, we may wish to identify demographic groups for which our prediction model has unacceptably high false positive rates or certify that no such groups exist. In this paper, we frame this task, often referred to as ``fairness auditing,''  in terms of multiple hypothesis testing. We show how the bootstrap can be used to simultaneously bound performance disparities over a collection of groups with statistical guarantees. Our methods can be used to flag subpopulations affected by model underperformance, and certify subpopulations for which the model performs adequately. Crucially, our audit is model-agnostic and applicable to nearly any performance metric or group fairness criterion. Our methods also accommodate extremely rich---even infinite---collections of subpopulations. Further, we generalize beyond subpopulations by showing how to assess performance over certain distribution shifts. We test the proposed methods on benchmark datasets in predictive inference and algorithmic fairness and find that our audits can provide interpretable and trustworthy guarantees.
\end{abstract}

\section{Introduction}
While black-box models may demonstrate impressive accuracy on average, their performance can still vary substantially between subpopulations. For example, an algorithm deployed for recidivism prediction exhibits significantly higher false positive rates for African-American relative to Caucasian parolees \cite{angwin2016machine}. Similar performance disparities have been documented in other high-stakes applications such as facial recognition and hiring \cite{buolamwini2018gender, dastin2018amazon}. 

Motivated by this concern, numerous stakeholders have solicited methods, often referred to as ``fairness audits,'' that can discover and quantify such disparities \cite{brundage2020toward, schaake2022audit}. Despite substantial prior work in this area \cite{morina2019auditing, xue2020auditing, diciccio20evaluating, tramer2017fairtest, taskesen2021statistical, si2021testing, yan2022active, von2022locating}, the definition of fairness auditing remains fraught. Fairness auditing is often framed as a single statistical test that rejects in the case of \emph{any} performance disparity over a limited set of sensitive subpopulations \cite{diciccio20evaluating, tramer2017fairtest, si2021testing, morina2019auditing, taskesen2021statistical,  xue2020auditing, roy2023fairness}. While follow-up investigation to localize disparities is desired (and often performed), this task  raises new challenges. For example, empirical parity across a limited collection of subgroups does not rule out substantial disparities among smaller subgroups \cite{kearns2018preventing}. Further, if we consider multiple performance metrics over a rich collection of subpopulations, discovering some disparity between two subgroups is hardly surprising. Unfortunately, existing methods for identifying localized (dis-)parities are accompanied by few statistical guarantees \cite{von2022locating, yan2022active, schaake2022audit}.

We develop a family of statistical methods that rigorously achieve two goals: (1) the ``certification'' of subpopulations for which the model performs adequately, and (2) the ``flagging'' of subpopulations that suffer harmful performance disparities. The proposed methods only require access to a so-called ``audit trail,'' i.e., model predictions on a data set held out from training \cite{brundage2020toward}, but not white-box access to the model itself. We provide a Python package, \textsf{fairaudit}, implementing these methods at \href{https://www.github.com/jjcherian/fairaudit}{github.com/jjcherian/fairaudit}.

\subsection{Preview of contributions} \label{sec:multiple_testing}
To motivate and summarize the main contributions of our paper, we preview two applications.

\subsubsection{Certifying conditional coverage} \label{sec:preview_certify}
Consider a training set $\{(X_i, Y_i)\}_{i = 1}^n$ and a test point $\{X_{n + 1}, Y_{n + 1})$ sampled i.i.d.~from some unknown distribution $P$. Using $\{X_i, Y_i\}_{i = 1}^n \cup \{X_{n + 1}\}$ as input, conformal prediction produces a set-valued function, denoted by $\hat{C}(\cdot)$, that satisfies the guarantee $\P (Y_{n + 1} \in \hat{C}(X_{n + 1}) ) \geq 1 - \alpha$ \emph{marginally} over the randomness in the training and test points. This marginal guarantee does not preclude loss of coverage, however, after we condition on $\hat{C}$ and $X \in G$. There may exist subsets $G \subseteq \mathcal{X}$ such that $\P ( Y \in \hat{C}(X) \mid \hat{C}, X \in G ) \ll 1 - \alpha$. Without strong assumptions on the data-generating distribution, guaranteeing conditional coverage over rich collections of subpopulations is known to be difficult \citep{barbercandeslimits2020}. 

\begin{figure}
\centering
\captionsetup[subfigure]{width=0.9\linewidth}%
\begin{subfigure}{.5\textwidth}
  \centering
  \includegraphics[width=\linewidth]{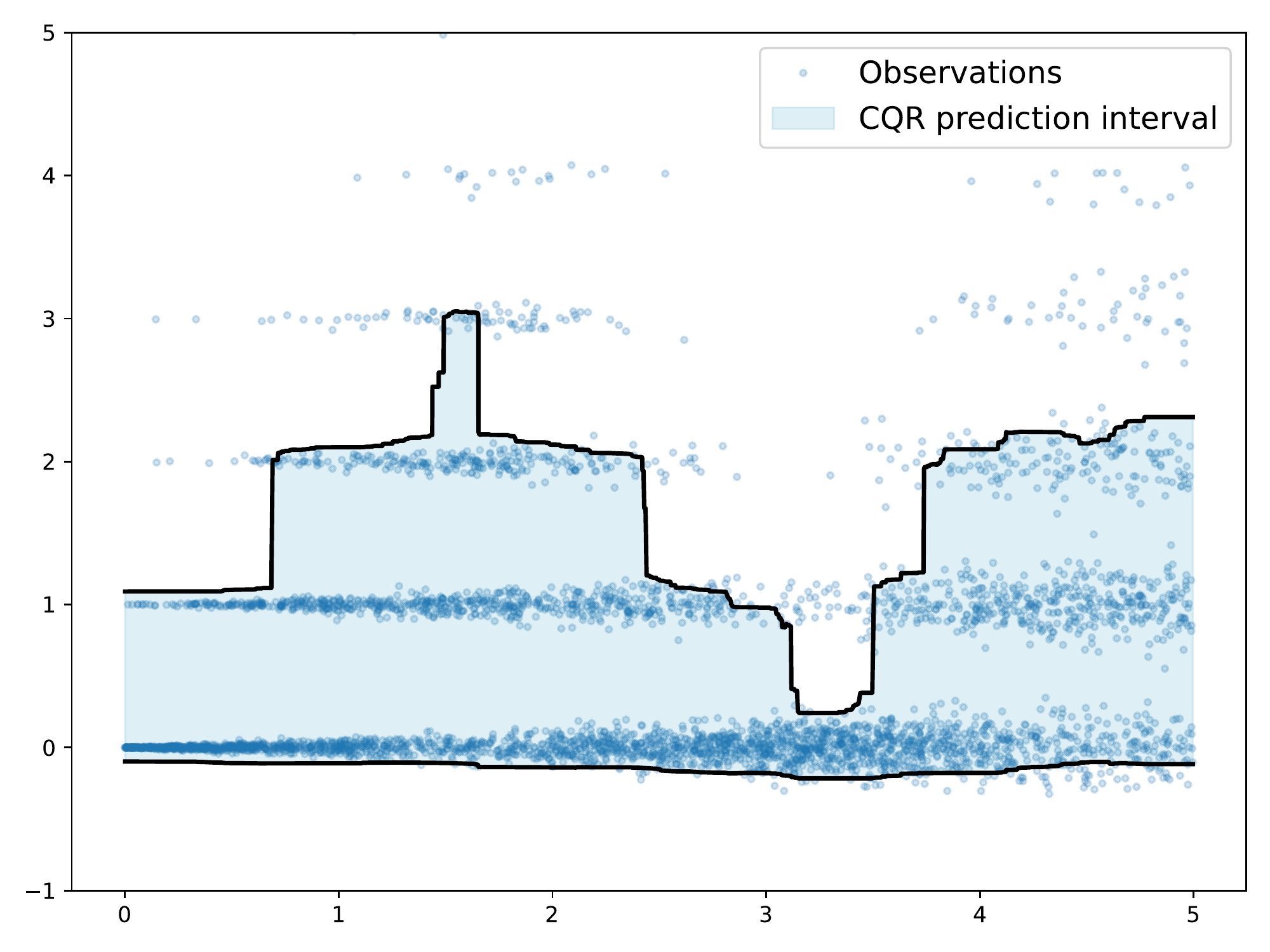}
  \caption{}
  \label{fig:cqr_interval}
\end{subfigure}%
\begin{subfigure}{.5\textwidth}
  \centering
  \includegraphics[width=\linewidth]{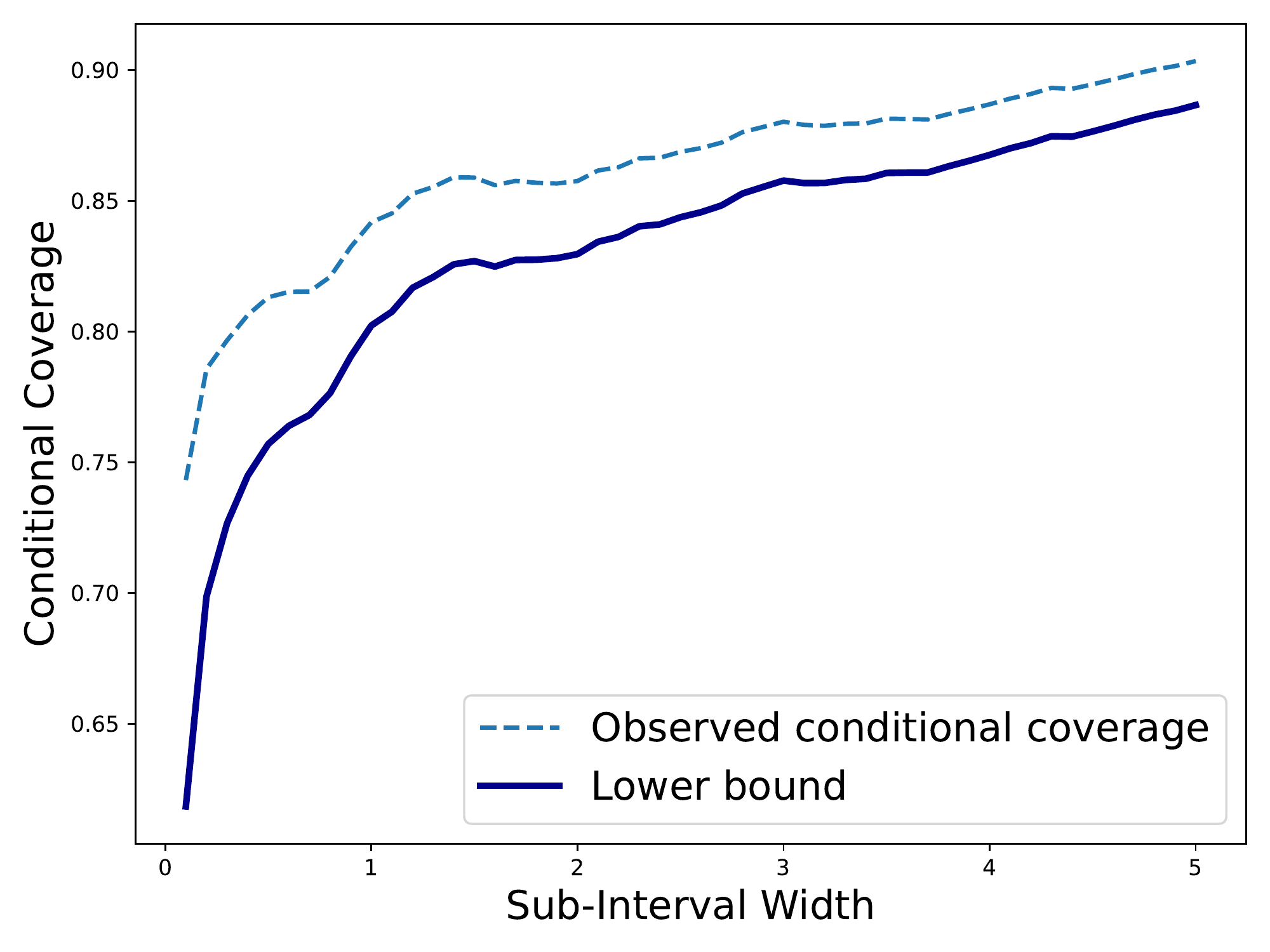}
  \caption{}
  \label{fig:cqr_bound}
\end{subfigure}
\caption{In \Cref{fig:cqr_interval}, we plot a $90\%$ conformal prediction set constructed using a quantile random forest on the synthetic dataset studied in \citep{romano2019conformalized}. \Cref{fig:cqr_bound} displays the prediction set's conditional coverage over all sub-intervals of a given width. The dashed line is the observed coverage, while the solid line plots the simultaneous lower bound obtained via \textsf{\Cref{algo:bootstrap_ci_rescaled}} ($\alpha = 0.05$, $p_* = 0.01$, $w_0 = \infty$).}
\label{fig:cqr_preview}
\end{figure}

We re-visit the synthetic data experiment in \citet{romano2019conformalized} to understand why one might wish to certify conditional coverage over a large collection of groups. After examining the prediction set plotted in \Cref{fig:cqr_interval} ($n = 4000$), we might perceive that the conditional coverage over most sub-intervals of $\cX$ does not deviate substantially from the marginal guarantee of 90\%. But can we make this observation rigorous? 

Given an arbitrary and potentially data-dependent collection of sub-intervals, we show how to issue a lower confidence bound on the conditional coverage for each sub-interval. To balance computational expense and complexity of the collection of groups considered, we limit our analysis to sub-intervals with endpoints belonging to $\{0,0.1,0.2,\dots,5\}$. The issued bounds will then hold simultaneously over all sub-intervals in this collection with high probability. Using the plotted audit trail, we lower bound the minimum conditional coverage over all sub-intervals of a given width, i.e., $\min_{G : \text{width}(G) = w} \P(Y \in \hat{C}(X) \mid \hat{C}, X \in G)$. \Cref{fig:cqr_bound} plots this simultaneously valid lower bound as well as the observed conditional coverage, $\min_{G : \text{width}(G) = w} \hat{\P}_n(Y \in \hat{C}(X) \mid \hat{C}, X \in G)$. For example, we can say with 95\% confidence that no sub-interval of length $1$ has conditional coverage worse than $80.2\%$, and no sub-interval of length $2$ has conditional coverage worse than $83.0\%$. Crucially, our guarantee is exact: in large samples, the probability that any lower bound is invalid converges to 5\%.

In this example, we issued simultaneously valid lower bounds over $1275$ sub-intervals. Our restriction to a finite collection of sub-intervals, however, is not necessary, and can be relaxed at the auditor's discretion. See \Cref{sec:certify_G} for a complete description of our procedure for producing lower bounds.
\subsubsection{Flagging false positive rate disparities} \label{sec:preview_flag}
Consider a district court evaluating whether the COMPAS recidivism prediction algorithm is biased. While the most notable previous work considers subpopulations defined by race \citep{angwin2016machine}, the court is likely to be interested in discrimination against any groups formed by the intersections of legally protected attributes, e.g., age,  gender, ethnicity. Over this larger collection of subpopulations, identifying the existence of some disparity is no longer of interest, but accurately localizing severe disparities is of great importance.

\begin{figure}
    \centering
    \includegraphics[scale=0.3]{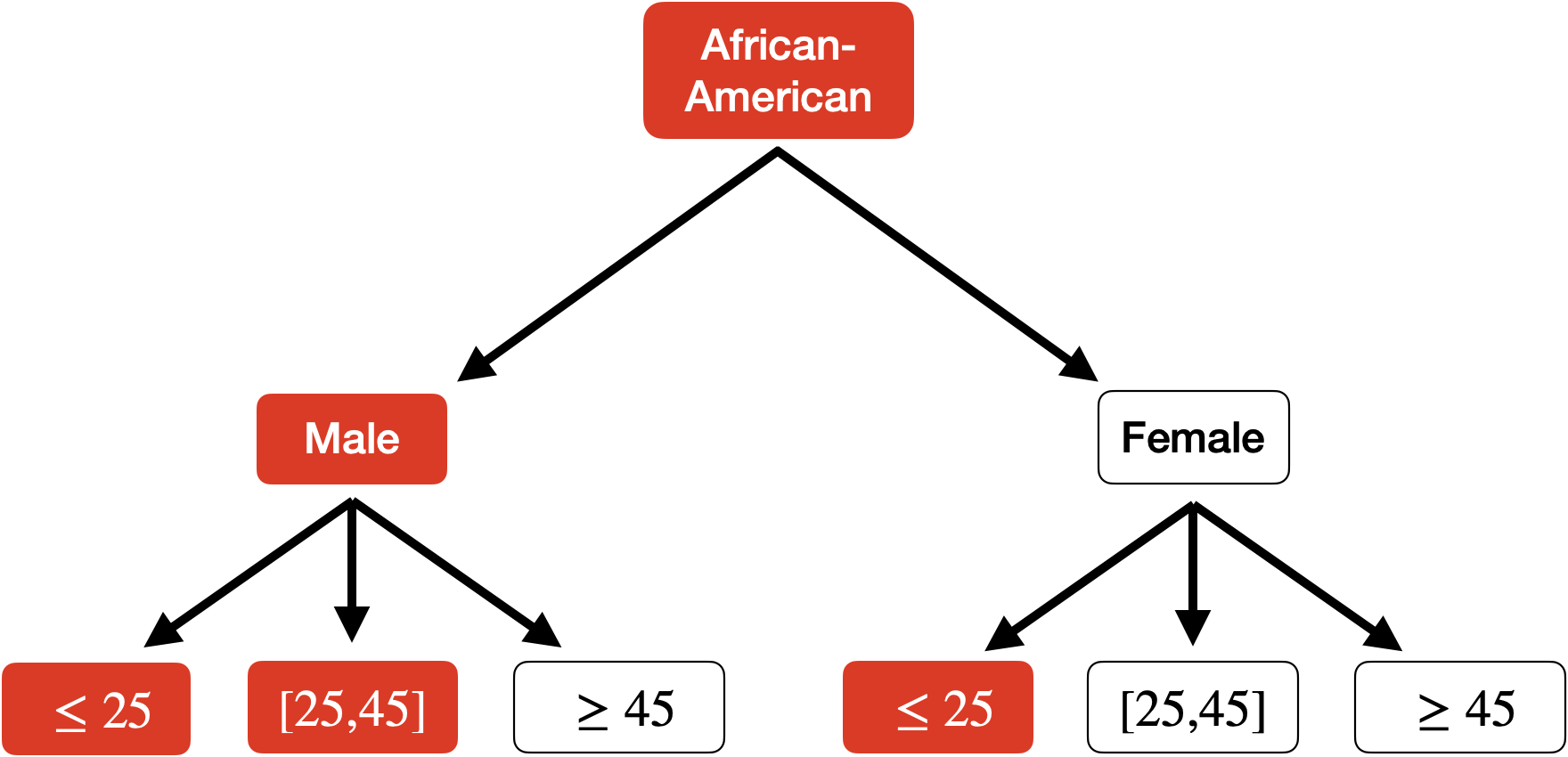}
    \caption{For the COMPAS algorithm, we audited 58 subpopulations formed by intersections of race, gender, and age. Here we plot subgroups of the African-American subpopulation. Boxes shaded in red denote groups flagged as having at least 5\% higher-than-average false positive rates.}
    \label{fig:compas_flags}
\end{figure}
Following \citet{angwin2016machine}, we apply our auditing method to a data set obtained by ProPublica in 2016 that includes COMPAS risk scores ($f(X) \in \{\text{low-risk}, \text{high-risk}\}$), defendant demographics ($X_{d}$), and true recidivism outcomes after two years ($Y \in \{0,1\}$) for $n = 6781$ individuals. In \Cref{fig:compas_flags}, we shade in red the demographic groups flagged for disproportionate false positive rates, i.e., those $G$ for which $$\P(f(X) = \text{high-risk} \mid Y = 0, X_d \in G) - \P(f(X) = \text{high-risk} \mid Y = 0) > 0.05.$$ 

To account for the inevitability of finding some disparity when auditing over many groups, our method issues a formal guarantee on the rate at which the issued flags are invalid. Since falsely flagging a performance disparity is less consequential than false certification of fairness, we provide a less conservative guarantee than the simultaneous validity of the previous example. Instead, we provide an asymptotically valid upper bound on the ``false discovery rate,'' i.e., the expected proportion of flags that are falsely issued. For the plotted example, we apply the flagging method described in \Cref{sec:flag_G} to control this proportion at 10\%. 

\subsection{Outline}

The rest of the paper is organized as follows. We first describe the problem setting and associated notation in \Cref{sec:auditing_method}. Each subsequent section is devoted to an auditing method. In \Cref{sec:certify_G}, we describe our procedures for certifying performance disparities over a collection of subpopulations. In \Cref{sec:flag_G}, we show how to flag performance disparities. Lastly, in \Cref{sec:beyond_subpop}, we show how to extend our methodology from subpopulations to collections of distribution shifts. 

\section{Set-up and notation} \label{sec:auditing_method}
We say that some prediction rule $f$ exhibits a performance disparity on a subpopulation $G$ if the mean of some metric $L(f(X), Y)$, conditional on $(X, Y)$ belonging to $G$, differs substantially from some target $\theta_P \in \R$. 

Our statistical audits proceed by testing and/or constructing bounds on the group-wise performance disparity,
\begin{align}
    \underbrace{\epsilon(G)}_{\text{disparity}} &\defeq \underbrace{\E_P[L(f(X), Y) \mid (X, Y) \in G]}_{\text{group-specific}} - 
    \underbrace{\theta_P}_{\text{target}}.
\end{align}
Subgroup membership may be defined by a subset of the covariates (e.g., certain sensitive attributes) that the prediction rule does not directly use. For notational simplicity, however, we will refer to the same covariate vector $X$ when defining both the prediction rule and group membership.

Nearly every group fairness definition can be expressed in terms of $\epsilon(G)$. For the false positive rate example, we audited a binary predictor given by $f(x)$. We then instantiated the flagging method for all intersectional subgroups with $Y = 0$ by testing the null hypothesis $H_0(G): \epsilon(G) \leq 0.05$ with 
\begin{align*}
L(f(x), y) = \indic{f(x) = 1} \quad\text{and}\quad \theta_P = \P (f(X) = 1 \mid Y = 0).
\end{align*}
In the conditional coverage example, we considered a \emph{fixed} set-valued predictor $\hat{C}(x)$. Then, we certified conditional coverage over all sub-intervals with endpoints in $\{0,0.1,\dots,5\}$ by lower bounding $\epsilon(G)$ for
\begin{align*}
    L(\hat{C}(x), y) = \mathbf{1}\{y \in \hat{C}(x)\} \quad\text{and}\quad \theta_P = 0.9.
\end{align*}
See \Cref{sec:fairness_examples} for additional examples.

To evaluate these disparities, the audit we devise requires access only to a hold-out data set $\cD = \{X_i, Y_i\}_{i = 1}^n~\simiid~P$; this is sometimes called an ``audit trail'' \citep{brundage2020toward}. 

If $\theta_P$ is not known a-priori, we will assume that it is possible to use this audit trail to construct a consistent estimator $\hat{\theta}$.  We will further assume that $\hat{\theta}(\cD) \defeq g( \sum_i h(x_i, y_i)/n)$ for some differentiable function $g$ and known features $h : \cX \times \cY \to \R^{k}$. We omit the argument of $\hat{\theta}$ when it is clear from context. An important implication of this assumption is asymptotic linearity, i.e., the existence of some $\psi(X, Y)$ such that $\sqrt{n}(\hat{\theta} - \theta_P) = \frac{1}{\sqrt{n}} \sum_{i = 1}^n \psi(X_i, Y_i) + o_P(1)$. Typically, $\theta_P = \E_P[L(f(X), Y)]$, and consequently, $\hat{\theta}(\cD) = \frac{1}{n} \sum_{i = 1}^n  L(f(x_i), y_i)$ will satisfy this assumption. 

We define the following terms for notational convenience. We replace $(X, Y) \in G$ with $G$ whenever the meaning is clear. For example, we let $\P_n(G) \defeq \frac{1}{n} \sum_{i =1}^n \indic{(x_i, y_i) \in G}$ denote the empirical probability of $(X, Y)$ belonging to $G$ in the hold-out set, and $\P(G) \defeq \P((X, Y) \in G)$ denote the population probability of $(X, Y)$ belonging to $G$. To further simplify our notation, we will also replace $L(f(X), Y)$ by the abbreviation $L$. We thus define the plug-in estimator of the disparity, $\hat{\epsilon}(G) \defeq |G|^{-1} \sum_{i \in G} L_i - \hat{\theta}$.

\section{Certifying performance}  \label{sec:certify_G}
\subsection{Methods}
\subsubsection{Bound certification}
We first consider the problem of \emph{certifying} subpopulations by providing a simultaneously valid confidence set for $\epsilon(G)$. To simplify our exposition, we construct a simultaneously valid \emph{lower bound} on the group-wise disparity, i.e., we define $\epsilon_{\text{lb}}(G)$ such that
\begin{align}
    \lim_{n \to \infty} \P(\epsilon_{\text{lb}}(G) \leq \epsilon(G)\,\text{for all $G \in \cG$}) = 1 - \alpha. \label{eqn:certify_criteria}
\end{align}
Though it may seem counterintuitive to lower bound a performance disparity, recall the conditional coverage example previewed in \Cref{sec:preview_certify}. Upper confidence bounds and intervals for $\epsilon(G)$ are obtained via a trivial modification described in \Cref{sec:certify_alt_nulls}. 

Naively, we might define $\epsilon_{\text{lb}}(G)$ using an upper bound on the maximum deviation between our performance disparity estimator and the true disparity, i.e., the $(1 - \alpha)$-quantile of $\sup_{G \in \mathcal{G}} \{\hat{\epsilon}(G) - \epsilon(G)\}$. Formally defining $\text{Quantile}(\alpha; X) \defeq \inf \{x : \alpha \leq \P(X \leq x) \}$, we could obtain a simultaneously valid lower bound on $\epsilon(G)$ via
\begin{align}
    \hat{\epsilon}(G) - \text{Quantile}\l(1 - \alpha; \sup_{G \in \cG} \{\hat{\epsilon}(G) - \epsilon(G)\} \r). \label{eqn:naive_est}
\end{align}

While it is straightforward to prove that \eqref{eqn:certify_criteria} holds for the proposed bound, there are two crucial problems. First, the quantile in \eqref{eqn:naive_est} cannot be estimated accurately: $\hat{\epsilon}(G)$ diverges for small groups, so when $\cG$ is large, $\hat{\epsilon}(G)$ does not converge uniformly to $\epsilon(G)$. As a consequence, standard asymptotic methods (e.g., bootstrap) will fail. Second, even if we could estimate this quantile, we would obtain a constant correction to the naive estimator for all groups. Since any such correction must be large to achieve simultaneously validity over small groups, this approach would lead to impractically conservative lower bounds.

To circumvent the first of these obstacles, we show that it is possible to consistently estimate the distribution of \begin{align}
\sup_{G \in \cG} \{\P(G) \cdot \P_n(G) \cdot (\hat{\epsilon}(G) - \epsilon(G)) \}. \label{eqn:bootstrap_process}
\end{align}
Intuitively, multiplying the naive process by $\P(G) \cdot \P_n(G)$ stabilizes its value for small groups. We then apply the bootstrap \citep{efron1979bootstrap, gine1990bootstrapping} to estimate the $(1 - \alpha)$-quantile of this process. Rigorously establishing bootstrap consistency requires a technical argument; see the proof of \Cref{thm:uniform_ci} for a detailed exposition. 

Mimicking our initial approach, we use an estimate of the $(1 - \alpha)$-quantile of \eqref{eqn:bootstrap_process} to construct a simultaneously valid lower bound on the true disparity. For all $G \in \cG$, we define $\epsilon_{\text{lb}}(G)$ such that $\P_n(G)^2 \cdot (\hat{\epsilon}(G) - \epsilon_{\text{lb}}(G))$ equals the $(1 - \alpha)$-quantile of \eqref{eqn:bootstrap_process}. Letting $t^*$ denote the bootstrap estimate of this quantile, i.e., the output of \textsf{\Cref{algo:bootstrap_ci}}, we obtain a simplified definition of $\epsilon_{\text{lb}}(G)$:
\begin{align}
    \epsilon_{\text{lb}}(G) \defeq \hat{\epsilon}(G) - \frac{t^*}{\P_n(G)^2}. \label{eqn:bound_closed_form}
\end{align} 
Given a valid estimate of $t^*$, the simultaneous validity of $\epsilon_{\text{lb}}(G)$ is a straightforward implication of our definition:
\begin{align*}
    \lim_{n \to \infty} \P(\epsilon_{\text{lb}}(G) \leq \epsilon(G)\,\text{for all $G \in \cG$}) &= \lim_{n \to \infty} \P(\hat{\epsilon}(G) - t^*/\P_n(G)^2 \leq \epsilon(G)\,\text{for all $G \in \cG$}) \nonumber \\
    &= \lim_{n \to \infty} \P\l (\sup_{G \in \cG} \P_n(G)^2 \cdot (\hat{\epsilon}(G) - \epsilon(G)) \leq t^* \r).
\end{align*}
The first equality follows from our definition of $\epsilon_{\text{lb}}(G)$, and the second follows from rearrangement. Replacing $\P_n(G)$ with $\P(G)$ (by Slutsky's lemma) and applying the definition of $t^*$ then completes our argument for simultaneous validity:
\begin{multline*}
    \lim_{n \to \infty} \P \l (\sup_{G \in \cG} \P_n(G)^2 \cdot (\hat{\epsilon}(G) - \epsilon(G)) \leq t^* \r ) \\ = \lim_{n \to \infty} \P\l (\sup_{G \in \cG} \P(G) \cdot \P_n(G) \cdot (\hat{\epsilon}(G) - \epsilon(G)) \leq t^* \r ) = 1 - \alpha.
\end{multline*}

\begin{algorithm}[t]
  \caption{Bootstrapping the lower confidence bound critical value}
  \label{algo:bootstrap_ci}
  \begin{algorithmic}[1]
    \State \textbf{Input:} Subpopulations $\cG$, audit trail $\cD$, level $\alpha$, number of bootstrap samples $B$
    \For{$b = 1,\dots,B$}
        \State Let $\cD^*_b$ be a sample with replacement of size $n$ from $\cD$;
        \State Define $\epsilon^*_b(G) \defeq \frac{1}{|G|^*} \sum_{(x^*_i, y^*_i) \in \cD^*_b} L_i^* - \hat{\theta}(\cD^*)$;
        \State Define $\P^*_b(G) \defeq \frac{1}{n} \sum_{(x^*_i, y^*_i) \in \cD^*_b} \indic{(x^*_i, y^*_i) \in G}$;
        \State $t^{(b)} = \max_{G \in \cG} \l \{\P_n(G) \cdot \P^*_b(G) \cdot (\epsilon^*_b(G) - \hat{\epsilon}(G)) \r\}$;
    \EndFor
    \State \textbf{Return:} $t^* = \text{Quantile} (1 - \alpha; \{t^{(b)} \}_{b = 1}^B )$
  \end{algorithmic}
\end{algorithm}

\begin{figure}
    \centering
    \includegraphics[scale=0.4]{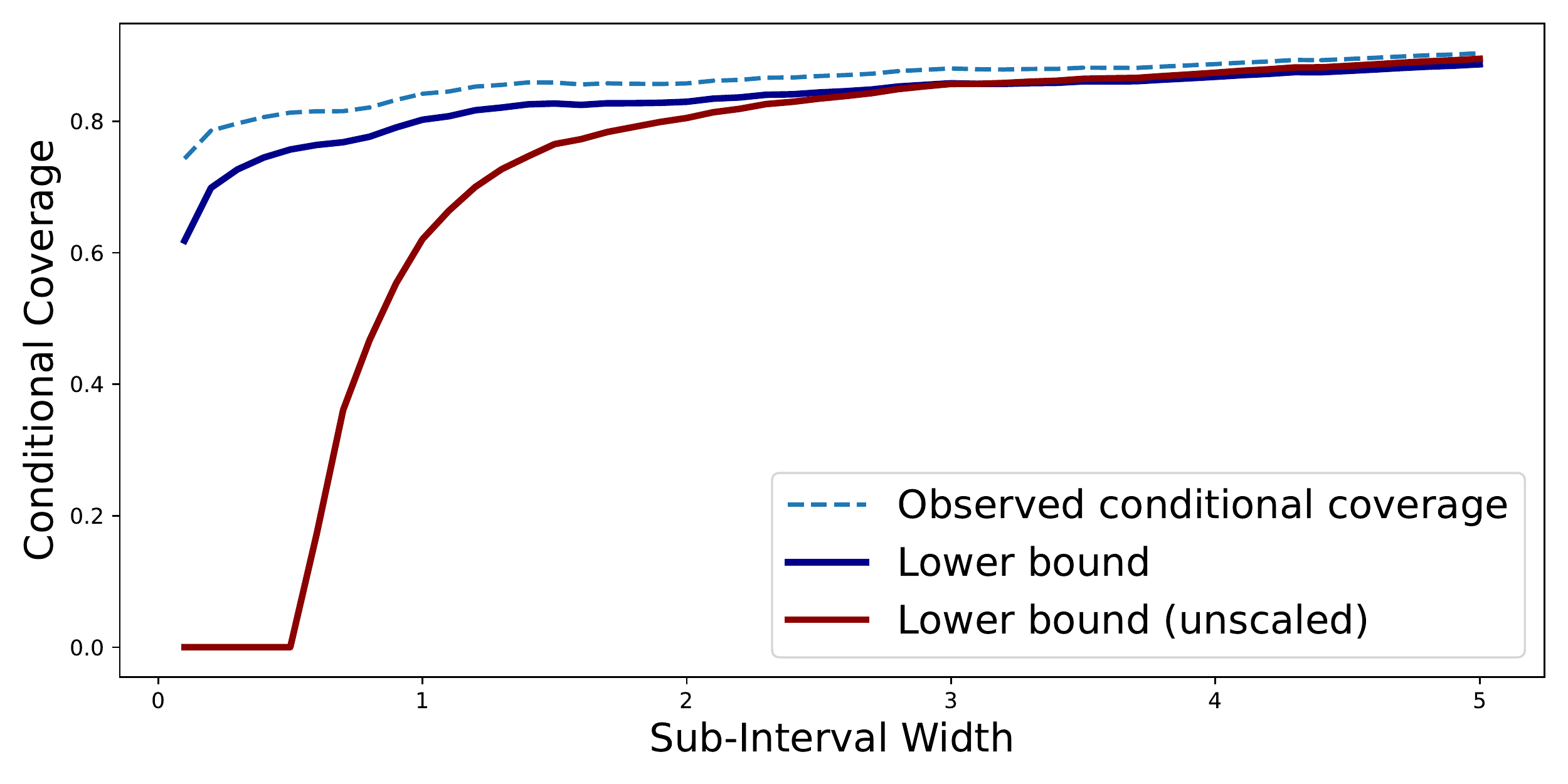}
    \caption{We reproduce the plot from \Cref{fig:cqr_bound}. Here the red curve denotes the lower bound on conditional coverage obtained via \eqref{eqn:bound_closed_form}. For smaller groups, it is substantially looser than the blue curve, i.e., the bound in \Cref{fig:cqr_bound}. The blue curve is obtained by rescaling the bootstrap process using $p_* = 0.01$ and $w_0 = \infty$.}
    \label{fig:cqr_bound_bad}
\end{figure}

While this bound satisfies the validity condition given by \eqref{eqn:certify_criteria}, directly applying \eqref{eqn:bound_closed_form} leads to practically unusable lower bounds on $\epsilon(G)$ for small groups. This is caused by a suboptimal dependence on group size in \eqref{eqn:bound_closed_form}. Asymptotically, our bound converges to 
\begin{align*}
    \epsilon_{\text{lb}}(G) = \hat{\epsilon}(G) - C_0 \cdot \left(\sqrt{\frac{1}{|G| \cdot {\color{red}\P_n(G)^3}} }\right),
\end{align*}
where $C_0$ is some group-independent constant and $|G|$ denotes the number of observed data points belonging to $G$. Even though the bound correction converges to $0$ at the expected $1/\sqrt{|G|}$ rate, the group-dependent factor of $\P_n(G)^3$ catastrophically inflates our confidence set for even moderately sized groups. \Cref{fig:cqr_bound_bad} reproduces \Cref{fig:cqr_bound} using \eqref{eqn:bound_closed_form} and shows that this definition produces vacuous lower bounds for all but the largest sub-intervals. 

To motivate our revised approach, we remark that the classical Wald confidence bound for $\epsilon(G)$ has the following asymptotic behavior,
\begin{align}
    \epsilon^\text{ideal}_{\text{lb}}(G) = \hat{\epsilon}(G) - C_1 \cdot \frac{\sigma_G}{\sqrt{|G|}}; \label{eqn:wald_bound}
\end{align}
$C_1$ is group-independent and $\sigma_G$ denotes the asymptotic standard deviation of $\sqrt{|G|} (\hat{\epsilon}(G) - \epsilon(G))$.

We can construct Wald-type bounds for all sufficiently large groups by rescaling the bootstrapped process. To understand how to define this scaling factor, we first study how rescaling the bootstrap process affects the resulting confidence bound. Let $s(G)$ and $\hat{s}(G)$ denote the population value and  estimator of the scaling factor, respectively. Then, instead of bootstrapping the $(1 - \alpha)$-quantile of $\sup_{G \in \cG} \{\P(G) \cdot \P_n(G) \cdot (\hat{\epsilon}(G) - \epsilon(G))\}$, we estimate the $(1 - \alpha)$-quantile of 
\begin{align}
    \sup_{G \in \cG} \left \{\frac{1}{s(G)} \cdot \P(G) \cdot \P_n(G) \cdot (\hat{\epsilon}(G) - \epsilon(G)) \right \}. \label{eqn:rescaled_process}
\end{align} 
Letting $t^*$ denote the $(1 - \alpha)$-quantile of \eqref{eqn:rescaled_process}, we mimic the previous derivation and obtain a closed-form solution for $\epsilon_{\text{lb}}(G)$ in terms of $\hat{\epsilon}(G)$ and $\hat{s}(G)$,
\begin{align}
    \epsilon_{\text{lb}}(G) \defeq \hat{\epsilon}(G) - t^* \cdot \frac{\hat{s}(G)}{\P_n(G)^2}. \label{eqn:scaled_closed_form}
\end{align}
Since $t^* = O_P(1/\sqrt{n})$, we can match the Wald-type bound in \eqref{eqn:wald_bound} by choosing $\hat{s}(G)$ to estimate $\P(G)^{3/2} \cdot \sigma_G$. 

Naively rescaling the process reproduces the divergence for small groups problem we encountered in our first approach. In particular, bootstrap consistency requires $\{1/\hat{s}(G)\}_{G \in \cG}$ to be a \emph{uniformly consistent} estimator of $\{1/s(G)\}_{G \in \cG}$. To this end, when $\P_n(G)$ is below some threshold $p_*$, we forgo estimating $\P(G)^{3/2} \cdot \sigma_G$. Instead, we scale by the more naive estimator $p_*^{3/2} \cdot \sqrt{\widehat{\text{Var}}(L)}$ where $\widehat{\text{Var}}(L)$ denotes the sample variance. 

To interpolate between these two group-size regimes, we set
\begin{align}
\hat{s}(G) = \max \{ \P_n(G), p_* \}^{3/2} \cdot \left( \frac{\P_n(G)}{\P_n(G) + w_0} \cdot \hat{\sigma}_G + \frac{w_0}{\P_n(G) + w_0} \cdot \widehat{\var}(L)^{1/2} \right), \label{eqn:rescaling_est}
\end{align}
where $w_0 > 0$ is a user-specified hyperparameter that controls the degree of shrinkage and $\hat{\sigma}_G$
is some point-wise, but not necessarily uniformly consistent, estimator of $\sigma_G$. By shrinking our estimate of the asymptotic variance of $\hat{\epsilon}(G)$ to a group-independent quantity, $\hat{s}(G)$ obtains the desired uniform consistency. In \textsf{\Cref{algo:bootstrap_ci_rescaled}}, we show how to define $\hat{\sigma}_G$ and compute the $t^*$ used in \eqref{eqn:scaled_closed_form}.

\begin{algorithm}[t]
  \caption{Bootstrapping the (rescaled) lower confidence bound critical value}
  \label{algo:bootstrap_ci_rescaled}
  \begin{algorithmic}[1]
    \State \textbf{Input:} Subpopulations $\cG$, audit trail $\cD$, level $\alpha$, threshold $p_*$, weight $w_0$, number of bootstrap samples $B$
    \For{$b = 1,\dots,B$}
        \State Let $\cD^*_b$ be a sample with replacement of size $n$ from $\cD$;
        \State Define $\P^*_b(G) \defeq \frac{1}{n} \sum_{(x^*_i, y^*_i) \in \cD^*_b} \indic{(x^*_i, y^*_i) \in G}$;
        \State Define $\epsilon^*_b(G) \defeq \frac{1}{\P^*_b(G) \cdot n} \sum_{(x^*_i, y^*_i) \in G} L^*_i - \hat{\theta}(\cD^*_b)$;
    \EndFor
    \State Define the asymptotic variance estimator by \begin{align*}
        \hat{\sigma}^2_G &\defeq \widehat{\text{Var}}(L \mid G) + \P_n(G) \l (\widehat{\text{Var}}(\psi) - 2 \cdot \widehat{\text{Cov}}(\,L\,, \psi \mid G) \r )
    \end{align*}
    where $\widehat{\text{Var}}(\cdot)$ and $\widehat{\text{Cov}}(\cdot)$ correspond to the sample (conditional) variance and covariance;
    \State Define $\hat{s}(G)$ by \eqref{eqn:rescaling_est};
    \For{$b = 1,\dots,B$}
        \State $t^{(b)} = \max_{G \in \cG} \{\frac{1}{\hat{s}(G)} \cdot \P_n(G) \cdot \P^*_b(G) \cdot (\epsilon^*_b(G) - \hat{\epsilon}(G)) \}$;
    \EndFor

    \State \textbf{Return:} $t^* = \text{Quantile} (1 - \alpha; \{t^{(b)} \}_{b = 1}^B )$
  \end{algorithmic}
\end{algorithm}

In practice, we observe that estimating the asymptotic variance of $\hat{\epsilon}(G)$ can harm the finite-sample validity of $\epsilon_{\text{lb}}(G)$. We thus recommend setting $w_0 = \infty$ unless adaptivity to the  variance of $\hat{\epsilon}(G)$ is deemed critical. 

Using the output of \textsf{\Cref{algo:bootstrap_ci_rescaled}} in \eqref{eqn:scaled_closed_form}, we obtain more practical confidence bounds. We produce the blue curve in \Cref{fig:cqr_bound_bad} using $p_* = 0.01$ and $w_0 = \infty$. There is no free lunch: observe that the red curve (unscaled) yields a tighter lower confidence bound for the largest sub-interval widths. Nevertheless, it is clear that the rescaled process produces a usable lower bound over a much wider range of group sizes. 

\Cref{thm:uniform_ci} states sufficient conditions for bootstrap consistency and, therefore, simultaneous validity of the lower bounds defined in \eqref{eqn:bound_closed_form} or \eqref{eqn:scaled_closed_form}. 
\begin{theorem} [Simultaneous validity] \label{thm:uniform_ci}
Assume that $L$ is bounded and that $L - \hat{\theta}$ is non-constant over at least one non-empty group. Further assume that $\cG$ has finite Vapnik-Chernovenkis (VC) dimension. Then, $\epsilon_{\textup{lb}}(G)$ is an asymptotic $(1 - \alpha)$-lower confidence bound for $\epsilon(G)$ that is simultaneously valid for all $G \in \cG$, i.e., 
\begin{align*}
    \lim_{n \to \infty} \P \l(\epsilon_{\textup{lb}}(G) \leq \epsilon(G) \textup{ for all $G \in \cG$} \r) = 1 - \alpha.
\end{align*}
\end{theorem}

We make two remarks regarding our assumptions. First, our restriction that $\cG$ has finite VC dimension is satisfied by most interpretable collections of groups: intervals, rectangles, halfspaces, etc. Second, in typical fairness applications, $L$ is $\{0,1\}$-valued and satisfies the boundedness assumption; for unbounded metrics, the auditor might truncate $L$ or, if appropriate, assume compactness of the domain. Furthermore, if $\cG$ is a finite collection, we may relax this assumption to $\text{Var}(L) < \infty$.

Even though \Cref{thm:uniform_ci} is an asymptotic result, a finite-sample approach, e.g., via empirical process concentration, can only satisfy a conservative coverage guarantee. By contrast, the simultaneous coverage of $\epsilon_{\text{lb}}(\cdot)$ converges to \emph{exactly} $1 - \alpha$ under any data-generating distribution. Even at small sample sizes, we show in \Cref{sec:certify_empirics} that the gap between nominal and realized coverage is minimal.

\subsubsection{Boolean certification}
Next, we consider issuing a Boolean certificate for $G$ if $\epsilon(G)$ lies above some pre-specified tolerance $\epsilon$. The trivial extension of our methods to certifying $\epsilon(G) < \epsilon$ and $|\epsilon(G)| < \epsilon$ is described in \Cref{sec:certify_alt_nulls}. While this approach returns strictly less information to the auditor compared to the confidence bound certificate, it offers computational benefits and can be more powerful in certain settings.

Formally, certifying $\epsilon(G) > \epsilon$ corresponds to testing the null hypothesis $\bar{H}_0(G) : \epsilon(G) \leq \epsilon$.
We require that, in large samples, all issued certificates are simultaneously valid with probability $1 - \alpha$. Equivalently,
\begin{align*}
    \lim_{n \to \infty} \P\left( \text{there exists any falsely certified $G \in \cG$} \right) \leq \alpha. 
\end{align*}
In the language of multiple testing, this desideratum is (asymptotic) strong family-wise error rate (FWER) control. 

To construct such a test, we show that the bootstrap can conservatively estimate the $(1 - \alpha)$-quantile of $\sup_{G \in \cG} \{\P_n(G) \cdot (\hat{\epsilon}(G) - \epsilon)\}$. Letting $t^*$ denote the estimate output by \textsf{\Cref{algo:bootstrap_fwer}}, we certify that $\epsilon(G) > \epsilon$ if $\P_n(G) \cdot (\hat{\epsilon}(G) - \epsilon) \geq t^*$. Simplifying, we reject the null when
\begin{align}
   \hat{\epsilon}(G) \geq \epsilon + \frac{t^*}{\P_n(G)}. \label{eqn:fwer_rejection}
\end{align}
For example, if $L$ is the coverage indicator, we certify $G$ if the empirical conditional coverage on $G$ is sufficiently high.

Even though the scaling of the rejection threshold in \eqref{eqn:fwer_rejection} is sub-optimal for small groups, correcting this would eliminate the singular advantage of the Boolean certification procedure: efficient optimization over certain infinite group collections, e.g., intervals, slabs. Observe that
\begin{align*}
    \P_n(G) \cdot (\hat{\epsilon}(G) - \epsilon) = \frac{1}{n} \sum_{i = 1}^n (L_i - \hat{\theta} - \epsilon) \indic{(x_i, y_i) \in G}, 
\end{align*}
i.e., the process is linear in the group-indicator. Rescaling by $\hat{s}(G)$ does away with this linearity, and line 4 of \textsf{\Cref{algo:bootstrap_fwer}} then requires brute-force search over all $G \in \cG$. For the sake of completeness, we describe this rescaled variant of the Boolean certification procedure in \textsf{\Cref{algo:bootstrap_fwer_appendix}}. 

\begin{algorithm}[t]
  \caption{Bootstrapping the Boolean certificate critical value}
  \label{algo:bootstrap_fwer}
  \begin{algorithmic}[1]
    \State \textbf{Input:} Subpopulations $\cG$, disparity $\epsilon$, audit trail $\cD$, level $\alpha$, number of bootstrap samples $B$
    \For{$b = 1,\dots,B$}
        \State Let $\cD^*_b $ be a sample with replacement of size $n$ from $\cD$;
        \State $t^{(b)} = \max_{G \in \cG} \{\P^*_b(G) \cdot (\epsilon^*_b(G) - \epsilon) - \P_n(G) \cdot (\hat{\epsilon}(G) - \epsilon)\}$;
    \EndFor
    \State \textbf{Return:} $t^* = \text{Quantile} (1 - \alpha; \{ t^{(b)} \}_{b = 1}^B )$
  \end{algorithmic}
\end{algorithm}

\Cref{thm:fwer_main_text} states that \eqref{eqn:fwer_rejection} produces valid certificates under the mild assumptions of \Cref{thm:uniform_ci}.
\begin{theorem}[FWER control for certification] \label{thm:fwer_main_text}
Retain the assumptions of \Cref{thm:uniform_ci}. Then, the certificates issued by \eqref{eqn:fwer_rejection} satisfy
\begin{align*}
\lim_{n \to \infty} \P\left( \textup{there exists any falsely certified $G \in \cG$} \right) \leq \alpha.
\end{align*}
\end{theorem}

Our bootstrap procedure accounts for overlap between subpopulations and improves upon more naive FWER-controlling procedures such as the Bonferroni test. However, the asymptotic FWER is only exactly $\alpha$ when $\epsilon(G) = \epsilon$ for all $G$. This compares unfavorably to the simultaneous confidence set guarantee, which promises exact Type I error control under any data-generating distribution. We compare and contrast the finite-sample validity of the Boolean and confidence set-based certification methods in the sequel.

\subsection{Empirical results} \label{sec:certify_empirics}
\subsubsection{Synthetic validation}
First, we verify that the (asymptotic) claims of \Cref{thm:uniform_ci} and \Cref{thm:fwer_main_text} are accurate in finite samples. We consider three synthetic data experiments; the results from each are presented in \Cref{tab:fwer_validation}. 

We initially consider a homoskedastic linear model. We sample $(X_i, Y_i)$ from
\begin{align}
    X_i \simiid \text{Unif}(0, 1), \quad Y_i \simiid \mathcal{N}(\beta_0 X_i, 1). \label{eqn:homoskedastic_dgp}
\end{align}
We then obtain $f(x) = \hat{\beta}x$ via ordinary least-squares on $1000$ training points sampled from this distribution. The performance metric of interest is squared-error loss, i.e., $L(f(X), Y) = (Y - f(X))^2$ and the target $\theta_P = 0$. Using held-out data sets of varying size, we issue Boolean certificates for sub-intervals $G \subseteq [0,1]$ over which $\epsilon(G) < 1$. 

In our second experiment, we validate our audit using a heteroskedastic linear model,
\begin{align}
    X_i \simiid \text{Unif}(0, 1), \quad Y_i \simiid \mathcal{N}(\beta_0 X_i, X_i). \label{eqn:heteroskedastic_dgp}
\end{align}
The model $f(X)$ and metric $L(f(X), Y)$ is obtained identically to the previous synthetic experiment. We then issue Boolean certificates for sub-intervals $G \subseteq [0,1]$ over which $\epsilon(G) < \epsilon$ for $\epsilon \in \{0.4, 0.5\}$. To simplify the verification of issued certificates, we consider $G$ with endpoints belonging to $\{0,0.1,\dots,1\}$. \Cref{fig:heteroskedastic_certificate} displays a trial experiment with an audit trail of sample size $400$. For this setting, we observe that the nominal error rate overestimates the realized probability of false certification. 

\begin{figure}
\centering
\captionsetup[subfigure]{width=0.9\linewidth}%
\begin{subfigure}{.5\textwidth}
  \centering
  \includegraphics[width=\linewidth]{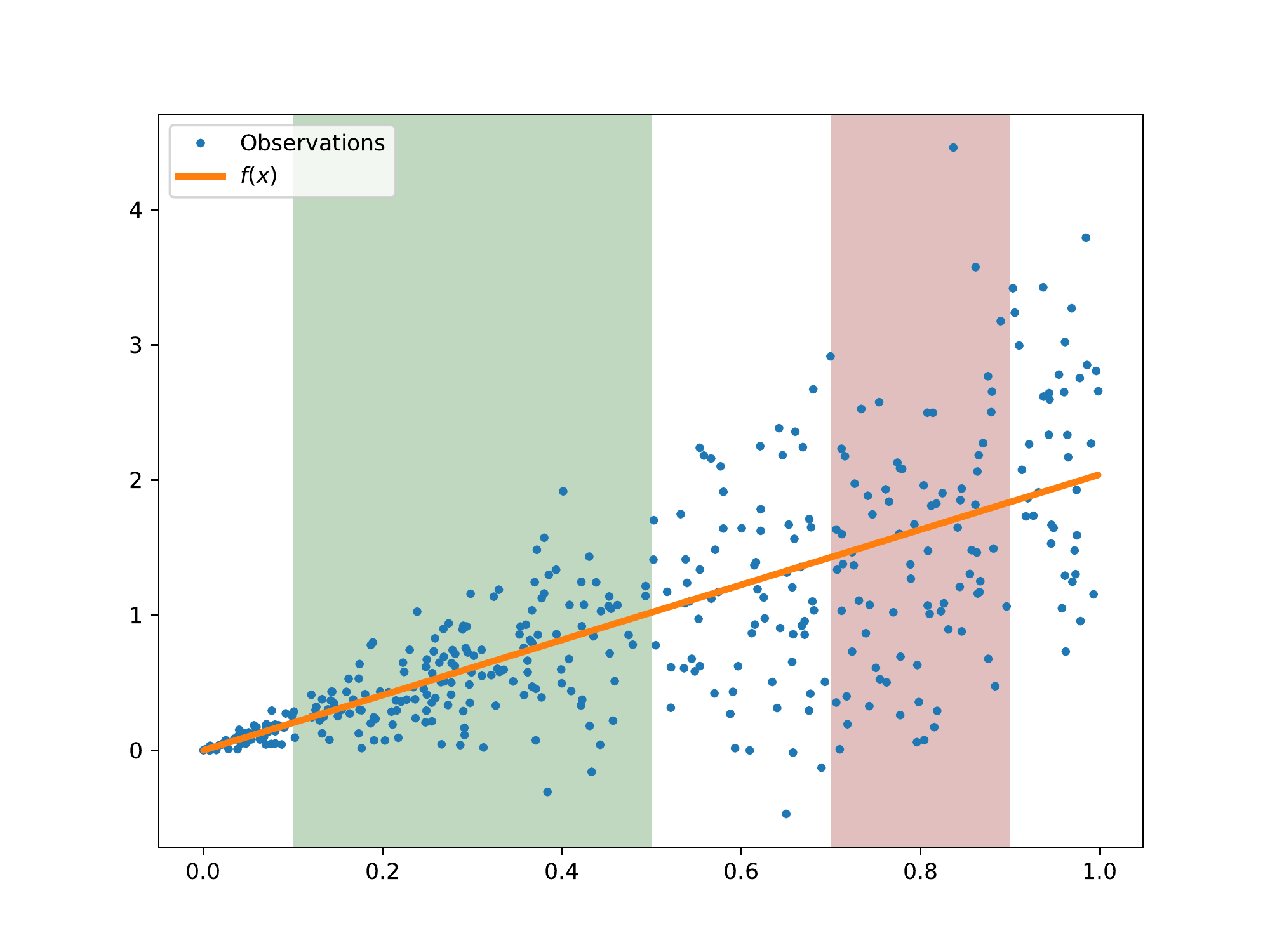}
  \caption{}
  \label{fig:heteroskedastic_certificate}
\end{subfigure}%
\begin{subfigure}{.5\textwidth}
  \centering
  \includegraphics[width=\linewidth]{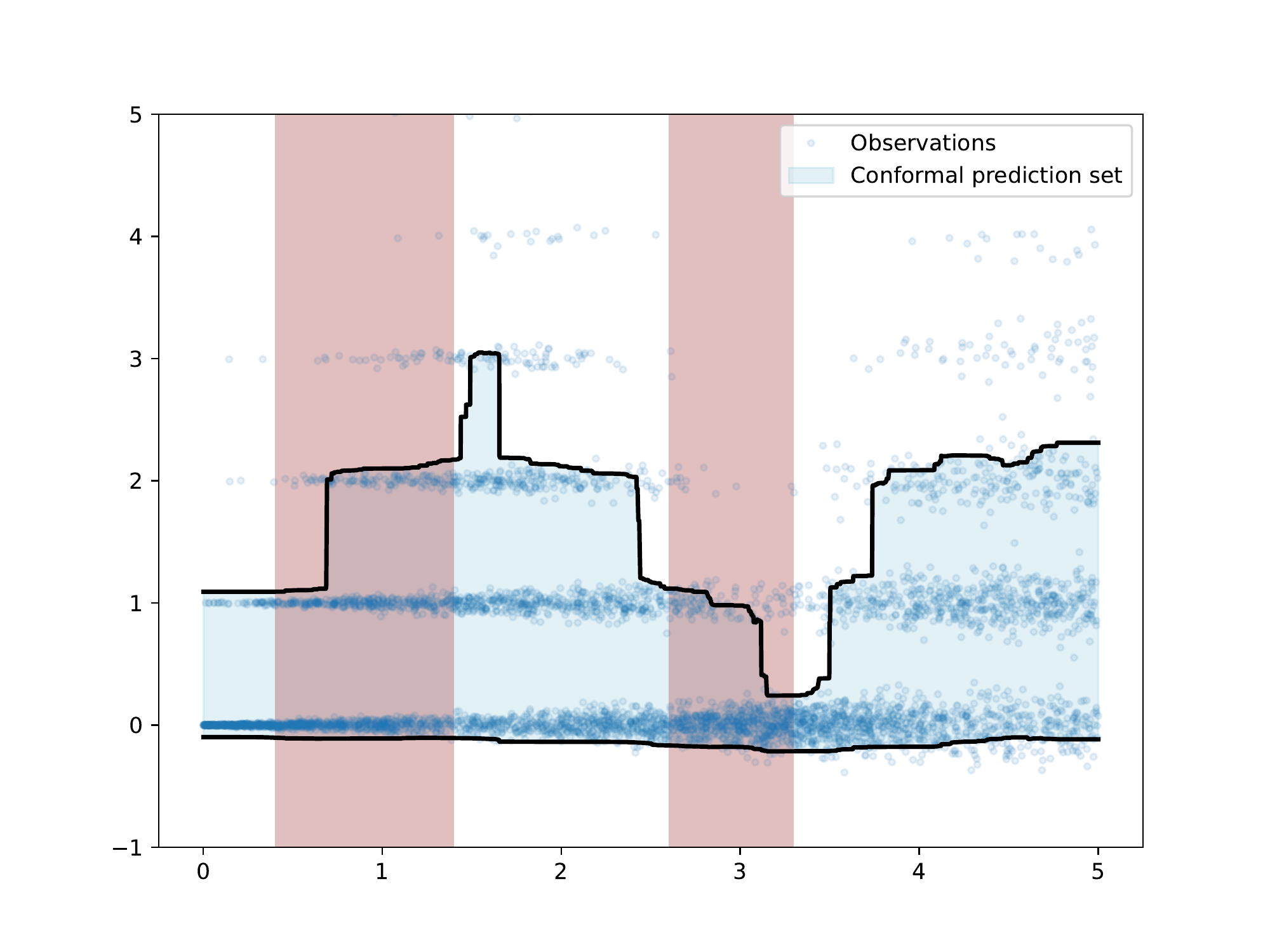}
  \caption{}
  \label{fig:cqr_certificate}
\end{subfigure}
\caption{\Cref{fig:heteroskedastic_certificate} displays a linear model over $400$ hold-out points generated from \eqref{eqn:heteroskedastic_dgp}. Constructing a $90\%$-confidence bound for the expected MSE yields upper bounds of $0.5$ (true MSE: $0.3$) and $1.15$ (true MSE: $0.80$) for the green and red sub-intervals, respectively. \Cref{fig:cqr_certificate} displays prediction intervals from a conformalized quantile random forest on $400$ hold-out points from the synthetic data-generating process in \citet{romano2019conformalized}. Constructing a $90\%$-confidence bound for the conditional coverage yields lower bounds of $81.1\%$ (true coverage: $92.4\%$) and $80.1\%$ (true coverage: $89.5\%$) for the two red sub-intervals, respectively. These bounds are simultaneously valid over arbitrarily many such queries.}
\label{fig:fwer_certificates}
\end{figure}

Last, we consider the synthetic dataset introduced by \cite{romano2019conformalized} and displayed in \Cref{fig:cqr_certificate}. Here, the conformalized quantile regression (CQR) is guaranteed to have 90\% marginal coverage, so we consider certifying sub-intervals for which the coverage exceeds 90\%, e.g. $\theta_P = 0.9$ and $L = \mathbf{1}\{Y \in \hat{C}(X)\}$. 

\Cref{tab:fwer_validation} displays the FWER of the certification audit over $200$ trials. We set the nominal error rate to $0.1$ and vary the sample size $n$ over $(100, 200, 400, 800, 1600)$. The results corroborate the predictions made by our theory. For large $n$, the nominal FWER matches or exceeds the realized level. When the null hypothesis does not hold at the boundary, i.e., $\epsilon(G)$ is not approximately equal to the threshold $\epsilon$, the nominal FWER can substantially overestimate the true level. \Cref{tab:fwer_validation} also displays the power, i.e., the proportion of sub-intervals achieving the certification threshold that are actually certified by our method.

By contrast, recall that the confidence bound approach to certification guarantees (asymptotically) exact coverage under any data-generating distribution. To verify this claim, we construct simultaneous $90\%$ confidence bounds for each of the aforementioned synthetic experiments. 

\Cref{tab:coverage_validation} summarizes our results: we observe that \textsf{\Cref{algo:bootstrap_ci}} and its rescaled variant \textsf{\Cref{algo:bootstrap_ci_rescaled}} $(p_* = 0.01, w_0 = \infty)$ obtain the nominal coverage in large samples. For a fixed threshold $\epsilon$, we define the power of a confidence set as the proportion of sub-intervals with true error rates below $\epsilon$ for which the confidence set excludes $\epsilon$. Note that while the unscaled confidence bounds have low power, i.e., they are less likely to exclude the targeted error level, rescaling mitigates this issue.

\begin{table}
	\centering
	\begin{tabular}{llllllll}
	\toprule
	&&\multicolumn{5}{c}{Sample size ($n$)}\\
	\cmidrule{3-7}
	&&100&200&400&800&1600\\
	\midrule
	\multirow{1}{*}{Model \eqref{eqn:homoskedastic_dgp} ($\epsilon = 1$)}&FWER&0.145&0.145&0.17&0.125&0.095\\[0.25in]

\multirow{2}{*}{Model \eqref{eqn:heteroskedastic_dgp} ($\epsilon=0.5$)}
&FWER&0.105&0.105&0.1&0.075&0.065\\
&Power&0.145&0.268&0.411&0.562&0.709\\
[0.25in]

\multirow{2}{*}{Model \eqref{eqn:heteroskedastic_dgp} ($\epsilon=0.4$)}
&FWER&0.045&0.04&0.015&0.025&0.02\\
&Power&0.036&0.079&0.21&0.375&0.57\\
[0.25in]

\multirow{2}{*}{CQR ($\epsilon = 0.9$)}&FWER&0.081&0.034&0.019&0.012&0.003\\
&Power&0.019&0.008&0.009&0.02&0.062\\
\bottomrule
	\end{tabular}
	\caption{FWER and power of certificates issued by \textsf{\Cref{algo:bootstrap_fwer}} with $B = 500$ and $\alpha = 0.1$. All results are based on 200 trials. We see that, for large $n$, the simultaneous validity guarantee holds.}
	\label{tab:fwer_validation}
\end{table}

\begin{table}
	\centering
	\begin{tabular}{llllllll}
	\toprule
	&&\multicolumn{5}{c}{Sample size ($n$)}\\
	\cmidrule{3-7}
	&&100&200&400&800&1600\\
	\midrule
	\multirow{3}{*}{Model \eqref{eqn:heteroskedastic_dgp} }&Coverage&0.84&0.83&0.86&0.895&0.905\\
&Power ($\epsilon = 0.5$)&0.043&0.081&0.15&0.256&0.378\\
&Power ($\epsilon = 0.4$)&0.003&0.005&0.012&0.058&0.187\\[0.25in]

	\multirow{3}{*}{Model \eqref{eqn:heteroskedastic_dgp} (rescaled) }&Coverage&0.84&0.81&0.845&0.885&0.88\\
&Power ($\epsilon = 0.5$)&0.149&0.270&0.427&0.608&0.743\\
&Power ($\epsilon = 0.4$)&0.038&0.091&0.24&0.457&0.633\\[0.25in]

\multirow{2}{*}{CQR (rescaled)}&Coverage&0.815&0.87&0.845&0.865&0.91\\
&Power ($\epsilon = 0.9$)&0.025&0.011&0.007&0.017&0.055\\
\bottomrule
	\end{tabular}
	\caption{Simultaneous coverage of confidence bounds issued by \textsf{\Cref{algo:bootstrap_ci}} and \textsf{\Cref{algo:bootstrap_ci_rescaled}} with $B = 500$, $\alpha = 0.1$, and $p_* = 0.01$. All results are based on $200$ trials. We see that, for large $n$, the certification procedure satisfies the simultaneous validity guarantee.}
	\label{tab:coverage_validation}
\end{table}

\subsubsection{Certifying COMPAS}
Next, we reconsider previous analyses of the COMPAS recidivism prediction instrument (RPI) and show how our methods can be used to establish rigorous guarantees. The COMPAS algorithm assigns defendants risk scores ranging from 1 to 10 based on an estimated likelihood of re-offending. Prior work showed that African-American defendants are more likely to be mis-classified as high-risk when compared to Caucasian defendants \citep{angwin2016machine}. In response to this finding, the creators of COMPAS, Northpointe Inc., argued that the algorithm is fair when evaluated by the predictive parity criterion \citep{dieterich2016compas, flores2016false}. While they provide statistical evidence for the absence of a significant racial bias by this measure, our methods allow for the construction of an explicit bound on the true disparity.

The predictive parity criterion is satisfied for a single group, $G$, when the positive predictive value (PPV) of $f$ for $G$ matches the PPV for the complement of $G$, i.e.,
\begin{align*}
    \P(Y = 1 \mid f(X) = 1, X \in G) = \P(Y = 1 \mid f(X) = 1, X \in G^c).
\end{align*}
Intuitively, the PPV measures how informative a positive prediction is. For example, if COMPAS classifies a defendant as high-risk ($f(X) = 1$), the PPV corresponds to the probability that they actually recidivate ($Y = 1$).

Following prior work, we binarize the COMPAS scores by defining $f(X)$ to be $1$ when the RPI score is $5$ or higher. Then, to certify a lower bound on the gap between an African-American and Caucasian defendant's PPVs, we consider the subset ($n = 2525$) of the audit trail with $f(X) = 1$ and $X_{\text{race}} \in \{\text{African-American}, \text{Caucasian}\}$. We instantiate our audit with $L$ corresponding to the indicator that $Y$ matches $f(X)$, $\cG$ containing just one group (African-American defendants), and $\theta_P$ denoting the PPV for White defendants:
\begin{align*}
    L(f(X), Y) = Y,\quad &\cG = \{\{(X, Y) \mid X_{\text{race}} = \text{African-American}, f(X) = 1\} \},\\
    \theta_P &= \E[Y = 1 \mid X_{\text{race}} = \text{Caucasian}].
\end{align*} 
We estimate $\theta_P$ using the empirical conditional expectation, $\hat{\theta} = \hat{\E}_n[Y = 1 \mid X_{\text{race}} = \text{Caucasian}]$. 

To verify the claim made by \cite{dieterich2016compas} that there is no reduction in PPV for African-American defendants relative to Caucasian defendants, we construct a $90\%$-lower confidence bound for $\epsilon(G) \defeq \P(Y = 1 \mid f(X) = 1, X \in G) - \P(Y = 1 \mid f(X) = 1, X \in G^c)$. Using this approach, we can certify the previous claim: among defendants receiving high-risk predictions, an African-American defendant is at least $1.87\%$ more likely to recidivate than a Caucasian defendant.

Our methodology is not essential to establishing this result since we only consider a single group. With these tools, however, we can establish PPV disparity bounds that hold simultaneously over every protected subgroup. Again using the COMPAS PPV on Caucasian defendants as our target, we construct simultaneously valid $90\%$ confidence intervals on the PPV disparity for every subpopulation formed by the intersection of race, sex, and age. In \Cref{fig:ppv_intervals}, we plot the intervals corresponding to subgroups of the African-American subpopulation and subgroups formed by intersections of sex and age alone. Our results validate Northpointe Inc.'s claims of PPV parity for several, albeit not all, African-American subpopulations. More generally, younger male subpopulations appear to have higher COMPAS PPV when compared to the Caucasian subpopulation.

\begin{figure}
\centering
\captionsetup[subfigure]{width=0.9\linewidth}%
\begin{subfigure}{.5\textwidth}
  \centering
  \includegraphics[width=\linewidth]{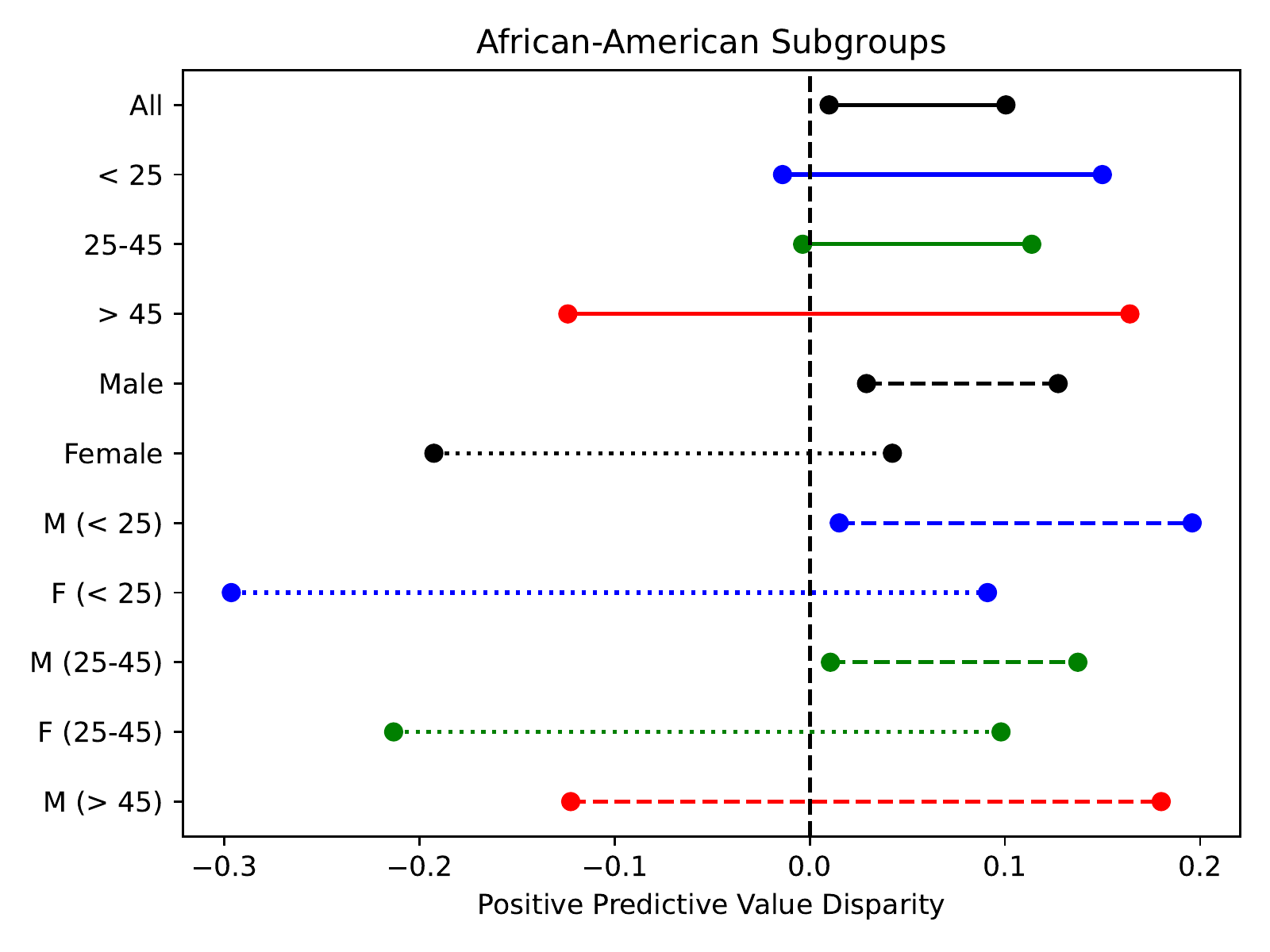}
  \caption{}
  \label{fig:ppv_intervals_afam}
\end{subfigure}%
\begin{subfigure}{.5\textwidth}
  \centering
  \includegraphics[width=\linewidth]{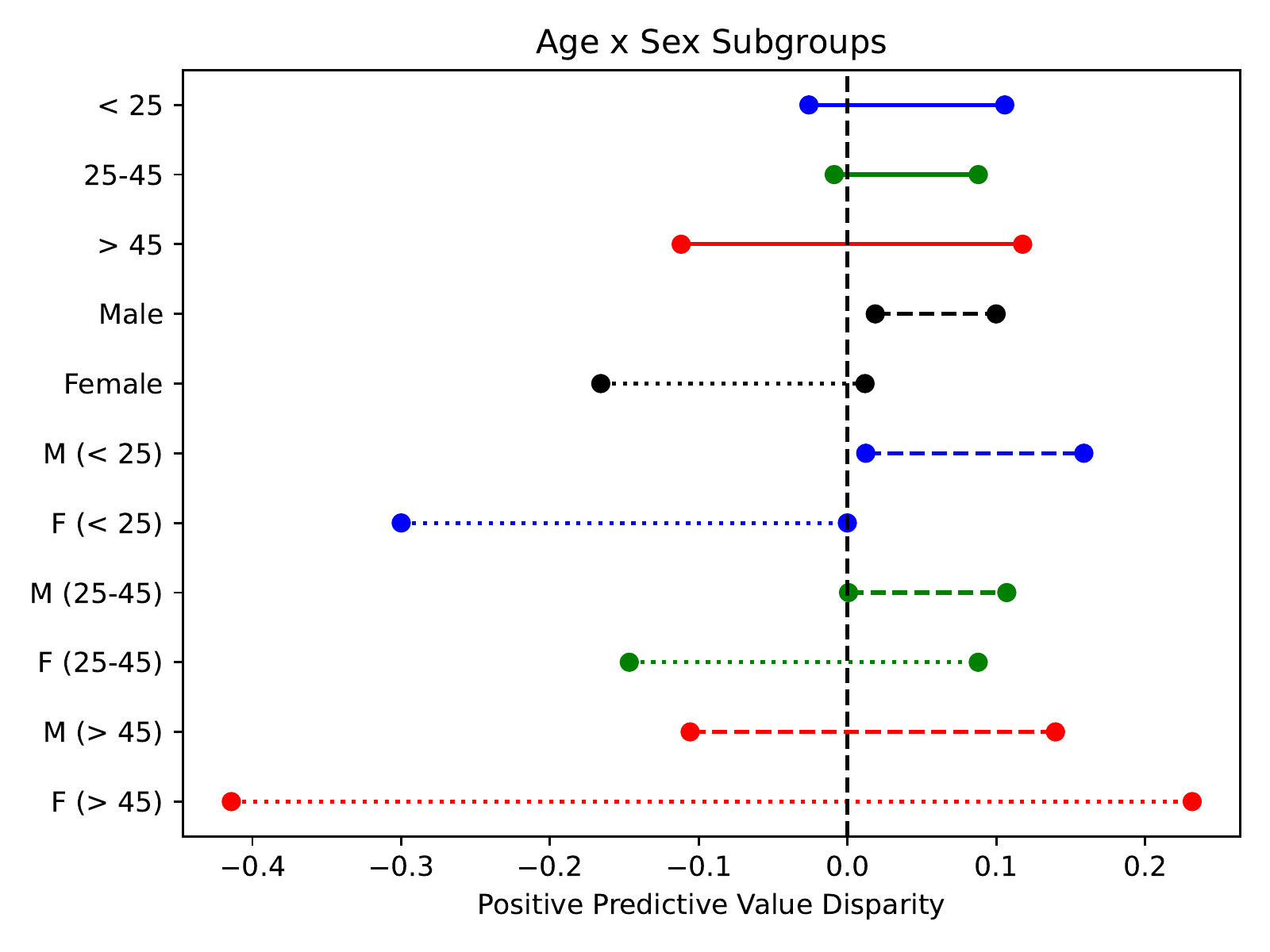}
  \caption{}
  \label{fig:ppv_intervals_all}
\end{subfigure}
\caption{In \Cref{fig:ppv_intervals_afam}, we plot 90\% confidence intervals for the difference in COMPAS PPV between each African-American subgroup and the entire Caucasian subgroup. In \Cref{fig:ppv_intervals_all}, we plot 90\% confidence intervals for the difference in COMPAS PPV between each subgroup formed by intersections of age and sex and the entire Caucasian subgroup.}
\label{fig:ppv_intervals}
\end{figure}

\section{Flagging performance disparities} \label{sec:flag_G}
\subsection{Methods}
We next consider the problem of \emph{flagging} subpopulations for which the disparity exceeds some tolerance, i.e., when
\begin{align}
    H_0(G) : \epsilon(G) \leq \epsilon. \label{eqn:flag_null}
\end{align}
fails to hold. False flags are less problematic than false certificates of performance, so a weaker notion of error control than simultaneous validity suffices. A-priori, we expect most groups to satisfy the null hypothesis given by \eqref{eqn:flag_null}. This leads us to consider a notion of error rate that is defined relative to the number of flags issued instead of the number of subgroups tested; in particular, we control the asymptotic false discovery rate (FDR). In our setting, controlling the false discovery rate translates to upper bounding the expected proportion of falsely flagged subpopulations by $\alpha$. For instance, in \Cref{sec:preview_flag}, we controlled the FDR at 10\%, implying that approximately 90\% of the flagged subgroups have a false positive rate at least 5\% higher than the population average.

Formalizing this criterion, we require
\begin{align*}
    \lim_{n \to \infty} \E \left[ \frac{\l|\{\text{falsely flagged $G \in \cG$}\} \r| }{\l|\{\text{flagged $G \in \cG$}\} \r| \vee 1} \right] \leq \alpha.
\end{align*}
We remark that simultaneous error control (familywise-error rate control) implies control of the false discovery rate, making the latter a strictly weaker criterion.

When the collection of subpopulations is finite, we propose to apply the Benjamini-Hochberg procedure (denoted by \textsf{BH}($\alpha$) in the sequel) to a collection of marginal p-values. \textsf{\Cref{algo:flag_p_values}} describes the computation of these p-values.
\begin{algorithm}[H]
  \caption{Constructing p-values for $G \in \cG$}
  \label{algo:flag_p_values}
  \begin{algorithmic}[1]
    \State \textbf{Input:} Subpopulations $\cG$, audit trail $\cD$, bootstrap samples $B$, tolerance $\epsilon$
    \For{$b = 1,\dots,B$}
        \State Let $\cD^*_{b}$ be a sample with replacement of size $n$ from $\cD$;
        \For{$G \in \cG$}
            \State $t^{(b)}(G) = \epsilon^*_b(G) - \hat{\epsilon}(G)$;
        \EndFor
    \EndFor
    \For{$G \in \cG$}
        \State $s^*(G) = \frac{1}{\Phi(3/4)} \cdot \text{Quantile} (0.5; \{ | t^{(b)}(G) | \}_{b = 1}^B )$;
        \State $p(G) = 1 - \Phi \l((\hat{\epsilon}(G) - \epsilon) /s^*(G) \r )$;
    \EndFor
    \State \textbf{Return:} $\{p(G)\}_{G \in \cG}$.
  \end{algorithmic}
\end{algorithm}
The following proposition states that the $\textsf{BH}(\alpha)$ procedure (applied to the output of \textsf{\Cref{algo:flag_p_values}}) controls the asymptotic false discovery rate under two cases of practical interest: first, the case of mutually disjoint groups, and second, the case of binary outcomes with arbitrarily overlapping group structure. The validity of the procedure in the first case is not surprising since the p-values for disjoint groups are independent. The proof of the second case is more subtle: we show that the binary metric implies a certain positive dependency\footnote{The positive dependency we identify is formally termed positive regression dependence on a subset (PRDS) \citep{benjamini2001control}.} among the p-values. Given this correlation structure, the \textsf{BH} procedure is known to be valid \citep{benjamini2001control}.
\begin{proposition}[FDR control]
Assume that $\P(G)$ and $\var \l(L \mid G \r)$ are bounded away from $0$ for all $G \in \cG$, $\theta_P$ is a-priori known, and that at least one of the following conditions holds:
\begin{enumerate}[label=(\roman*), itemsep=-0.5ex]
    \item $\{G\}_{G \in \cG}$ are mutually disjoint;
    \item $L$ takes values in $\{0,1\}$.
\end{enumerate}
If we flag the rejections of the $\textsf{BH}(\alpha)$ procedure on $\{p(G)\}_{G \in \cG}$, then the false discovery rate is asymptotically controlled at level $\alpha$.
\label{prop:fdr_control_bh}
\end{proposition}
We expect \Cref{prop:fdr_control_bh} to remain valid under violations of the stated assumptions. Even outside of the two cases stated in \Cref{prop:fdr_control_bh} (e.g., when we must estimate $\theta_P$), prior experiments with the \textsf{BH} algorithm and our own empirics suggest that this procedure will not violate FDR control \citep{fithian2022conditional}. 


If flagging with FDR control over an infinite collection of subpopulations is desired, we suggest a two-stage procedure. First, split the audit trail and use one split to discover a finite sub-collection of interpretable subpopulations, e.g., fit a regression tree and let each leaf define a subpopulation of interest. Then, use the other split to run the flagging procedure validated by \Cref{prop:fdr_control_bh}.

\subsection{Empirical results}
\subsubsection{Folktables} \label{sec:folktables}
We evaluate the flagging methodology on an income prediction dataset derived from the 2018 Census American Community Survey Public Use Microdata and made available in the Folktables package \citep{ding2021retiring, flood2015integrated}. Using the California data set filtered to over-16 individuals who earned at least $\$100$ in the past year, we aim to predict whether an individual's income exceeds \$50,000. We include age, place of birth, education, race, marital status, occupation, sex, race, and hours worked in the fitted prediction rule.

To validate the (asymptotic) FDR control result in \Cref{prop:fdr_control_bh}, we fit logistic and linear regression models to a training set of $1000$ data points, and then sample audit trails of varying size from the remaining data. We flag subpopulations for which: (1) the misclassification rate is higher than a fixed threshold of $0.5$, (2) the misclassification rate is higher than the population average error rate, (3) the mean-squared error (MSE) of the predicted income is higher than the population mean-squared error. Each of these tasks sheds light on the relevance of \Cref{prop:fdr_control_bh}. The first flagging task satisfies the assumptions of \Cref{prop:fdr_control_bh}, while the other two violate the stated assumptions. 

Since the audit trails are sampled with replacement, the data-generating distribution $P$ is the uniform distribution over the finite population of data held-out from model fitting. Therefore, a flag is falsely issued if the flagged subgroup's error rate on the entire held-out data set fails to exceed the stated threshold. \Cref{tab:fdr_validation} shows that over $1000$ trials, the estimated FDR for each task is well below the nominal bound of $0.1$ at every sample size tested. This is because the null p-values are, in practice, conservative, as the null hypothesis \eqref{eqn:flag_null} rarely holds with equality. 

\begin{table}
	\centering
	\begin{tabular}{lllllllll}
	\toprule
	&&\multicolumn{5}{c}{Sample size ($n$)}\\
	\cmidrule{3-8}
	&Task&100&200&400&800&1600&3200\\
	\midrule
	\multirow{3}{*}{Folktables}&$1$&0.045&0.038&0.032&0.025&0.021&0.014\\
&$2$&0.003&0.005&0.003&0.004&0.003&0.004\\
&$3$&0.0&0.0&0.0&0.0&0.0&0.002\\
\bottomrule
	\end{tabular}
	\caption{FDR of flags issued with $B = 500$ and $\alpha = 0.1$ for the tasks described in \Cref{sec:folktables}. All results are based on 1000 trials. We see that, for any $n$, the FDR guarantee is conservative.}
	\label{tab:fdr_validation}
\end{table}

\subsubsection{COMPAS}
Prior analysis of the COMPAS RPI has shown that the false positive rate of the high-risk designation is substantially higher for African-American defendants compared to Caucasian defendants \citep{angwin2016machine}. We revisit this often-studied example of fairness auditing to determine if we can identify any other demographic groups that suffer from false positive rates at least 5\% higher than the average defendant. In particular, we audit over all intersections of race, sex, and age group ($n = 6781$, $|\cG| = 48$). \Cref{fig:compas_flags} plots the issued flags for subsets of the African-American subpopulation; we can further localize the false positive rate (FPR) disparity among African-American defendants to younger African-American defendants. As shown in \Cref{fig:compas_flags_fpr}, our method also flags nearly every under-25 subgroup, suggesting that this disparity affects young defendants more generally.

In the previous section, we \emph{certified} that the positive predictive value (PPV) of the COMPAS RPI is higher for African-American defendants compared to Caucasian defendants. Since the COMPAS creators claim that PPV is a more appropriate measure of fairness \citep{dieterich2016compas}, we investigate whether any other demographic groups suffer from harmful PPV disparities. As shown in \Cref{fig:compas_flags_ppv}, we are still able to flag certain subpopulations for having at least 5\% lower PPV compared to the average.

\begin{figure}
\centering
\captionsetup[subfigure]{width=0.9\linewidth}%
\begin{subfigure}{.5\textwidth}
  \centering
  \includegraphics[width=.9\linewidth]{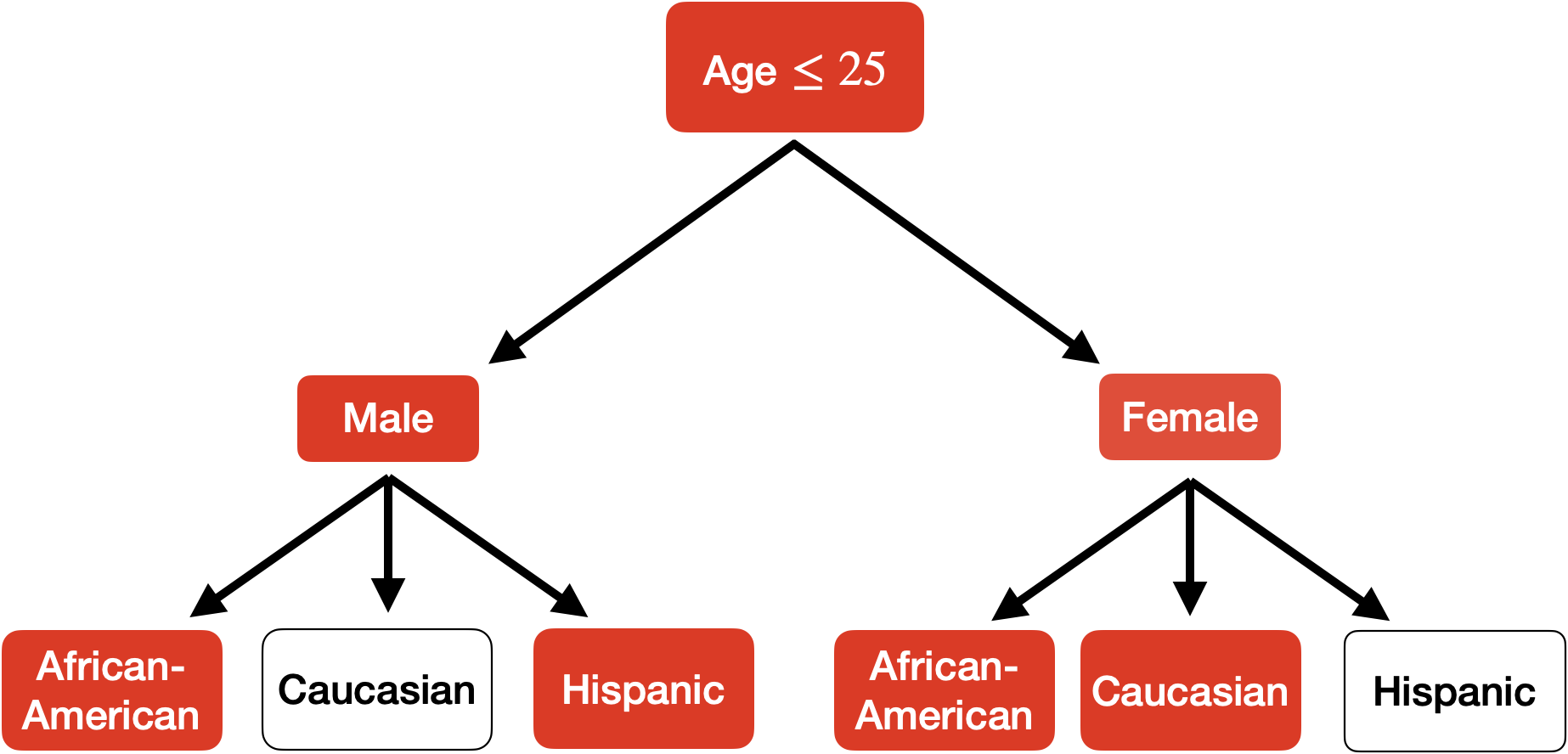}
  \caption{}
  \label{fig:compas_flags_fpr}
\end{subfigure}%
\begin{subfigure}{.5\textwidth}
  \centering
  \includegraphics[scale=0.23]{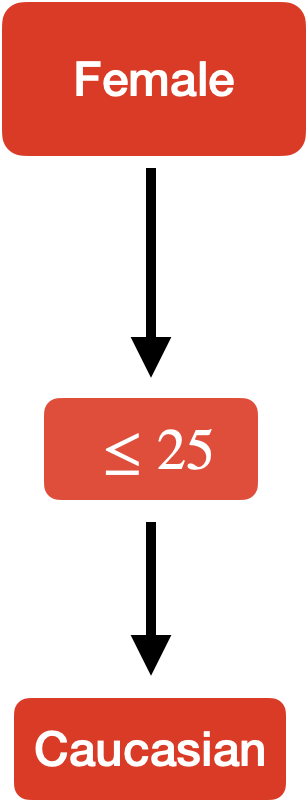}
  \caption{}
  \label{fig:compas_flags_ppv}
\end{subfigure}
\caption{In \Cref{fig:compas_flags_fpr}, the red boxes correspond to groups flagged as having substantially higher-than-average false positive rates. Most under-25 subpopulations suffer from this disparity. In \Cref{fig:compas_flags_ppv}, the red boxes denote groups flagged as having substantially lower-than-average positive predictive values. Though we are able to flag only three nested subpopulations (Females, Females under 25, and Caucasian females under 25), these results suggest that certain subpopulations still face disparities according to this measure.}
\label{fig:compas_flags_empirics}
\end{figure}

\section{Beyond subpopulations} \label{sec:beyond_subpop}
\subsection{Methods}
Auditing over subpopulations is equivalent to assessing performance over the class of distribution shifts indexed by the tilts, $\{\indic{(X, Y) \in G}\}_{G \in \cG}$. In particular,
\begin{align*}
    \epsilon(G) = \E_P \l[L\mid G\r] - \theta_P = E_{P_G} [L] - \theta_P,
\end{align*}
where $dP_G(x, y) \propto \indic{(x, y) \in G} dP(x, y)$. 
\begin{figure}
\centering
\captionsetup[subfigure]{width=0.9\linewidth}%
\begin{subfigure}{.5\textwidth}
  \centering
  \includegraphics[width=\linewidth]{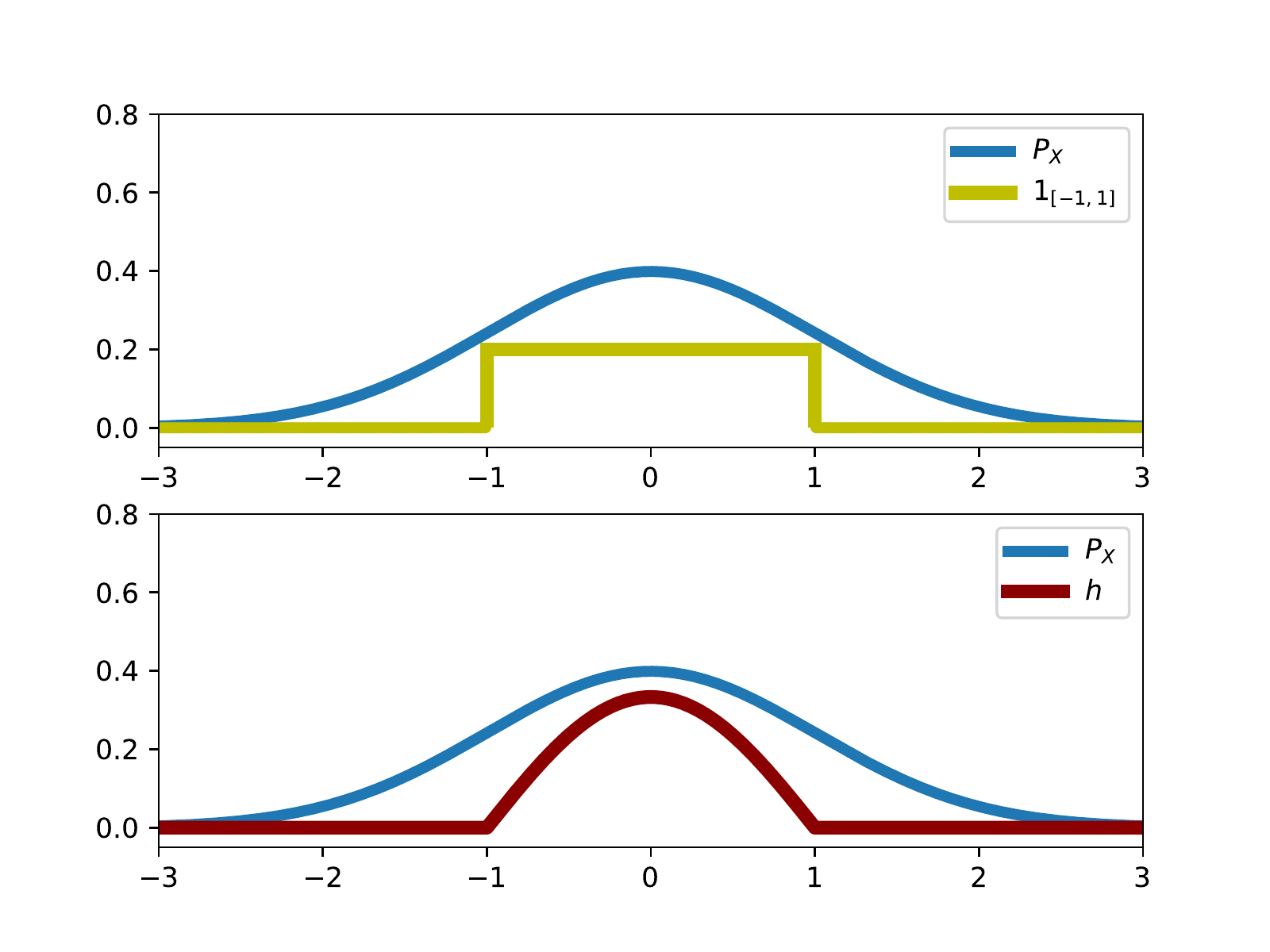}
  \caption{}
  \label{fig:tilt_funcs}
\end{subfigure}%
\begin{subfigure}{.5\textwidth}
  \centering
  \includegraphics[width=\linewidth]{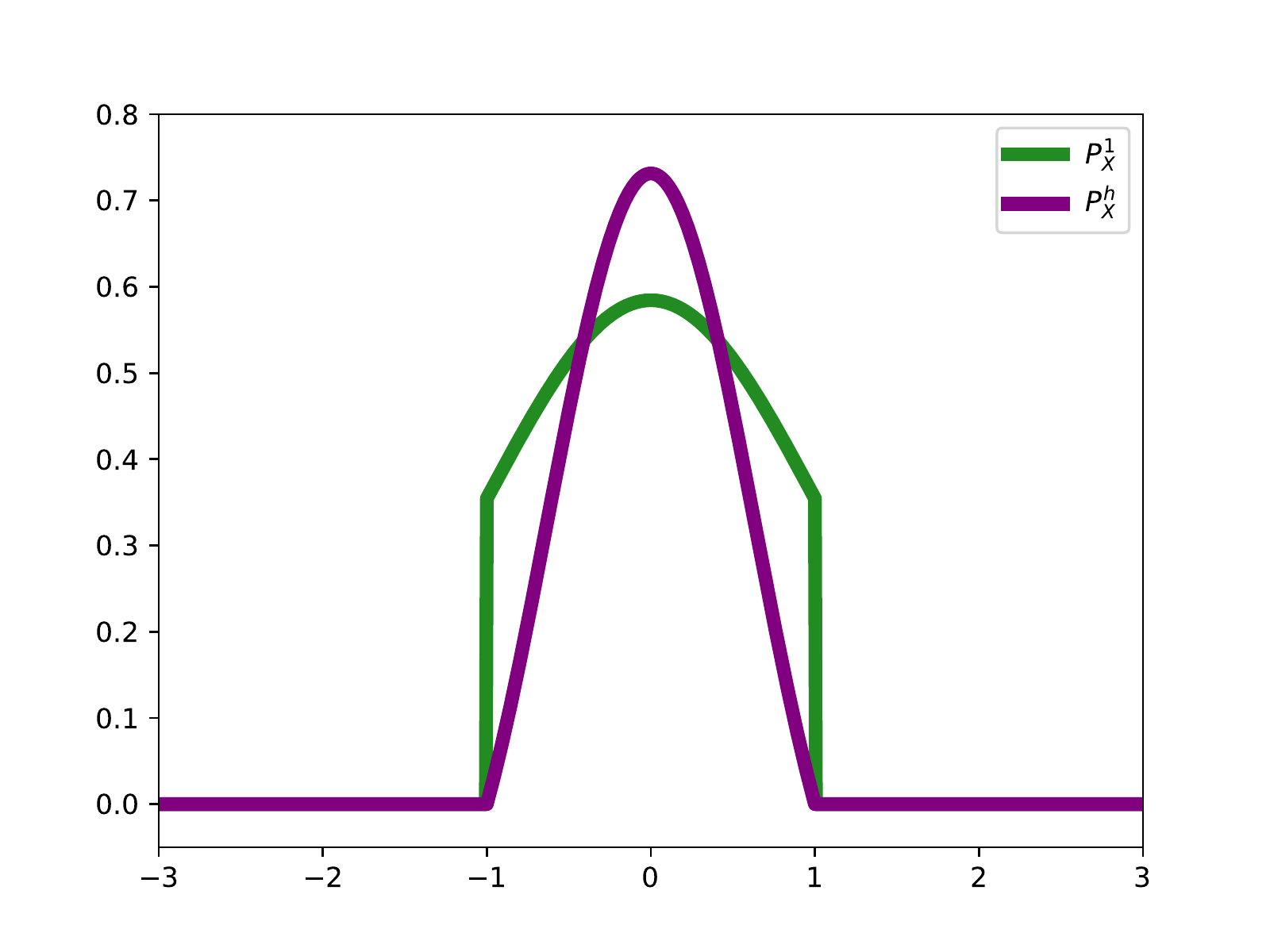}
  \caption{}
  \label{fig:tilted_dist}
\end{subfigure}
\caption{In \Cref{fig:tilt_funcs}, we plot an indicator tilt that corresponds to conditioning on $X$ belonging to the interval $[-1,1]$, and a non-negative tilt that resembles the indicator. \Cref{fig:tilted_dist} plots the (similar) probability densities over $X$ produced by both tilts.}
\label{fig:tilts}
\end{figure}
In this section, we consider the natural generalization of these tilts from 0-1 valued functions to collections of non-negative functions. For such a collection, $\cH$, we might wish to bound the following discrepancy for each $h$, 
\begin{align*}
    \epsilon(h) = \E_{P_h}[L(f(X), Y)] - \theta_P,
\end{align*}
where $dP_h(x, y) \propto h(x, y) dP(x, y)$. \Cref{fig:tilts} shows how non-negative functions that approximate an indicator can produce similar tilted distributions over $\cX$.

To motivate this auditing task, consider the problem of assessing an autonomous vehicle's performance over a diverse set of environments. While most work on distributional robustness focuses on assessing performance under the worst-case shift belonging to some set, such an appraisal of model robustness is inherently limited \citep{duchi2021learning}. The worst-case over some tractable collection of tilts is unlikely to be a useful proxy for most real-world distribution shifts. By contrast, our methods allow the modeler to query any shift belonging to the audited collection. As new environments are encountered, the modeler can reliably assess which environments may have benign effects on performance and which may be especially problematic. 

When $\cH$ is finite, our methods for certifying and flagging performance are unchanged. For infinite $\cH$, our previous results still apply: we can certify performance discrepancies over any VC class of tilts. While this is by no means an exhaustive list, such classes include the set of non-negative polynomials and any collection of bounded monotone tilts over a single covariate \citep{van2000asymptotic}.

Generalizing from binary to non-negative tilts expands auditing to certain function classes with \emph{infinite} VC dimension. We describe sufficient conditions for bootstrap-based auditing over general function classes $\cF$ in \Cref{sec:rkhs_appendix}, but here we highlight one collection in particular: the unit ball of a Reproducing Kernel Hilbert Space (RKHS). 

Let $h$ denote any non-negative function belonging to the unit ball of a RKHS; we denote the collection of such functions by $\cH^+_1$. Define $\epsilon(h)$ as the disparity under this tilt, i.e.,
\begin{align*}
\epsilon(h) \defeq \E_{P_{{h}}} [L] = \frac{\E_P \l[(L - \theta_P) h(X) \r]}{\E_P[h(X)]}.
\end{align*}
Let $\hat{\E}_n[f(X)]$ denote the plug-in estimator for the expectation of $f$ under the empirical distribution. Then, we define $\hat{\epsilon}(h) \defeq (\hat{\E}_n[h(X)])^{-1} \hat{\E}_n [(L - \hat{\theta}) h(X)]$. We will assume hereafter that $\theta_P = 0$, but a generalization of our approach to estimated targets is given in \Cref{sec:rkhs_appendix}.

We highlight two important characteristics of the RKHS auditing task. First, for a suitably chosen RKHS, $h$ can approximate \emph{any} smooth tilt of the covariate distribution, $P_X$, defined over a compact subset \citep{micchelli2006universal}. Given sufficient data, this allows us to issue guarantees on model performance for essentially arbitrary groups and covariate shifts. 

Auditing over the RKHS unit ball offers another advantage: we can construct a confidence bound for $\epsilon(h)$ without the onerous optimization present at each step of \textsf{\Cref{algo:bootstrap_ci}}. Recall that in \textsf{\Cref{algo:bootstrap_ci}}, we used the bootstrap to estimate the $(1 - \alpha)$-quantile of $\sup_{G \in \cG} \{\P_n(G) \cdot \P(G) \cdot (\hat{\epsilon}(G) - \epsilon(G))\}$. Each iteration then required solving a challenging combinatorial optimization problem over $\mathcal{G}$. At first glance, the analogous task for RKHS-based auditing appears even more difficult. For each bootstrap sample, we must compute
\begin{align*}
    \max_{h \in \cH_1^+} \{\hat{\E}_n[h(X)] \cdot \hat{\E}^*_b[h(X)] \cdot (\epsilon^*_b(h) - \hat{\epsilon}(h))\}.
\end{align*}
Naively, this optimization problem is intractable.

In \textsf{\Cref{algo:bootstrap_rkhs}}, however, we are able to reduce this task to computing the top eigenvalue of a low-rank matrix. Two observations are crucial to this reduction. First, the supremum of the process indexed by $h \in \cH_1^+$ is upper bounded by the supremum of the same process indexed by the unrestricted unit ball $\cH_1$, i.e.,
\begin{align}
    t^* \stackrel{\cdot}{=} \text{Quantile} \l(1 - \alpha; \sup_{h \in \cH_1} \{\E_P[h(X)] \cdot \hat{\E}_n[h(X)] \cdot (\hat{\epsilon}(h) - \epsilon(h))\} \r). \label{eqn:rkhs_quantile}
\end{align}

We can then simplify the maximization problem in each iteration of the bootstrap algorithm by exploiting the finite-dimensional representer theorem for RKHS functions. This theorem states that for any finite set $\{x_i\}_{i = 1}^n$ and $h \in \cH_1$, $\{h(x_i)\}_{i = 1}^n = \mathbf{K} \mathbf{w}$ for $\{\mathbf{K}_{ij}\}_{i,j = 1}^n = \{k(x_i, x_j)\}_{i,j=1}^n$ and some $\mathbf{w} \in \mathbb{R}^n$ \citep{steinwart2008support}. The other steps of this reduction, which are technical but uninformative, are deferred to the proof of \Cref{thm:uniform_ci_rkhs}.

\begin{algorithm}[t]
  \caption{Bootstrapping the RKHS confidence set critical value}
  \label{algo:bootstrap_rkhs}
  \begin{algorithmic}[1]
    \State \textbf{Input:} Kernel $k$, audit trail $\cD$, level $\alpha$, bootstrap samples $B$
    \State Define $\mathbf{L} \defeq \{L(f(x_i), y_i)\}_{i = 1}^n$;
    \State Define $\mathbf{K} \defeq \{k(x_i, x_j)\}_{i,j = 1}^n$;
    \For{$b = 1,\dots,B$}
        \State Sample $\mathbf{w} \sim \text{Mult}\l (n; \frac{1}{n},\dots,\frac{1}{n} \r)$;
        \State $\mathbf{A} = \frac{1}{n^2} \l(\l(\mathbf{w} \odot \mathbf{L}\r) \mathbf{1}^\top - \mathbf{w} \mathbf{L}^\top\r)$;
        \State $t^{(b)} = \lambda_{\max} \l(\mathbf{K}^{1/2} \l(\frac{\mathbf{A} + \mathbf{A}^\top}{2}\r) \mathbf{K}^{1/2} \r)$;
    \EndFor
    \State \textbf{Return:} $t^* = \text{Quantile} (1 - \alpha; \{t^{(b)}\}_{b = 1}^B ) $
  \end{algorithmic}
\end{algorithm}

Using the output of \textsf{\Cref{algo:bootstrap_rkhs}}, we construct a lower confidence bound for $\epsilon(h)$ as
\begin{align}
    \epsilon_{\text{lb}}(h) \defeq \hat{\epsilon}(h) - \frac{t^*}{\l(\frac{1}{n} \sum_{i = 1}^n h(x_i)\r)^2}. \label{eqn:rkhs_ub}
\end{align}
We remark that a rescaling similar to \eqref{eqn:rescaled_process} can also be applied in this setting. The resulting bootstrap computation for the rescaled RKHS process, however, is prohibitively expensive. Nevertheless, we include a detailed description of the appropriate rescaling and bootstrap algorithm in \Cref{sec:rkhs_appendix}.

\Cref{thm:uniform_ci_rkhs} states the assumptions under which $\epsilon_{\text{lb}}(h)$ is a simultaneously valid confidence bound.
\begin{theorem}[Simultaneous RKHS confidence bound validity] \label{thm:uniform_ci_rkhs}
Assume that $\var(L)$ is bounded away from $0$, $\|L\|_\infty$ and $\|k(X, X)\|_\infty$ are finite, $k(\cdot, x)$ is continuous for all $x$, and that $k(\cdot, \cdot)$ is a positive definite kernel. Then, 
\begin{align*}
    \lim_{n \to \infty} \P \l(\epsilon_{\textup{lb}}(h) \leq \epsilon(h) \textup{ for all $h \in \mathcal{H}^+_1$} \r) \geq 1 - \alpha.
\end{align*}
\end{theorem}

Commonly used kernels, such as the Gaussian and Laplace kernels, satisfy the assumptions given in \Cref{thm:uniform_ci_rkhs}. 

\subsection{Empirical results}
We first validate our coverage guarantee for RKHS-based confidence sets using a synthetic experiment in which the ground truth is known. Formally, we use the same data-generating process as in \eqref{eqn:heteroskedastic_dgp}, but discretize $\cX$ as follows,
\begin{align}
    X_i \simiid \text{Unif}(\{0,0.01,0.02,\dots,1\}),\quad Y_i = \mathcal{N}(\beta_0 X_i, X_i).\label{eqn:dgp_rkhs}
\end{align}
In each trial, we sample $\beta_0 \simiid \mathcal{N}(0, 1)$ and use $1000$ training points, $(X_i, Y_i)$, to fit the OLS predictor, $f(x) = \hat{\beta}^\top x$. We then audit over covariate shifts corresponding to the non-negative functions belonging to the unit ball of a Gaussian RKHS with varying bandwidths $\sigma \in \{0.1, 0.5, 1\}$.

\begin{table}[t]
	\centering
	\begin{tabular}{llllllll}
	\toprule
	&&\multicolumn{5}{c}{Sample size ($n$)}\\
	\cmidrule{3-7}
	&$\sigma$&100&200&400&800&1600\\
	\midrule
	\multirow{3}{*}{Model \eqref{eqn:dgp_rkhs}}&$1$&0.92&0.945&0.94&0.89&0.91\\
&$0.5$&0.93&0.955&0.935&0.885&0.915\\
&$0.1$&0.95&0.95&0.935&0.93&0.95\\
\bottomrule
	\end{tabular}
	\caption{Realized percentile of $t^*$ output by \textsf{\Cref{algo:bootstrap_rkhs}} with $B = 500$ and $\alpha = 0.1$. The nominal percentile is $1 - \alpha = 0.9$. All results are based on $200$ trials. In small samples, the bootstrap approximation can be conservative.}
	\label{tab:rkhs_validation}
\end{table}

Since evaluating the coverage under all non-negative functions in the RKHS is infeasible, we instead check that the bootstrap approximation to \eqref{eqn:rkhs_quantile} is valid. Recall that $t^*$ estimates the $(1 - \alpha)$-quantile of 
\begin{align}
    \sup_{h \in \cH_1} \{\E_P[h(X)] \cdot \hat{\E}_n[h(X)] \cdot (\hat{\epsilon}(h) - \epsilon(h))\}. \label{eqn:rkhs_process}
\end{align} 
In \Cref{tab:rkhs_validation}, we compute the percentile of \eqref{eqn:rkhs_process} realized by $t^*$ for the synthetic experiment, i.e., $\P(\sup_{h \in \cH_1} \{\E_P[h(X)] \cdot \hat{\E}_n[h(X)] \cdot (\hat{\epsilon}(h) - \epsilon(h))\} \leq t^*)$. Our results show that in small samples, the RKHS audit may be conservative: the realized percentile is sometimes larger than the nominal level.

\begin{figure}
\centering
\includegraphics[scale=0.4]{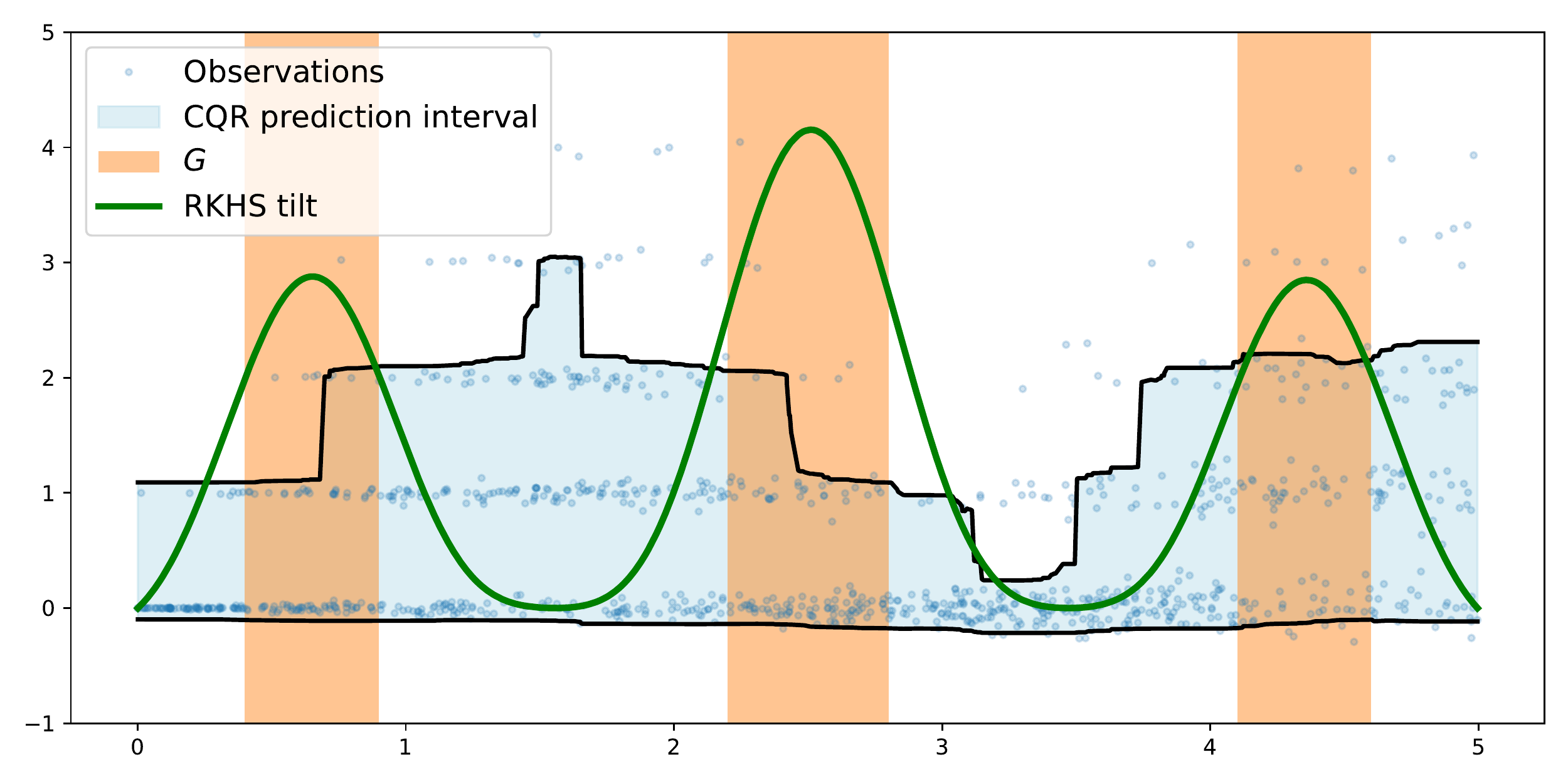}
\caption{To provide an approximate coverage guarantee on the union of the three highlighted intervals, we tilt the covariate distribution by the depicted function belonging to the Gaussian RKHS with bandwidth $\sigma = 0.5$. For the audit trail displayed, we can lower bound the ``tilted'' coverage at 84.0\%. }
\label{fig:rkhs_fit}
\end{figure}

We also revisit the synthetic dataset of \cite{romano2019conformalized} to practically demonstrate how \eqref{eqn:rkhs_ub} can provide coverage guarantees over complex subgroups. \Cref{fig:rkhs_fit} shows how an arbitrarily chosen union of 3 sub-intervals\footnote{This subgroup is not an element of the sub-interval collection we previously considered.} can be approximated by a Gaussian RKHS function with bandwidth $\sigma = 0.5$. Given an audit trail of $n = 1000$ points, we issue a conditional coverage lower bound of 84.0\%; over infinitely many such tilts, the issued bounds are simultaneously valid with probability at least $90\%$.

\paragraph{Acknowledgments}
We thank Lihua Lei, Jonathan Taylor, Isaac Gibbs, Tim Morrison, and Anav Sood for helpful discussions. We especially thank Kevin Guo for sharing his notes on bootstrap and empirical process theory. 

J.J.C. was supported by the John and Fannie Hertz Foundation. E.J.C. was supported by the Office of Naval Research grant N00014-20-1-2157, the National Science Foundation grant DMS-2032014, the Simons Foundation under award 814641, and the ARO grant 2003514594.

\newpage 
\bibliographystyle{plainnat}
\bibliography{citations}

\newpage
\appendix

\section{Connections to fairness} \label{sec:fairness_examples}
Recall that group fairness definitions ask for approximate parity of $L$ across all protected subpopulations \citep{dwork12fairness, hardt2016equality, corbett2017algorithmic, kearns2018preventing}. We show how to use our methods to audit two popular group fairness definitions below.

\paragraph{Multicalibration} We say that a predictor $f(X)$ is calibrated if the conditional expectation of the binary label $Y$ given that $f(X) = v$ matches $v$, i.e., $\E[Y \mid f(X) = v] = v$. For the purpose of measuring classification quality across groups, we might require the binary classifier to be calibrated over many subgroups. Thus, for all subgroups $G \in \tilde{\cG}$ and (potentially binned) predicted values $v$, the $\gamma$-multicalibration fairness criterion \citep{hebert2018multicalibration} requires that
\begin{align*}
    |\E[Y \mid f(X) = v, X \in G] - v| \leq \gamma.
\end{align*}
The formal multicalibration definition excludes groups smaller than some auditor-set threshold. Although the user may filter $\tilde{\cG}$ as they see fit, a-priori, our methods do not exclude groups by their size. However, we can set a similar threshold in the rescaled certification method (\textsf{\Cref{algo:bootstrap_ci_rescaled}}). This threshold does not prevent the auditor from querying any group for a bound on $\gamma$, but it ensures that confidence sets are narrower for groups above the threshold.

Choosing $L, \theta_P$, and $\cG$ carefully, we can apply our methods to certify and flag multicalibration over $\tilde{\cG}$. Letting $\mathcal{V}$ denote the set of unique values $f(X)$ can take, we set 
\begin{gather*}
    L \defeq Y - f(X) \qquad \theta_P \defeq 0 \\
    \cG \defeq \{G \cap \{(X, Y) \mid f(X) = v\} \mid G \in \tilde{\cG}, v \in \mathcal{V}\}.
\end{gather*}
Let $G_v$ denote each group in $\cG$. 

Certifying $\gamma$-multicalibration for $G$ is equivalent to establishing that $\max_{v} |\epsilon(G_v)| \leq \gamma$.  Using our auditing methods, we can estimate $\gamma$ by constructing simultaneously valid confidence intervals for all $\{\epsilon(G_v)\}_{G_v \in \cG}$. For a particular $G$, our bound on $\gamma$ then equals the maximum (absolute) value taken over all issued intervals for $\{\epsilon(G_v)\}_{v \in \mathcal{V}}$. For a fixed threshold $\gamma$, we could also run our Boolean certification method with $H_0(G_v) : |\epsilon(G_v)| \geq \gamma$ and certify $G$ if the null is rejected for all $\{G_v\}_{v \in \mathcal{V}}$. 

Alternatively, we might flag $G$ for violation of $\gamma$-multicalibration by testing $H_0(G_v) : |\epsilon(G_v)| > \gamma + \epsilon$ for some disparity tolerance $\epsilon$. Using the method outlined in \Cref{sec:flag_G}, we can construct p-values for each of these null hypotheses. Then, the \textsf{BH} procedure is run on all of the p-values computed. If any $H_0(G_v)$ is rejected, we flag $G$.

\paragraph{Equalized odds}
Given a collection of subsets of $\cX$ denoted by $\cG_X$, we say that a binary predictor $f$ satisfies the equalized odds criterion \citep{hardt2016equality, woodworth2017learning} if for all $G \in \cG_X$, both its true positive rates are equalized,
\begin{align*}
    \P(f(X) = 1 \mid Y = 1, X \in G) = \P(f(X) = 1 \mid Y = 1),
\end{align*}
and its false positive rates are equalized,
\begin{align*}
    \P(f(X) = 1 \mid Y = 0, X \in G) = \P(f(X) = 1 \mid Y = 0).
\end{align*}

A practitioner interested in this fairness criterion might then wish to audit the performance disparity $\epsilon(G) \defeq \max(|\epsilon(G_0)|, |\epsilon(G_1)|)$ where
\begin{align*}
    \epsilon(G_i) \defeq \P(f(X) = 1 \mid Y = i, X \in G) - \P(f(X) = 1 \mid Y = i).
\end{align*}

We can instantiate a certification audit for $\epsilon(G)$ by running our methods twice: once to bound $\epsilon(G_0)$ and a second time for $\epsilon(G_1)$. For the first audit, let 
\begin{gather*}
    L \defeq \indic{f(x) = 1} \qquad \theta_P \defeq \P(f(X) = 1 \mid Y = 0)  \qquad 
    \cG = \l \{ G \times \{0\} \mid G \in \cG_X \r \}.
\end{gather*}
For the second audit, let 
\begin{gather*}
    L \defeq \indic{f(x) = 1} \qquad \theta_P \defeq \P(f(X) = 1 \mid Y = 1) \qquad 
    \cG = \l \{ G \times \{1\} \mid G \in \cG_X \r \}.
\end{gather*}
We construct confidence sets or Boolean certificates using a nominal Type I error threshold of $\alpha/2$ for each auditing task. We remark that this union bound is practically tight since the two data sets corresponding to these audits are disjoint. Our final certificate on $\epsilon(G)$ then consists of the ``worse'' of the two auditor outputs. For example, we might upper bound $\epsilon(G)$ by the maximum (absolute) value included in the two issued confidence intervals for $\epsilon(G_0)$ and $\epsilon(G_1)$.  

For the flagging task, we test $H_0(G_i) : |\epsilon(G_i)| \leq \epsilon$. Using the method outlined in \Cref{sec:flag_G}, we can construct p-values for each of these null hypotheses. Then, the \textsf{BH} procedure can be directly run on all of the p-values computed. We flag $G$ if either $H_0(G_0)$ or $H_0(G_1)$ is rejected.

\paragraph{Individual fairness} Other fairness criteria fall into a category known as ``individual fairness'' measures. These quantify the intuition that similar inputs, $x$ and $x'$, should be treated by $f$ similarly \citep{dwork12fairness}. While this definition of fairness cannot be tested using sampled model predictions, one might audit whether some notion of individual fairness holds with high probability among protected subgroups. Though we do not elaborate on such an extension here, we remark that prior work on the bootstrap of U-statistics and processes allows the natural extension of our auditing procedures to  fairness measures defined over pairs of data points \citep{arcones1992bootstrap, huskova1993consistency}.

\section{Certification audits}
\subsection{Notation and review}
We will begin by defining some relevant notation and reviewing certain basic results about the convergence of stochastic processes.  Given $n$ i.i.d. samples, $\P_n \defeq n^{-1} \sum_{i = 1}^n \delta_{(X_i, Y_i)}$ denotes their empirical distribution.  If $\{ (X_i^*, Y_i^*) \}_{i = 1}^n$ are i.i.d. samples from $\P_n$ conditional on $\{ (X_i, Y_i) \}_{i = 1}^n$, then $\P_n^* \defeq n^{-1} \sum_{i = 1}^n \delta_{(X_i^*, Y_i^*)}$ denotes their empirical distribution.

For a function $f : \cX \times \cY \to \R^k$, $P[f]$ is shorthand for $\E_P[f(X, Y)]$, $\P_n[f]$ is shorthand for $n^{-1} \sum_{i = 1}^n f(X_i, Y_i)$, and $\P_n^*[f]$ is shorthand for $n^{-1} \sum_{i = 1}^n f(X_i^*, Y_i^*)$.  We also write $(\P_n - P)[f]$ in place of $\P_n[f] - P[f]$.  Given a class of functions $\cF$, we think of $f \mapsto \sqrt{n}(\P_n - P)[f]$ as a mapping belonging to $\ell_\infty(\cF)$. 

We will typically require and/or argue that $\cF$ is a $P$-Donsker class. This means that the empirical process indexed by $f \in \cF$ converges in distribution to a tight Gaussian limit in $\ell_{\infty}(\cF)$. Formally, $\sqrt{n}(\P_n - P)[\cdot] \cd \mathbb{G}[\cdot]$ where the limiting process $f \mapsto \mathbb{G}[f]$ is a Gaussian process that is also a tight Borel-measurable element of $\ell_\infty(\cF)$. If $\cF$ is a $P$-Donsker class, then it is also $P$-Glivenko-Cantelli \citep{van2000asymptotic}, i.e., $\sup_{f \in \cF} |(\P_n - P) [f]| \cp 0$.

Next, we recall some results relating Donsker classes to VC dimension and bootstrap consistency. When we consider function classes of indicators indexed by subpopulations $\cG$, i.e. $\cF = \{\indic{(X, Y) \in G} \mid G \in \cG\}$, the Donsker property is \emph{equivalent} to assuming that $\cG$ is VC.
\begin{lemma}[Theorem 11.4.1 in \cite{dudley1984course}] \label{lma:vc_iff}
Under suitable measurability assumptions, $\cF := \{ (x, y) \mapsto \mathbf{1} \{ (x, y) \in G \} \, : \, G \in \cG \}$ is Donsker if and only if $\text{VC}(\cG) < \infty$.
\end{lemma}

Last, we recall that the bootstrap approximation, $\sqrt{n} (\P_n^* - \P_n) [h]$, is valid if $\cH$ is $P$-Donsker. Let $\mathbb{G}_n^*$ be shorthand for the bootstrap empirical process $\sqrt{n} (\P_n^* - \P_n)$ and $\mathbb{G}$ denote the limiting empirical process. 
\begin{lemma}[Theorem 23.7 in \cite{van2000asymptotic}] \label{lma:donsker_bootstrap}
For every Donsker class $\cH$ of measurable functions with finite envelope function $F$, $\sup_{g \in \text{BL}_1(\ell^\infty(\cH))} \l| \E^*[g(\mathbb{G}^*_n)] - \E[g (\mathbb{G})] \r| \cp 0$ where $BL_1(\ell_{\infty}(\cH))$ is the set of bounded 1-Lipschitz functions taking $\ell_{\infty}(\cH)$ into $\R$.
\end{lemma}

\subsection{Proofs of certification theorems}
\label{sec:power_certify}
\paragraph{Estimating $\theta_P$.} In the main text, we assumed that $\hat{\theta}$ is a differentiable function of sample means. Here we will generalize this assumption to accommodate other estimators that can be bootstrapped. 

We assume that the bootstrap and sampling distributions of $\hat{\theta}$ admit asymptotic linear expansions:
\begin{align*}
    \sqrt{n}(\hat{\theta}(\cD^*) - \hat{\theta}) &= \frac{1}{\sqrt{n}} \sum_{i = 1}^n \psi(X^*_i, Y^*_i) + o_{P_n}(1) \\
    \sqrt{n} (\hat{\theta} - \theta_P) &= \frac{1}{\sqrt{n}} \sum_{i = 1}^n \psi(X_i, Y_i) + o_P(1).
\end{align*}
Here, $\psi$ is an influence function with mean zero and finite variance.  It is easy to verify that this condition is satisfied by the $\hat{\theta}$ given in the main text. Our generalization enables auditing even if $\hat{\theta}$ is more complicated, e.g., it is the solution to some maximum likelihood estimation problem. 


Next, we prove the main theorems of \Cref{sec:certify_G}. As explained in the main text, the primary technical difficulty lies in proving that $t^*$ consistently estimates the $(1 - \alpha)$-quantile of a particular stochastic process. There are two crucial preliminary results required to establish this result. First, we must show that bootstrap consistency implies consistency of the quantile estimate. Second, we must show that our proposed rescaling estimator $\hat{s}(G)$ is uniformly consistent for some estimand.

\paragraph{Quantile consistency}
First, to provide a unique definition of the quantile, we define the $\alpha$-quantile of a random variable $X$ as
\begin{align*}
    \text{Quantile}(\alpha; X) \defeq \inf_{x} \{x \mid \alpha \leq \P(X \leq x)\}.
\end{align*}

The proof of the main theorem requires continuity and strict increase for the supremum of the limiting process at its $(1 - \alpha)$-quantile. The following result establishes mild conditions under which this holds.
\begin{lemma} \label{cor:quantile_gaussian_2}
Let $\mathbb{Z}(G)$ denote some Gaussian process that is the limit of an empirical process indexed by some countable class $\cG$. Then, if $\textup{Var}(\mathbb{Z}(G)) > 0$ for some $G \in \cG$, the distribution function of $\sup_{G \in \cG} \mathbb{Z}(G)$ is continuous and strictly increasing on $\R_+$ and the distribution function of $\inf_{G \in \cG} \mathbb{Z}(G)$ is continuous and strictly increasing on $\R_-$.
\end{lemma}
\begin{proof}
    Because $\mathbb{Z}(G)$ is a centered process, we remark that $\sup_{G} \mathbb{Z}(G) \stackrel{d}{=} -\inf_{G} \mathbb{Z}(G)$. First, we prove that for all $x > 0$, the distribution function of $\sup_{G \in \cG} \mathbb{Z}(G)$ is continuous at $x$. Under the stated assumptions, \citet[Corollary~1.3]{gaenssler2007continuity} prove that $\sup_{G} |\mathbb{Z}(G)|$ has a continuous distribution function at $x$, i.e., $\P(\sup_{G} |\mathbb{Z}(G)| = x) = 0$. Then,
    \begin{align*}
        0 = \P(\sup_{G} |\mathbb{Z}(G)| = x) = \P(\{\sup_{G} \mathbb{Z}(G) = x\} \cup \{\inf_{G} \mathbb{Z}(G) = -x\}),
    \end{align*}
    which implies $\P(\sup_{G} \mathbb{Z}(G) = x) = \P(\inf_{G} \mathbb{Z}(G) = -x) = 0$. This is equivalent to the desired continuity statement.

    Next, to prove strict increase at $x$, we argue by contradiction. If $\sup_{G} \mathbb{Z}(G)$ did \emph{not} have a strictly increasing distribution function on $\R_+$, then there exists $0 < x_1 < x_2$ such that 
    \begin{align*}
        \P(\sup_{G} \mathbb{Z}(G) \leq x_1) = \P(\sup_{G} \mathbb{Z}(G) \leq x_2) 
        \iff 
        \P(\inf_{G} \mathbb{Z}(G) \geq -x_1) = \P(\inf_{G} \mathbb{Z}(G) \geq -x_2).
    \end{align*}
    \citet[Corollary~1.3]{gaenssler2007continuity} also prove that $\sup_{G} |\mathbb{Z}(G)|$ satisfies $\P(\sup_{G} |\mathbb{Z}(G)| \leq x_1) < \P(\sup_{G} |\mathbb{Z}(G)| \leq x_2)$.  Since 
    \begin{multline*}\P(\sup_{G} |\mathbb{Z}(G)| \leq x) = \P(\sup_{G} \mathbb{Z}(G) \leq x) + \P(\inf_G \mathbb{Z}(G) \geq -x) \\
    - \P(\{\sup_{G} \mathbb{Z}(G) \leq x\} \cap \{\inf_{G} \mathbb{Z}(G) \geq -x\}),
    \end{multline*}
    we conclude that
    \begin{align*}
        \P(\{\sup_{G} \mathbb{Z}(G) \leq x_1\} \cap \{\inf_{G} \mathbb{Z}(G) \geq -x_1\}) > \P(\{\sup_{G} \mathbb{Z}(G) \leq x_2\} \cap \{\inf_{G} \mathbb{Z}(G) \geq -x_2\}).
    \end{align*}
    This is a contradiction since the event on the LHS is a subset of the event on the RHS. We conclude that both $\sup_G \mathbb{Z}(G)$ and $\inf_G \mathbb{Z}(G)$ have strictly increasing distribution functions on the positive and negative reals, respectively.
\end{proof}

We make two comments on \Cref{cor:quantile_gaussian_2}. First, the cited result of \citet{gaenssler2007continuity} is stated for empirical processes indexed by a VC class, but their argument applies to any limiting Gaussian process indexed by a $P$-Donsker class. Second, the countable assumption is not crucial, and is only included to avoid verifying certain technical measurability conditions. In this article, we expect the ``pointwise measurability'' (approximately equivalent to well-approximation by a dense countable subset) condition given in \citet{vaart1996weak} to hold, but this condition must be checked for each function class. As a consequence, for simplicity, we will also assume throughout that the function classes we work with are countable. 

\paragraph{Uniform consistency of $\hat{s}$}
Recall that we define 
\begin{align*}
    \hat{s}(G) \defeq \max(\P_n(G), p_*)^{3/2} \cdot \hat{\sigma}(G, w_0)
\end{align*}
for
\begin{align*}
   \hat{\sigma}(G, w_0) \defeq \l(\frac{\P_n(G)}{\P_n(G) + w_0}\r) \cdot \hat{\sigma}_G + \l(\frac{w_0}{\P_n(G) + w_0}\r) \cdot  \sqrt{\widehat{\text{Var}}(L)}
\end{align*}
and
\begin{align*}
    \hat{\sigma}^2_G \defeq \widehat{\text{Var}}(L \mid G) + \P_n(G) \l(\widehat{\text{Var}}(\psi) - 2 \cdot \widehat{\text{Cov}}(L, \psi \mid G) \r).
\end{align*}

\begin{lemma} \label{lma:rescaling_consistency}
    Assume that $0 < \textup{Var}(L) < \infty$ and $\text{VC}(\cG) < \infty$. If $p_*, w_0 > 0$, then $\sup_{G \in \cG} | 1/\hat{s}(G) - 1/s(G)| \cp 0$, where 
    \begin{align*}
        s(G) = \max(\P(G), p_*)^{3/2} \cdot  \l[\l(\frac{\P(G)}{\P(G) + w_0}\r) \cdot \sigma_G + \l(\frac{w_0}{\P(G) + w_0}\r) \cdot \sqrt{\textup{Var}(L)} \r].
    \end{align*}
\end{lemma}

\begin{proof}
    To simplify notation, we will replace $\indic{(X, Y) \in G}$ with $\indG$ throughout. 
    
    We first compute the asymptotic variance of $\sqrt{|G|} (\hat{\epsilon}(G) - \epsilon(G))$ by linearizing the random quantity:
    \begin{align*}
        \sqrt{|G|} \l(\hat{\epsilon}(G) - \epsilon(G)\r) &= \sqrt{|G|} \l[\l(\frac{1}{|G|} \sum_{(X_i, Y_i) \in G} L_i - \hat{\theta}\r) - (\E[L \mid G] - \theta_P) \r] \\
        &= \sqrt{|G|} \l[\frac{1}{|G|} \sum_{i = 1}^n (L_i - \E[L \mid G]) \cdot \indG - \frac{1}{n} \sum_{i = 1}^n \psi_i \r] \\
        &= \sqrt{n} \l[\frac{1}{n} \sum_{i = 1}^n \frac{L_i - \E[L \mid G]}{\sqrt{\P(G)}} \cdot \indG - \sqrt{\P(G)} \cdot \psi_i \r] + o_P(1).
    \end{align*}
    So, we conclude that
    \begin{align*}
        \sigma^2_G &\defeq \text{Var} \l(\frac{L - \E[L \mid G]}{\sqrt{\P(G)}} \cdot \indG - \sqrt{\P(G)} \cdot \psi \r) \\
        &= \text{Var} \l(\frac{L - \E[L \mid G]}{\sqrt{\P(G)}} \cdot \indG \r) + \text{Var} \l(\sqrt{\P(G)} \cdot \psi \r) -2 \cdot \text{Cov}\l(\frac{L - \E[L \mid G]}{\sqrt{\P(G)}} \cdot \indG, \sqrt{\P(G)} \cdot \psi  \r) \\
        &= \text{Var}(L \mid G) + \P(G) \cdot \left(\text{Var}(\psi) - 2 \cdot \text{Cov}(L, \psi \mid G) \right).
    \end{align*}

    Next, after some rearrangement, observe that $1/\hat{s}^2(G) - 1/s^2(G)$ equals 
    \begin{multline*}
        \frac{\P_n(G) + w_0}{\max(\P_n(G), p_*)^{3} \cdot \l(\P_n(G) \cdot \hat{\sigma}_G + w_0 \cdot \sqrt{\widehat{\text{Var}}(L)}\r)} \\
        - \frac{\P(G) + w_0}{\max(\P(G), p_*)^{3/2} \cdot \l(\P(G) \cdot \sigma_G + w_0 \cdot \sqrt{\text{Var}(L)}\r)}.
    \end{multline*}
    Combining the two fractions yields a numerator of
    \begin{multline*}
        (\P_n(G) + w_0) \l(\max(\P(G), p_*)^{3/2} \cdot \l(\P(G) \cdot \sigma _G + w_0 \cdot \sqrt{\text{Var}(L)}\r) \r) \\
        - (\P(G) + w_0) \l(\max(\P_n(G), p_*)^{3/2} \cdot \l(\P_n(G) \cdot \hat{\sigma}_G + w_0 \cdot \sqrt{\widehat{\text{Var}}(L)} \r) \r)
    \end{multline*}
    and a denominator of
    \begin{multline*}
        \max(\P_n(G), p_*)^{3/2} \cdot \max(\P(G), p_*)^{3/2} \cdot \l(\P_n(G) \cdot \hat{\sigma}_G + w_0 \cdot \sqrt{\widehat{\text{Var}}(L)}\r)  \\ 
        \cdot \l(\P(G) \cdot \sigma_G + w_0 \cdot \sqrt{\text{Var}(L)}\r).
    \end{multline*}
    To prove uniform consistency, it is sufficient to prove that
    \begin{multline*}
        \lim_{n \to \infty} \inf_{G \in \cG} \max(\P_n(G), p_*)^{3/2} \cdot \max(\P(G), p_*)^{3/2} \cdot \l(\P_n(G) \cdot \hat{\sigma}_G + w_0 \cdot \sqrt{\widehat{\text{Var}}(L)}\r) \\\cdot \l(\P(G) \cdot \sigma_G + w_0 \cdot \sqrt{\text{Var}(L)}\r) > 0
    \end{multline*}
    and
    \begin{multline*}
        \lim_{n \to \infty} \sup_{G \in \cG} \l | (\P_n(G) + w_0) \l(\max(\P(G), p_*)^{3/2} \cdot \l(\P(G) \cdot \sigma_G + w_0 \cdot \sqrt{\text{Var}(L)}\r) \r) \r. \\
        \l . - (\P(G) + w_0) \l(\max(\P_n(G), p_*)^{3/2} \cdot \l(\P_n(G) \cdot \hat{\sigma}_G + w_0 \cdot \sqrt{\widehat{\text{Var}}(L)} \r) \r) \r| = 0.
    \end{multline*}

    For the first of these tasks, observe that the denominator is lower bounded by
    \begin{align*}
        p_*^3 \cdot \l(w_0^2 \cdot \sqrt{\widehat{\text{Var}}(L)} \cdot \sqrt{\text{Var}(L)}\r) \to p_*^3 \cdot w_0^2 \cdot \text{Var}(L) > 0.
    \end{align*}
    The numerator requires a more careful analysis. We distribute the $\P_n(G) + w_0$ factor, apply the triangle inequality, and analyze the first term in the numerator:
    \begin{multline*}
        \l | \P_n(G) \cdot \max(\P(G), p_*)^{3/2} \cdot \l(\P(G) \cdot \sigma _G + w_0 \cdot \sqrt{\text{Var}(L)}\r) \r .\\
        -  \l. \P(G) \cdot \max(\P_n(G), p_*)^{3/2} \cdot \l(\P_n(G) \cdot \hat{\sigma}_G + w_0 \cdot \sqrt{\widehat{\text{Var}}(L)} \r) \r|.
    \end{multline*}
    Adding and subtracting $\P(G) \cdot \max(\P(G), p_*)^{3/2} \cdot (\P(G) \cdot \sigma_G + w_0 \cdot \sqrt{\text{Var}(L)} )$ and applying a triangle inequality, we obtain:
    \begin{multline*}
        \l |\P_n(G) - \P(G)\r| \cdot \max(\P(G), p_*)^{3/2} \cdot (\P(G) \cdot \sigma_G + w_0 \cdot \sqrt{\text{Var}(L)} ) \\
        + \l | \max(\P(G), p_*)^{3/2} - \max(\P_n(G), p_*)^{3/2} \r| \cdot \P(G) \cdot (\P(G) \cdot \sigma_G + w_0 \cdot \sqrt{\text{Var}(L)} )  \\
         + \P(G) \cdot \max(\P_n(G), p_*)^{3/2} \cdot \l|\P(G) \cdot \sigma_G - \P_n(G) \cdot \hat{\sigma}_G  \r| \\
         + \P(G) \cdot \max(\P_n(G), p_*)^{3/2} \l |w_0 \cdot (\sqrt{\text{Var}(L)} - \sqrt{\widehat{\text{Var}}(L)} )  \r|.
    \end{multline*}
    Since the sample variance of $L$ does not depend on $G$ and is consistent under the stated assumptions, the uniform convergence of the fourth term to $0$ is easy to see. To prove uniform convergence of the first and second terms, we must show that the $n$-independent term is bounded, i.e.,
    \begin{align*}
        (\P(G) \cdot \sigma_G + w_0 \cdot \sqrt{\text{Var}(L)}) &\leq \P(G)^{1/2} \sqrt{\E[(L - \E[L \mid G])^2 \cdot \indG]} + w_0 \cdot \sqrt{\text{Var}(L)}) \\
        &\leq \P(G)^{1/2} \sqrt{\E[(L - \E[L])^2 \cdot \indG]} + w_0 \cdot \sqrt{\text{Var}(L)}) \\
        &\leq \P(G)^{1/2} \sqrt{\E[(L - \E[L])^2]} + w_0 \cdot \sqrt{\text{Var}(L)}) \\
        &\leq (1 + w_0) \cdot \sqrt{\text{Var}(L)} = C < \infty.
    \end{align*}
    Thus, since $\text{VC}(\cG) < \infty$, $(\P_n - P) [\indG]$ is $P$-Glivenko-Cantelli, and 
    \begin{multline*}
        \sup_{G \in \cG} \l |(\P_n(G) - \P(G)) \cdot \max(\P(G), p_*)^{3/2} \cdot (\P(G) \cdot \sigma_G + w_0 \cdot \sqrt{\text{Var}(L)} ) \r| \\ \leq C \cdot \sup_{G \in \cG} \l |\P_n(G) - \P(G) \r| \cp 0.
    \end{multline*}   
    Since a $P$-Glivenko Cantelli class is preserved under truncation (Example~2.10.11 in \cite{vaart1996weak}), we then also conclude that
    \begin{multline*}
        \sup_{G \in \cG} |\max(\P(G), p_*)^{3/2} - \max(\P_n(G), p_*)^{3/2}| \cdot \P(G) \cdot \l(\P(G) \cdot \sigma_G + w_0 \cdot \sqrt{\text{Var}(L)} \r)\\
        \leq C \cdot \sup_{G \in \cG} |\max(\P(G), p_*)^{3/2} - \max(\P_n(G), p_*)^{3/2}| \cp 0.
    \end{multline*}
    Last, we show the third term converges to $0$. The $n$-independent part is bounded by $1$, so we can ignore that. So, we need to show that $\sup_{G \in \cG} |\P(G) \cdot \sigma_G - \P_n(G) \cdot \hat{\sigma}_G| \cp 0$. To avoid writing the square root, we will square the two terms for the time being. Then,
    \begin{align*}
        \P(G)^2 \cdot \sigma^2_G &= \P(G)^2 \l[\text{Var}(L \mid G) + \P(G) \cdot (\text{Var}(\psi) - 2\cdot \text{Cov}(L, \psi \mid G))\r] \\
        &= \P(G)^2 (\E[L^2 \mid G] - \E[L \mid G]^2) \\
        &\quad + \P(G)^3 \cdot (\text{Var}(\psi) - 2 \cdot (\E[L \cdot \psi \mid G]  - \E[L \mid G] \E[\psi \mid G])) \\
        &= \P(G) \E[L^2 \cdot \indG] - \E[L \cdot \indG]^2 + \P(G)^3 \cdot \text{Var}(\psi) \\
        &\quad - 2 \cdot \P(G)^2 \cdot \E[(L \cdot \psi) \indG] \\
        &\quad - 2 \cdot \P(G) \cdot \E[L \cdot \indG]\E[\psi \cdot \indG]
    \end{align*}
    We can obtain an analogous expansion of $\P_n(G)^2 \hat{\sigma}^2_G$. Then, we need to show that each matching term converges uniformly. We will only explicitly work out the argument for one pair of terms, but the argument for the remainder should be clear. First, observe that $\cF_g = \{g \cdot \indG : G \in \cG\}$ is a $P$-Glivenko-Cantelli class so long as $P[|g|] < \infty$ (Corollary~3 in \cite{gine1984some}). Then,
    \begin{multline*}
        \sup_{G \in \cG} \l |\P_n(G)^2 \cdot \P_n[(L \cdot \psi) \cdot \indG] - \P(G)^2 \cdot P[(L \cdot \psi) \cdot \indic{X, Y) \in G}] \r| \\
        \leq \sup_{G \in \cG} \l |\P_n(G)^2 \cdot (\P_n[(L \cdot \psi) \cdot \indG] -P[(L \cdot \psi) \cdot \indG] )\r|  \\
        + \sup_{G \in \cG} \l|(\P_n(G)^2 - \P(G)^2) \cdot  P[(L \cdot \psi) \cdot \indG] \r|
    \end{multline*}
    We can further upper bound the RHS by
    \begin{align*}
        \underbrace{\sup_{G \in \cG} \l| (\P_n - P)[(L \cdot \psi) \cdot \indG] \r|}_{o_P(1)} \, + \, 2 \cdot \underbrace{P[|L \cdot \psi|]}_{\text{$\leq P[L^2]P[\psi^2] < \infty$}} \cdot \underbrace{\sup_{G \in \cG} |\P_n(G) - \P(G)|}_{o_P(1)}.
    \end{align*}
    A similar argument can be made for each of the other terms in this expression. 

    The above derivation was only for the term obtained by multiplying $\P_n(G)$. The proof of uniform convergence to $0$ is essentially identical for the other term, so we will not repeat the derivation. Thus, we conclude that the numerator converges uniformly over all $G$ to $0$. This yields the claimed uniform consistency of $1/\hat{s}(G)$.
\end{proof}

\paragraph{Bound certification}
First, we prove that our lower confidence bound construction is valid.
\textsf{\Cref{algo:bootstrap_ci_appendix}} restates our method for defining the critical value $t^*$. Then, using the output of \textsf{\Cref{algo:bootstrap_ci_appendix}}, we construct an asymptotically valid lower confidence bound for arbitrary $G \in \cG$ by setting
\begin{align*}
\epsilon_{\text{lb}}(G) = \hat{\epsilon}(G) - t^*  \cdot \frac{\hat{s}(G)}{\P_n(G)^2}.
\end{align*}
\Cref{thm:uniform_ci_appendix} restates \Cref{thm:uniform_ci}. 

\begin{algorithm}
  \caption{Bootstrapping the (rescaled) lower confidence bound critical value}
  \label{algo:bootstrap_ci_appendix}
  \begin{algorithmic}[1]
    \State \textbf{Input:} Subpopulations $\cG$, audit trail $\cD$, level $\alpha$, threshold $p_*$, weight $w_0$, number of bootstrap samples $B$
    \For{$b = 1,\dots,B$}
        \State Let $\cD^*_b$ be a sample with replacement of size $n$ from $\cD$;
        \State Define $\P^*_b(G) \defeq \frac{1}{n} \sum_{(x^*_i, y^*_i) \in \cD^*_b} \indic{(x^*_i, y^*_i) \in G}$;
        \State Define $\epsilon^*_b(G) \defeq \frac{1}{\P^*_b(G) \cdot n} \sum_{(x^*_i, y^*_i) \in G} L^*_i - \hat{\theta}(\cD^*_b)$;
    \EndFor
    \State Define the asymptotic variance estimator by \begin{align*}
        \hat{\sigma}^2_G &\defeq \widehat{\text{Var}}(L \mid G) + \P_n(G) \l (\widehat{\text{Var}}(\psi) - 2 \cdot \widehat{\text{Cov}}(\,L\,, \psi \mid G) \r )
    \end{align*}
    where $\widehat{\text{Var}}(\cdot)$ and $\widehat{\text{Cov}}(\cdot)$ correspond to the sample (conditional) variance and covariance;
    \State Define $\hat{s}(G)$ by \eqref{eqn:rescaling_est};
    \For{$b = 1,\dots,B$}
        \State $t^{(b)} = \max_{G \in \cG} \{\frac{1}{\hat{s}(G)} \cdot \P_n(G) \cdot \P^*_b(G) \cdot (\epsilon^*_b(G) - \hat{\epsilon}(G)) \}$;
    \EndFor

    \State \textbf{Return:} $t^* = \text{Quantile} (1 - \alpha; \{t^{(b)} \}_{b = 1}^B )$
  \end{algorithmic}
\end{algorithm}
\begin{theorem} \label{thm:uniform_ci_appendix}
Assume that $L$ is bounded and that $L - \hat{\theta}$ is non-constant over at least one non-empty group. Further assume that $\cG$ has finite Vapnik-Chernovenkis (VC) dimension. Then, $\epsilon_{\textup{lb}}(G)$ is an asymptotic $(1 - \alpha)$-lower confidence bound for $\epsilon(G)$ that is simultaneously valid for all $G \in \cG$.
\end{theorem}
\begin{proof}
    To simplify notation, we will replace $\indic{(X, Y) \in G}$ with $\indG$ throughout. We first restate the result we aim to prove.
    \begin{align*}
        \lim_{n \to \infty} \P \l(\exists\, G \in \cG \text{ s.t. } \epsilon_\text{lb}(G) > \epsilon(G) \r) &= \lim_{n \to \infty} \P \l(\exists\, G \in \cG \text{ s.t. } \hat{\epsilon}(G) - t^* \cdot \frac{ s^*(G) }{\P_n(G)^2} > \epsilon(G) \r)
    \end{align*}
    Rearranging the latter event, we obtain
    \begin{multline*}
        \lim_{n \to \infty}\P \l(\sup_{G \in \cG} \l\{ \l(\frac{\P_n(G)^2}{\hat{s}(G)} \r) \cdot (\hat{\epsilon}(G) - \epsilon(G)) \r\} > t^* \r) \\
        = \lim_{n \to \infty} \P \l(\sup_{G \in \cG} \l\{ \l(\frac{\P_n(G)}{\hat{s}(G)} \r) \cdot \P_n[(L - \hat{\theta} - \epsilon(G)) \cdot \indG] \r\} > t^* \r).
    \end{multline*}
    Then, we claim that $\sup_{G \in \cG} | \P_n(G)/\hat{s}(G) - \P(G)/s(G) | \cp 0$. This follows by observing that
    \begin{align*}
        \sup_{G \in \cG} \l| \frac{\P_n(G)}{\hat{s}(G)} - \frac{\P(G)}{s(G)} \r| &\leq \sup_{G \in \cG} \l| \frac{\P_n(G)}{\hat{s}(G)} - \frac{\P_n(G)}{s(G)} \r| + \sup_{G \in \cG} \l| \frac{\P_n(G)}{s(G)} - \frac{\P(G)}{s(G)} \r| \\&\leq \sup_{G \in \cG} |\P_n(G)| \cdot \sup_{G \in \cG} \l| \frac{1}{\hat{s}(G)} - \frac{1}{s(G)}\r|  + \l|\frac{1}{\inf_{G \in \cG} s(G)} \r| \cdot \sup_{G \in \cG} |\P_n(G) - \P(G)|.
    \end{align*}
    Observe that $\sup_{G \in \cG} |\P_n(G)| \leq 1$ and $\inf_{G \in \cG} s(G) \geq p_*^{3/2} \cdot (w_0 / (1 + w_0)) \cdot \sqrt{\text{Var}(L)} > 0$. Then, by \Cref{lma:rescaling_consistency} and the assumption that $\cG$ is $P$-Donsker, and thus also $P$-Glivenko-Cantelli, we conclude the desired uniform consistency result.

    To apply Slutsky's lemma to the uniformly consistent estimator, we must prove that $\sqrt{n} \cdot \P_n[(L - \hat{\theta} - \epsilon(G)) \cdot \indG] = O_P(1)$. Note that
    \begin{align*}
        &\left |\sup_{G \in \cG} \sqrt{n} \cdot \P_n[(L - \hat{\theta} - \epsilon(G)) \cdot \indG] \right|\\
        &\qquad = \left| \sup_{G \in \cG} \l \{\sqrt{n} \cdot (\P_n - P)[(L - \theta_P - \epsilon(G)) \cdot \indG] - \sqrt{n} (\hat{\theta} - \theta_P) \cdot \P_n [\indG] \r\}\right| \\
        &\qquad \leq \left| \sup_{G \in \cG} \sqrt{n} (\P_n - P)[(L - \theta_P - \epsilon(G)) \cdot \indG] \right| + \l |\sqrt{n} (\hat{\theta} - \theta_P) \r | \cdot \sup_{G \in \cG} \P_n(G).
    \end{align*}
    
    By standard Donsker preservation results \citep[Section 2.10]{vaart1996weak}, the function class $\{(L - \theta_P - \epsilon(G)) \cdot \indG \mid G \in \cG\}$ is $P$-Donsker and hence $\sup_{G \in \cG} \sqrt{n}(\P_n - P)[(L - \theta_P - \epsilon(G)) \mathbf{1}_G] = O_P(1)$.  Meanwhile, $| \sqrt{n}(\hat{\theta} - \theta_P)| \cdot \sup_{G \in \cG}] \P_n(G) \leq | \sqrt{n}( \hat{\theta} - \theta_P)| = O_P(1)$ by assumption.  Thus, this upper bound is $O_P(1)$.
    
    
    Then, applying Slutsky's lemma, we obtain
    \begin{align*}
        &\l(\frac{\P_n[\indG]}{\hat{s}(G)} \r) \cdot \sqrt{n} \cdot \P_n \l[\l(L - \hat{\theta} - \epsilon(G)\r) \indG \r]  \\
        &\quad = \frac{1}{s(G)} \cdot \sqrt{n} \cdot \P_n \l[\l((L - \hat{\theta})P[\indG] - P[(L - \theta_P) \indG]\r) \indG    \r] + o_P(1) \\
        &\quad = \frac{1}{s(G)} \cdot \sqrt{n} \left(\P_n [(L - \hat{\theta}) \indG] \cdot P[\indG] - P[(L - \theta_P) \indG] \cdot \P_n[\indG]\right) + o_P(1).
    \end{align*}

    Our goal is to now prove that the bootstrap analogue to this process is consistent, i.e., we must show that this process is indexed by a $P$-Donsker class. To this end, observe that we can rewrite the process as
    \begin{multline*}
        \frac{P[\indG]}{s(G)} \cdot \sqrt{n} \left (\P_n[(L - \hat{\theta}) \cdot \indG] - P[(L - \theta_P) \cdot \indG] \right) - \frac{P[(L - \theta_P)\indG]}{s(G)} \cdot \sqrt{n}(\P_n - P)[\indG] + o_P(1)
    \end{multline*}
    We then linearize the first term by observing that 
    \begin{multline*}
        \sqrt{n} \left(\P_n[(L - \hat{\theta}) \cdot \indG] - P[(L - \theta_P) \cdot \indG]\right) = \sqrt{n}(\P_n - P)[L \cdot \indG] - \theta_P \cdot \sqrt{n}(\P_n - P)[\indG] \\
        - P[\indG] \cdot \sqrt{n}(\hat{\theta} - \theta_P) + \underbrace{\frac{1}{\sqrt{n}} \cdot  \sqrt{n}(\P_n - P)[\indG] \cdot \sqrt{n} (\hat{\theta} - \theta_P )}_{o_P(1)}.
    \end{multline*}
    Replacing this in the previous expansion, we simplify and obtain
    \begin{multline*}
         \sqrt{n} (\P_n- P)\left [\frac{P[\indG]}{s(G)} \cdot L \cdot \indG\right] - \sqrt{n}(\P_n - P) \left[\frac{P[L \cdot \indG]}{s(G)} \cdot \indG \right] -\sqrt{n}(\P_n - P) \left[ \frac{P[\indG]^2}{s(G)} \cdot \psi \right ] \\
         + o_P(1).
    \end{multline*}

    We claim that the three empirical processes in this display are all indexed by Donsker classes. Using the first process as an example, we observe that $\{P[\indG]/s(G) \mid G \in \cG\}$ is a universal, uniformly bounded Donsker class, while $L \cdot \indG$ is a $P$-Donsker class because the product of a bounded measurable function and a Donsker class is Donsker. The pairwise product of two uniformly bounded Donsker classes is Donsker. Clearly, $\{P[\indG]/s(G) \cdot L \cdot \indG \mid G \in \cG\}$ is the subset of such a product, so by Theorem~2.10.1 in \cite{vaart1996weak}, the first process is indexed by a Donsker class. Near-identical arguments prove the last two processes are also $P$-Donsker. Recall that the Donsker property is preserved under addition. So, we conclude that
    \begin{align}
        \sqrt{n}(\P_n - P)\left[\frac{P[\indG]}{s(G)} \cdot L \cdot \indG -  \frac{P[L \cdot \indG]}{s(G)} \cdot \indG - \frac{P[\indG]^2}{s(G)} \cdot \psi\right] + o_P(1) \label{eqn:donsker_ci_process}
    \end{align}
    is indexed by a subset of a Donsker class, i.e., a Donsker class. 

    Via an identical derivation, we also show that the bootstrap process in \textsf{\Cref{algo:bootstrap_ci_appendix}}, $\frac{1}{\hat{s}(G)} \cdot \P_n(G) \cdot \P^*_n(G) \cdot (\epsilon^*(G) - \hat{\epsilon}(G))$ is equal to
    \begin{align*}
        \sqrt{n}(\P^*_n - \P_n)\left[\frac{P[\indG]}{s(G)} \cdot L \cdot \indG -  \frac{P[L \cdot \indG]}{s(G)} \cdot \indG - \frac{P[\indG]^2}{s(G)} \cdot \psi\right] + o_P(1).
    \end{align*}
    Applying \Cref{lma:donsker_bootstrap}, we conclude that the bootstrap is consistent. Due to the continuous mapping theorem, we may take a $\sup$ over $G \in \cG$ and conclude that 
    the process sampled in \textsf{\Cref{algo:bootstrap_ci_appendix}} is asymptotically equivalent to the process stated in the probability of interest, 
    \begin{align*}
        \lim_{n \to \infty} \P \l(\sup_{G \in \cG} \l\{ \sqrt{n} \cdot \l(\frac{\P_n(G)}{\hat{s}(G)} \r) \cdot \P_n[(L - \hat{\theta} - \epsilon(G)) \cdot \indG] \r\} \geq t^* \cdot \sqrt{n} \r).
    \end{align*}

    Last, we must show that $t^* \cdot \sqrt{n}$ consistently estimates the $(1 - \alpha)$-quantile of the supremum of \eqref{eqn:donsker_ci_process}. Recall that \citet[Lemma~11.2.1(ii)]{lehmann2005testing} establishes quantile consistency whenever the distribution function is continuous and strictly increasing at the point of interest. Our assumption that $L - \hat{\theta}$ is non-constant for some non-empty $G$ implies that the asymptotic variance of $\frac{1}{s(G)} \cdot \P_n(G) \cdot \P(G) \cdot (\hat{\epsilon}(G) - \epsilon(G))$ is non-zero for some $G$. Then, the $(1 - \alpha)$-quantile of the supremized process is strictly greater than $0$ and \Cref{cor:quantile_gaussian_2} implies the desired result.

    Completing the proof,
    \begin{multline*}
        \lim_{n \to \infty} \P \l(\exists\, G \in \cG \text{ s.t. } \epsilon_\text{lb}(G) > \epsilon(G) \r) \\ = \lim_{n \to \infty} \left( \sup_{G \in \cG} \l \{ \frac{1}{\hat{s}(G)} \cdot \P_n(G) \cdot \P^*_n(G) \cdot (\epsilon^*(G) - \hat{\epsilon}(G)) \r\} > t^* \right) = \alpha.
    \end{multline*}
\end{proof}

\paragraph{Boolean certification}
We next consider Boolean certification, in which we test the null $$H_0(G): \epsilon(G) \leq \epsilon.$$

In \textsf{\Cref{algo:bootstrap_fwer_appendix}}, we define a rescaled variant of \textsf{\Cref{algo:bootstrap_fwer}}.  Recall that we issue a certificate for $G$ when
\begin{align*}
    \hat{\epsilon}(G) \geq \epsilon + \frac{t^*}{\P_n(G)}.
\end{align*}
\Cref{thm:fwer_appendix} then proves the issued certificates are simultaneously valid with probability $1 - \alpha$. 

\begin{algorithm}
  \caption{Bootstrapping the Boolean certificate critical value (rescaled)}
  \label{algo:bootstrap_fwer_appendix}
  \begin{algorithmic}[1]
    \State \textbf{Input:} Subpopulations $\cG$, disparity $\epsilon$, audit trail $\cD$, level $\alpha$, threshold $p_*$, weight $w_0$, number of bootstrap samples $B$
    \For{$b = 1,\dots,B$}
        \State Let $\cD^*_b$ be a sample with replacement of size $n$ from $\cD$;
        \State Define $\P^*_b(G) \defeq \frac{1}{n} \sum_{(x^*_i, y^*_i) \in \cD^*_b} \indic{(x^*_i, y^*_i) \in G}$;
        \State Define $\epsilon^*_b(G) \defeq \frac{1}{\P^*_b(G) \cdot n} \sum_{(x^*_i, y^*_i) \in G} L^*_i - \hat{\theta}(\cD^*_b)$;
    \EndFor
    \State Define the asymptotic variance estimator by \begin{align*}
        \hat{\sigma}^2_G &\defeq \widehat{\text{Var}}(L \mid G) + \P_n(G) \l (\widehat{\text{Var}}(\psi) - 2 \cdot \widehat{\text{Cov}}(\,L\,, \psi \mid G) \r )
    \end{align*}
    where $\widehat{\text{Var}}(\cdot)$ and $\widehat{\text{Cov}}(\cdot)$ correspond to the sample (conditional) variance and covariance;
    \State Define $\hat{s}(G)$ by \eqref{eqn:rescaling_est};
    \For{$b = 1,\dots,B$}
        \State $t^{(b)} = \max_{G \in \cG} \l \{ \frac{1}{\hat{s}(G)} \cdot \l(\P^*_b(G) \cdot (\epsilon^*_b(G) - \epsilon) - \P_n(G) \cdot (\hat{\epsilon}(G) - \epsilon)\r) \r \}$;
    \EndFor
    \State \textbf{Return:} $t^* = \text{Quantile} (1 - \alpha; \{ t^{(b)} \}_{b = 1}^B )$
  \end{algorithmic}
\end{algorithm}

\begin{theorem} \label{thm:fwer_appendix}
Assume that $L$ has finite variance and that $L - \hat{\theta}$ is non-constant over at least one non-empty group. Further assume that $\cG$ has finite VC dimension. Then, 
\begin{align*}
    \lim_{n \to \infty} \P(\textup{there exists any falsely certified $G \in \cG$}) \leq \alpha.
\end{align*}
\end{theorem}

\begin{proof}
    To simplify notation, we will replace $\indic{(X, Y) \in G}$ with $\indG$ throughout. 
    
    First, we restate the result we need to prove. 
    \begin{align*}
        \lim_{n \to \infty} \P \l(\sup_{G \in \{G : H_0(G) \text{ holds}\}} \hat{\epsilon}(G) \geq \epsilon + t^* \cdot \frac{\hat{s}(G)}{\P_n(G)} \r) &\leq \alpha.
    \end{align*}
    Rearranging, we obtain
    \begin{multline*}
        \lim_{n \to \infty} \P \l(\sup_{G \in \{G : H_0(G) \text{ holds}\}} \hat{\epsilon}(G) \geq \epsilon + t^* \cdot \frac{\hat{s}(G)}{\P_n(G)} \r) \\ = \lim_{n \to \infty} \P \l(\sup_{G \in \{G : H_0(G) \text{ holds}\}} \frac{1}{\hat{s}(G)} \cdot \P_n(G) \cdot (\hat{\epsilon}(G) - \epsilon) \geq t^* \r).
    \end{multline*}
    Under $H_0(G)$, by assumption, $\P(G) \cdot (\epsilon(G) - \epsilon) \leq 0$, so it follows that
    \begin{align*}
        &\lim_{n \to \infty} \P \l(\sup_{G \in \{G : H_0(G) \text{ holds}\}} \l \{\frac{1}{\hat{s}(G)}\cdot \P_n(G) \cdot (\hat{\epsilon}(G) - \epsilon) \r \} \geq t^* \r) \\
        &\quad \leq \lim_{n \to \infty} \P \l(\sup_{G \in \{\cG : H_0(G) \text{ holds}\}}   \l \{\frac{1}{\hat{s}(G)} \cdot \l(\P_n(G) \cdot (\hat{\epsilon}(G) - \epsilon) - \P(G) \cdot (\epsilon(G) - \epsilon)\r) \r \} \geq t^* \r) \\
        &\quad \leq \P \l(\sup_{G \in \cG} \l \{\frac{1}{\hat{s}(G)} \cdot \l(\P_n(G) \cdot (\hat{\epsilon}(G) - \epsilon) - \P(G) \cdot (\epsilon(G) - \epsilon) \r) \r \} \geq t^* \r)
    \end{align*}
    For $B = \infty$, recall that $t^*$ is the $(1 - \alpha)$-quantile of $\sup_{G \in \cG} \{(1/\hat{s}(G)) \cdot (\P^*_n(G) \cdot (\epsilon^*(G) - \epsilon) - \P_n(G) \cdot (\hat{\epsilon}(G) - \epsilon))\}$. We must show that $t^*$ is consistent for the $(1 - \alpha)$-quantile of $\sup_{G \in \cG} \{(1/\hat{s}(G)) \cdot (\P_n(G) \cdot (\hat{\epsilon}(G) - \epsilon) - \P(G) \cdot (\epsilon(G) - \epsilon))\}$.

    Since we established that $1/\hat{s}(G)$ is uniformly consistent for $1/s(G)$ in the proof of \Cref{thm:uniform_ci_appendix}, we can apply Slutsky's lemma and replace $1/\hat{s}(G)$ with $\frac{1}{s(G)}$ so long as $\P_n(G) \cdot (\hat{\epsilon}(G) - \epsilon) - \P(G) \cdot (\epsilon(G) - \epsilon) $ is $O_P(1)$. We establish this by showing that it is asymptotically equivalent to an empirical process indexed by some $P$-Donsker function class. To this end, we rewrite $\sqrt{n} [\P_n(G) \cdot (\hat{\epsilon}(G) - \epsilon) - \P(G) \cdot (\epsilon(G) - \epsilon)]$ as
    \begin{align*}
       \sqrt{n} (\P_n - P)[(L - \epsilon) \cdot \indG] - \sqrt{n} \left(\hat{\theta} \cdot \P_n [\indG] - \theta_P \cdot P[\indG] \right) .
    \end{align*}
    We can further expand $\sqrt{n} \left(\hat{\theta} \cdot \P_n [\indG] - \theta_P \cdot P[\indG] \right)$ to obtain
    \begin{align*}
        \theta_P \cdot  \l(\sqrt{n} (\P_n - P) [\indG]\r) + \P(G) \cdot \l( \sqrt{n}(\hat{\theta} - \theta_P) \r) 
        + \underbrace{\frac{1}{\sqrt{n}} \l(\sqrt{n} (\P_n - P) [\indG] \r) \l( \sqrt{n}(\hat{\theta} - \theta_P)\r)}_{o_P(1)}.
    \end{align*}
    Combining these results and applying the asymptotic linearity of $\sqrt{n}(\hat{\theta} - \theta_P)$, we can now linearize $\sqrt{n}\l[\P_n(G) \cdot (\hat{\epsilon}(G) - \epsilon) - \P(G) \cdot (\epsilon(G) - \epsilon)\r]$ to yield the (asymptotically) equivalent process,
    \begin{align*}
        \sqrt{n} (\P_n - P)  [(L - \epsilon) \cdot \indG] 
        - \sqrt{n} (\P_n - P) [\theta_P \cdot \indG]
        - \sqrt{n}(\P_n - P) [\P(G) \cdot \psi] + o_P(1).
    \end{align*}
    Proceeding identically, we can also linearize the bootstrap analogue to this process, i.e.,
    \begin{multline*}
        \sqrt{n} (\P^*_n - \P_n) [(L - \epsilon) \cdot \indG] 
        - \sqrt{n} (\P^*_n - \P_n) [\theta_P  \cdot \indG]
        - \sqrt{n}(\P^*_n - \P_n) [ \P(G) \cdot \psi] + o_{P_n}(1).
    \end{multline*}
    The three terms are empirical or bootstrap empirical processes, respectively, indexed by the following function classes:
    \begin{align*}
        \cF_1 \defeq \{(L - \epsilon) \cdot \indG \mid G \in \cG\} \qquad 
        \cF_2 \defeq \{\theta_P \cdot \indG \mid G \in \cG\} \qquad 
        \cF_3 \defeq \{\P(G) \cdot \psi \mid G \in \cG\}.
    \end{align*}
    Applying Lemma~9.9(vi) of \cite{kosorok2008introduction} with $g(x, y) = L - \epsilon$ and $\cF = \{\indic{x \in G} : G \in \cG\}$ implies that $\cF_1$ is VC. Then, our assumption that $\E_P[L^2] < \infty$ implies that $\cH_1$ is $P$-Donsker (Theorem~2.10.20 in \cite{vaart1996weak}). $\cF_2$ is $P$-Donsker by \Cref{lma:vc_iff}. Last, observe that the uniform entropy integral of $\cF = \{\P(G) \mid G \in \cG\}$ is finite. By the same argument as $\cF_1$, we conclude from Theorem~2.10.20 in \cite{vaart1996weak} that $\cF_3$ is $P$-Donsker.

    Since the Donsker property is preserved under pointwise addition (Example~2.10.7 in \cite{vaart1996weak}) and when taking subsets (Theorem~2.10.1 in \cite{vaart1996weak}), we conclude that $\cF_1 + \cF_2 + \cF_3$ is also $P$-Donsker. Thus, applying Slutsky's lemma and \Cref{lma:donsker_bootstrap}, we have shown that that the bootstrap distribution $\sqrt{n}\l[(1/\hat{s}(G)) \cdot (\P^*_n(G) \cdot (\epsilon^*(G) - \epsilon) - \P_n(G) \cdot (\hat{\epsilon}(G) - \epsilon))\r]$ consistently estimates the limiting distribution of $\sqrt{n}\l[(1/\hat{s}(G)) \cdot (\P_n(G) \cdot (\hat{\epsilon}(G) - \epsilon) - \P(G) \cdot (\epsilon(G) - \epsilon))\r]$. Since the $\sup$ is continuous, we can immediately claim by the continuous mapping theorem that 
    \begin{align*}
        \sup_{G \in \cG} \sqrt{n}\l[(1/\hat{s}(G)) \cdot (\P^*_n(G) \cdot (\epsilon^*(G) - \epsilon) - \P_n(G) \cdot (\hat{\epsilon}(G) - \epsilon))\r]
    \end{align*}
    consistently estimates its limiting distribution.

    To conclude that $t^*$ is a valid critical value, we need to prove that the distribution function of the limiting distribution of $\sup_{G \in \cG} \{\sqrt{n}\l[(1/\hat{s}(G)) \cdot (\P_n(G) \cdot (\hat{\epsilon}(G) - \epsilon) - \P(G) \cdot (\epsilon(G) - \epsilon))\r]\}$ is strictly increasing and continuous at its $(1 - \alpha)$-quantile. 

    The existence of some non-empty $G$ such that $L - \hat{\theta}$ is non-constant implies this fact. In particular, this assumption implies that the asymptotic variance of the process evaluated at $G$ is greater than $0$. Thus, the $(1 - \alpha)$-quantile of the limiting process is strictly greater than $0$. Then, \Cref{cor:quantile_gaussian_2} yields the desired result.
    
    Summarizing our argument, we have proven that the asymptotic probability of false certification can be upper bounded in the following manner,
    \begin{multline*}
        \lim_{n \to \infty} \P \l(\sup_{G \in \{G : H_0(G) \text{ holds}\}} \hat{\epsilon}(G) \geq \epsilon + t^* \cdot \frac{\hat{s}(G)}{\P_n(G)} \r) \\
        \leq \lim_{n \to \infty} \P \l(\sup_{G \in \cG} \l\{ \frac{1}{\hat{s}(G)} \cdot (\P_n(G) \cdot (\hat{\epsilon}(G) - \epsilon) - \P(G) \cdot (\epsilon(G) - \epsilon)) \r\} \geq t^* \r) 
        \\ 
        = \lim_{n \to \infty} \P \l(\sup_{G \in \cG} \l\{ \frac{1}{\hat{s}(G)} \cdot (\P^*_n(G) \cdot (\epsilon^*(G) - \epsilon) - \P_n(G) \cdot (\hat{\epsilon}(G) - \epsilon)) \r\} \geq t^* \r) = \alpha.
    \end{multline*}
\end{proof}

\Cref{thm:fwer_main_text} follows immediately from \Cref{thm:fwer_appendix} when $\hat{s}(G) \defeq 1$. We remark also that any looseness in our upper bound of the Type I error disappears if $\epsilon(G) = \epsilon$ for all $G$. As alluded to in the main text, this corresponds to the condition under which $t^*$ achieves exact Type I error control.

\subsection{Alternative certification goals} \label{sec:certify_alt_nulls}
Here, we describe how to extend the certification procedures described in the previous section and the main text to alternative notions of performance disparities. 

When constructing \emph{upper} confidence bounds or certifying $\epsilon(G) < \epsilon$, we simply multiply $t^*$ by $-1$. So, now $\epsilon_{\text{ub}}(G) = \hat{\epsilon}(G) + t^* \hat{s}(G)/\P_n(G)^2$,
and we certify when $\hat{\epsilon}(G) \leq \epsilon - t^*/\P_n(G)$.

Observe that the argument for validity of the upper confidence bound / certificate goes through if we replace the $(1 - \alpha)$-quantile of the $\sup$-process with the $\alpha$-quantile of the $\inf$-process. The latter, however, is just $-1$ times the former. Recall that the infimum of a centered Gaussian process is equal in distribution to $-1$ times the supremum of that process.

Next, we consider the problem of constructing confidence intervals. We propose to bootstrap the absolute process,
\begin{align*}
    \sup_{G \in \cG} \l|\frac{1}{\hat{s}(G)} \cdot \P(G) \cdot \P_n(G) \cdot (\hat{\epsilon}(G) - \epsilon(G)) \r|.
\end{align*}
Then, if $t^*$ denotes the bootstrap estimate of the $(1 - \alpha)$-quantile of this process, the confidence set is constructed via
\begin{align*}
    \l[\hat{\epsilon}(G) - t^* \cdot \frac{\hat{s}(G)}{\P_n(G)^2}, \hat{\epsilon}(G) + t^* \cdot \frac{\hat{s}(G)}{\P_n(G)^2}\r ].
\end{align*}

Last, for the interval certification task, we test the null hypothesis $H_0(G) : |\epsilon(G)| \geq \epsilon$.  Equivalently, we test:
\begin{align*}
    \bar{H}_0(G) &: \l \{\epsilon(G) \geq \epsilon \r\} \bigcup \l \{\epsilon(G) \leq - \epsilon \r\}.
\end{align*}
This is also known as a ``bioequivalence'' null and it can be tested by running the one-sided Boolean certification procedures for $H_0(G): \epsilon(G) \geq \epsilon$ and $H_0(G): \epsilon(G) \leq - \epsilon$ and certifying when both are rejected. Formally, observe that we would test the first null at level $\alpha$ by rejecting when $\hat{\epsilon}(G) \leq \epsilon - t^*_1/\P_n(G)$ and the second null at level $\alpha$ by rejecting when $\hat{\epsilon}(G) \geq -\epsilon + t^*_2/\P_n(G)$.  Then, we certify $G$ as satisfying $|\epsilon(G)| < \epsilon$ whenever both of these inequalities simultaneously hold.

\section{Flagging audits} \label{sec:flagging_appendix}
To prove that our flagging procedure is asymptotically valid, we first establish that the proposed test statistic is asymptotically normal.
Recall that $\hat{\epsilon}(G) \defeq |G|^{-1} \sum_{i \in G} L_i - \hat{\theta}$.  We denote the vector of statistics $\{ \hat{\epsilon}(G) \}_{G \in \cG}$ as $\hat{\epsilon}$, and the vector of corresponding true disparities as $\epsilon$.

\begin{proposition} \label{prop:clt_ipw}
Assume that $\P(G)$ and $\var\l(L \mid G \r)$ are bounded away from $0$ for all $G \in \cG$. If $m =|\cG| < \infty$, then $\sqrt{n} \l( \hat{\epsilon} - \epsilon \r) \leadsto\mathcal{N}_m \l(\mathbf{0}_m, \Sigma \r)$.
\end{proposition}
\begin{proof}
Apply the CLT to the asymptotic linear expansion derived in \Cref{lma:rescaling_consistency}.
\end{proof}

For any group $G_j \in \cG$, \Cref{prop:clt_ipw} implies that the p-values
\begin{gather*}
    p_1(G_j; \epsilon) \defeq \Phi \l( \frac{\sqrt{n} \l(\hat{\epsilon}(G_j) - \epsilon \r)}{\sqrt{\Sigma_{jj}}} \r), \quad p_2(G_j; \epsilon) \defeq 1 - \Phi \l( \frac{\sqrt{n} \l(\hat{\epsilon}(G_j) + \epsilon \r)}{\sqrt{\Sigma_{jj}}} \r), \\
    p_3(G_j) \defeq p_1(G_j; -\epsilon) \wedge p_2(G_j; -\epsilon)
\end{gather*}
are marginally asymptotically valid for the null hypotheses $H_1(G_j) : \epsilon(G_j) \geq \epsilon$, $H_2(G_j) : \epsilon(G_j) \leq - \epsilon$,  $H_3(G_j) : |\epsilon(G_j)| \leq \epsilon$, respectively, i.e., $\limsup_{n \to \infty} \P_{H_i}( p_i (G_j) \leq u) \leq u$.

$\Sigma_{jj}$ is not known, but by Slutsky's lemma, we can replace $\Sigma_{jj}$ with any consistent estimator. We can compute such an estimator analytically, but below we rely on the bootstrap to construct such an estimator. In particular, let
\begin{align*}
    s^*_j \defeq \textup{Quantile} \l( 0.5, \sqrt{n} \l|\epsilon^*(G) - \hat{\epsilon}(G) \r| \r).
\end{align*}
Then, \Cref{lma:consistency_mad} establishes conditions under which $s^*_j$ is a consistent estimator of $\sqrt{\Sigma_{jj}}$.
\begin{lemma}\label{lma:consistency_mad}
Retain the assumptions of \Cref{prop:clt_ipw}. 
Then, $s^*_j/\Phi^{-1}(3/4) \cp \sqrt{\Sigma_{jj}}$.
\end{lemma}
\begin{proof}
Applying the bootstrap delta method and continuous mapping theorem we can conclude that the bootstrap distribution, $\sqrt{n} \l|\epsilon^*(G_j) - \hat{\epsilon}(G_j) \r|$
consistently estimates the distribution of $\sqrt{n} \l|\hat{\epsilon}(G_j) - \epsilon(G_j) \r|$.  Since the limiting distribution of the latter is continuous and strictly increasing everywhere (\Cref{prop:clt_ipw}), we can apply \citet[Lemma 5.10]{van2000asymptotic} to conclude that the bootstrap estimate of the median absolute deviation is consistent for the asymptotic median absolute deviation. Then, the result follows by recalling the well-known fact that the median absolute deviation of a Gaussian distribution with variance $\sigma^2$ is equal to $\Phi^{-1}(3/4) \cdot \sigma$.
\end{proof}

\Cref{lma:consistency_mad} thus implies that we can replace $\sqrt{\Sigma_{jj}}$ in the p-value definitions with the Monte Carlo estimate given by $s^*_j$. 

With these results, we can prove \Cref{prop:fdr_control_bh} after recalling some definitions and well-known conditions regarding the validity of the \textsf{BH} procedure. 

\begin{definition}
We say that $X$ has \underline{positive regression dependency on a subset $I_0$} (PRDS on $I_0$) if for any increasing set $D$ and for each $i \in I_0$, $\P(X \in D \mid X_i = x)$ is increasing in $x$.
\end{definition}

\begin{example}[Case 3.1 in \cite{benjamini2001control}]
The one-sided Gaussian p-values obtained by testing $H_0: \mu \leq \mu^*$ or $H_0: \mu \geq \mu^*$, using $T \sim \mathcal{N}(\mu, \Sigma)$ are PRDS on $I_0$ if $\Sigma_{ij} \geq 0$ for all $i \in I_0$ and $j \in [m]$.
\end{example}

\begin{theorem}[Theorem~1.2 in \cite{benjamini2001control}]
If the joint distribution of the test statistics is PRDS on the subset of test statistics corresponding to the $m_0$ true null hypotheses, the \textsf{BH}($\alpha$) procedure controls the FDR at level less than or equal to $\frac{m_0}{m} \alpha$. \label{thm:bh_prds}
\end{theorem}

For the reader's convenience, we restate \Cref{prop:fdr_control_bh} before completing its proof.

\begin{proposition}
Assume that $\P((X, Y) \in G)$ and $\var \l(L(X, Y) \mid (X, Y) \in G \r)$ is bounded away from $0$ for all $G \in \cG$, $\theta_P$ is a-priori known, and that one of the following conditions holds:
\begin{enumerate}[label=(\roman*), itemsep=-0.5ex]
    \item $\{G\}_{G \in \cG}$ are mutually disjoint;
    \item $L$ takes values in $\{0,1\}$.
\end{enumerate}
If we flag the rejections of the \textsf{BH}($\alpha$) procedure on $\{p(G)\}_{G \in \cG}$, then the false discovery rate is asymptotically controlled at level $\alpha$.
\end{proposition}
\begin{proof}
For clearer indexing, we let $\cG = \{G_j\}_{j = 1}^m$. We assume w.l.o.g. that $\theta_P = 0$. Then, \Cref{prop:clt_ipw} shows that $\sqrt{n} \l \{\hat{\epsilon}(G_j)  - \epsilon(G_j) \r\}_{j = 1}^m \leadsto \mathcal{N} \l(\mathbf{0}_{m}, \Sigma \r)$.  The off-diagonal entries of $\Sigma$ equal
\begin{align*}
    \frac{\E \l[(L - \epsilon(G_j)) (L - \epsilon(G_k)) \indic{G_j \cap G_k} \r]}{\P(G_j) \P(G_k)}.
\end{align*}

Under condition (i), the off-diagonal entries are all $0$ since the indicator in the numerator always equals $0$.

Consider condition (ii). In this setting, we can rewrite the covariance expression as
\begin{align*}
    \l(\epsilon(G_j \cap G_k) \l(1 - \epsilon(G_j) - \epsilon(G_k) \r) + \epsilon(G_j) \epsilon(G_k) \r)\frac{\P( G_j \cap G_k)}{\P( G_j) \P(G_k)}.
\end{align*}
Then, assume for the sake of contradiction that this expression is negative. Then,
\begin{align*}
    0 &> \epsilon(G_j \cap G_k) \l(1 - \epsilon(G_j) - \epsilon(G_k) \r) + \epsilon(G_j) \epsilon(G_k) \\
    &\geq 1 - \epsilon(G_j) - \epsilon(G_k) + \epsilon(G_j) \epsilon(G_k) \\
    &= (1 - \epsilon(G_j)) (1 - \epsilon(G_k)).
\end{align*}
But since the last expression is non-negative under condition (ii), we obtain a contradiction and must conclude that the asymptotic covariance is non-negative.

In both cases, we showed that the asymptotic distribution of these test statistics is a multivariate Gaussian with non-negative covariance. As a consequence, the one-sided p-values output by \textsf{\Cref{algo:flag_p_values}} are asymptotically PRDS on the set of nulls. Then, applying the bounded convergence theorem and \Cref{thm:bh_prds} yields the desired result.
\end{proof}

While we do not have a proof of FDR control when testing the two-sided null, it is widely speculated that the \textsf{BH} procedure enjoys FDR control in this setting \cite{fithian2022conditional}. Moreover, even when the PRDS property does not hold (i.e., outside of the conditions given in \Cref{prop:fdr_control_bh}), it is, in practice, quite challenging to obtain substantial violations of FDR control when applying the \textsf{BH} procedure to asymptotic Gaussian p-values.

\section{Auditing distribution shifts} \label{sec:rkhs_appendix}
Before considering the problem of auditing over shifts belonging to the unit ball of an RKHS, we observe that neither the proof of \Cref{thm:uniform_ci_appendix} nor the proof of \Cref{thm:fwer_appendix} rely on $\cG$ being a VC class. Since $\cG$ is VC, $\cF = \{\indic{(X, Y) \in G} \mid G \in \cG\}$ is a $P$-Donsker class. This implies that the proofs employed above are valid even if we replace $\cG$ with some generic Donsker function class.

Extending our results to the RKHS setting, however, require some additional work. To prove that our method for constructing confidence sets on $\epsilon(h)$ are valid, we must prove some preliminary results regarding the unit ball of an RKHS.

\begin{lemma} \label{lma:rkhs_donsker}
    Assume that $\|k(X, X)\|_\infty$ is finite, and that $k(\cdot, x)$ is continuous. Then, the unit ball of the RKHS induced by $k$, which we denote by $\mathcal{H}_1$, is a P-Donsker class. 
\end{lemma}
\begin{proof}
    Lemma 4.28 and Lemma 4.33 of \cite{steinwart2008support} show that the assumptions are sufficient to guarantee that $\mathcal{H}_1$ is a separable Hilbert space and a subset of the space of bounded and continuous functions. The conclusion then follows from Theorem~1.1 of \cite{marcus1985relationships} with $T$ chosen to be the identity. The identity mapping is trivially linear, and also meets the assumption of continuity in the $\sup$-norm because the former is dominated by the RKHS norm.
\end{proof}

\begin{lemma} \label{lma:rkhs_donsker_factors}
Retain the assumptions of \Cref{lma:rkhs_donsker}. Then, for uniformly bounded functions $L$ and $M$, $\tilde{\mathcal{H}}_1 \defeq \{P[L \cdot h] \cdot h \mid h \in \cH_1 \}$ is a $P$-Donsker class.
\end{lemma}
\begin{proof}
    First, observe that $P[L \cdot h]$ is uniformly bounded for all $h \in \cH_1$:
    \begin{align*}
        P[L \cdot h] &\leq \|L\|_\infty \|k(X, X)\|_\infty \|h\|_{\cH} \leq \|L\|_\infty \|k(X, X)\|_\infty =: C.
    \end{align*}
    Thus, $\tilde{\cH}_1$ is a subset of $\cH_{C}$, i.e., it is a dilation of $\cH_{1}$. Since $C \cdot \cH_1$ is $P$-Donsker (Theorem~2.10.6 in \cite{vaart1996weak}) and any subset of a $P$-Donsker class is $P$-Donsker (Theorem~2.10.1 in \cite{vaart1996weak}), we conclude that $\tilde{\cH}_1$ is $P$-Donsker.
\end{proof}

Here we generalize the main theorem to include disparities that are defined relative to an estimated threshold $\hat{\theta}$; this threshold is assumed to satisfy the asymptotic linearity and bootstrap consistency assumptions stated in \Cref{sec:power_certify}.  \textsf{\Cref{algo:bootstrap_rkhs_appendix}} modifies \textsf{\Cref{algo:bootstrap_rkhs}} so that the bootstrap accounts for estimation error in $\hat{\theta}$.

\begin{algorithm}
  \caption{Bootstrapping the RKHS confidence set critical value with estimated threshold}
  \label{algo:bootstrap_rkhs_appendix}
  \begin{algorithmic}[1]
    \State \textbf{Input:} Kernel $k$, audit trail $\cD$, level $\alpha$, bootstrap samples $B$
    \State Define $\mathbf{L} \defeq \{L(f(x_i), y_i)\}_{i = 1}^n$;
    \State Define $\mathbf{K} \defeq \{k(x_i, x_j)\}_{i,j = 1}^n$;
    \For{$b = 1,\dots,B$}
        \State Sample $\mathbf{w} \sim \text{Mult}\l (n; \frac{1}{n},\dots,\frac{1}{n} \r)$;
        \State Estimate bootstrap threshold deviation $t = \frac{1}{n} \sum_{i = 1}^n (\mathbf{w}_i - 1) \cdot \psi_i$;
        \State $\mathbf{A} = \frac{1}{n^2} \l(\l(\mathbf{w} \odot \mathbf{L}\r) \mathbf{1}^\top - \mathbf{w} \mathbf{L}^\top - t \cdot \mathbb{I}_n \r)$;
        \State $t^{(b)} = \lambda_{\max} \l(\mathbf{K}^{1/2} \l(\frac{\mathbf{A} + \mathbf{A}^\top}{2}\r) \mathbf{K}^{1/2} \r)$;
    \EndFor
    \State \textbf{Return:} $t^* = \text{Quantile} (1 - \alpha/4; \{t^{(b)}\}_{b = 1}^B ) $
  \end{algorithmic}
\end{algorithm}

\begin{lemma} \label{lma:algo_trick}
    For $B = \infty$, the $t^*$ output by \textsf{\Cref{algo:bootstrap_rkhs_appendix}} equals the $(1 - \alpha)$-quantile of
    $$ 
    \sup_{h \in \cH_1} (\P^*_n - \P_n)[(L - \hat{\theta}(\cD^*)) \cdot \P_n[h] - \P_n[(L - \hat{\theta}) \cdot h]) h].
    $$
\end{lemma}

\begin{proof}
    First, we observe that the empirical process of interest is equal\footnote{To avoid repeating the same derivation, see the proof of \Cref{thm:uniform_ci_rkhs_appendix} for a full exposition of this equivalence for .} to
    \begin{align*}
        \P_n[h] \P^*_n[L \cdot h] - \P^*_n[h] \P_n[L \cdot h] - \P_n[h]^2 \cdot (\P^*_n - \P_n)[\psi] + o_P(1).
    \end{align*}
    We can rewrite the (linearized) process of interest using a multinomial variable,
    \begin{align*}
        \sup_{h \in \cH_1} \P_n[ W \cdot (L 
        \cdot \P_n[h] - \P_n[L \cdot h])h] - \P_n[h]^2 \cdot \P_n[(W - 1) \cdot \psi],
    \end{align*}
    for $W \sim \text{Mult}(n,1/n)$.
    
    We can rewrite the process in terms of inner products between the unknown function $h \in \cH$ and the kernel function,
    \begin{multline*}
        \sup_{h \in \cH_1} \l(\langle \P_n[W \cdot L \cdot k(X, \cdot)], h \rangle \langle \P_n[k(X, \cdot)], h \rangle - \langle \P_n[L \cdot k(X, \cdot)], h \rangle \langle \P_n[W \cdot k(X, \cdot)], h \rangle \r) \\
        - \P_n[(W - 1) \cdot \psi] \langle \P_n[k(X, \cdot)], h \rangle^2.
    \end{multline*}
    Since we know that the optimal $h^*$ must be of the form $\sum_{i = 1}^n \alpha_k k(\cdot, x_i)$ (recall the direct sum decomposition of any RKHS), we can rewrite the above supremum as
    \begin{align*}
        \sup_{\alpha : \alpha^\top \mathbf{K} \alpha \leq 1} \frac{1}{n^2}\l( \alpha^\top \mathbf{K} (\mathbf{w} \odot \mathbf{L}) \mathbf{1}^\top \mathbf{K} \alpha -  \alpha^\top \mathbf{K} \mathbf{w}\mathbf{L}^\top \mathbf{K} \alpha - t \cdot \alpha^T \mathbf{K} \mathbf{K} \alpha \r),
    \end{align*}
    where $\mathbf{K} = \{k(x_i, x_j)\}_{i,j=1}^n$, $\mathbf{w} = (W_1,\dots,W_n)^\top$, $\mathbf{L} = (L_1,\dots,L_n)^\top$, and $t = \P_n[(W - 1) \cdot \psi]$.
    
    Letting $\mathbf{A} = \frac{1}{n^2} [(\mathbf{w} \odot \mathbf{L}) \mathbf{1}^\top - \mathbf{w} \mathbf{L}^\top - t \mathbb{I}]$, we obtain the equivalent objectives,
    \begin{align*}
        \sup_{\alpha : \alpha^\top \mathbf{K} \alpha \leq 1} \frac{1}{2} \l(\alpha^\top \mathbf{K} \l( \mathbf{A} + \mathbf{A}^\top \r) \mathbf{K} \alpha  \r) &= \sup_{\beta : \|\beta\|_2 \leq 1} \frac{1}{2} \l(\beta^\top \mathbf{K}^{1/2} \l(\mathbf{A} + \mathbf{A}^\top \r) \mathbf{K}^{1/2} \beta \r) \\
        &= \frac{1}{2} \lambda_{\max} \l(\mathbf{K}^{1/2} \l(\mathbf{A} + \mathbf{A}^\top \r) \mathbf{K}^{1/2} \r).
    \end{align*}
\end{proof}

We remark that if $\theta_P$ is not estimated, $t = 0$ and the identity matrix in the definition of $A$ can be ignored. This greatly simplifies the computation required for the maximum eigenvalue, since $A + A^T$ is now low rank. As a consequence, we generally do not recommend using RKHS-based confidence sets for performance metrics that are defined relative to an unknown threshold.
    
With these preliminary results in hand, we can now prove our main theorem regarding RKHS confidence set validity. Given $t^*$ output by \textsf{\Cref{algo:bootstrap_rkhs_appendix}}, recall that we obtain a lower confidence bound by setting
\begin{align*}
    \epsilon_{\text{lb}}(h) \defeq \hat{\epsilon}(h) - \frac{t^*}{\l(\frac{1}{n} \sum_{i = 1}^n h(x_i)\r)^2}. 
\end{align*}
For notational convenience, we let $\cH^+_1$ denote the set of all non-negative functions belonging to $\cH_1$.
\begin{theorem} \label{thm:uniform_ci_rkhs_appendix}
Assume that $\var(L) > 0$, $\|L\|_\infty$ and $\|k(X, X)\|_\infty$ are finite, $k(\cdot, x)$ is continuous, and that $k(\cdot, \cdot)$ is a positive definite kernel. Then, 
\begin{align*}
    \lim_{n \to \infty} \P \l(\epsilon_{\textup{lb}}(h) \leq \epsilon(h) \textup{ for all $h \in \mathcal{H}_1^+$} \r) \geq 1 - \alpha.
\end{align*}
\end{theorem}
\begin{proof}
    The desired result is equivalent to
    \begin{align*}
        \lim_{n \to \infty} \P \l(\exists\,h \in \cH^+_1 \text{ s.t. }\l(\P_n[h] \cdot \P_n [h] \cdot (\hat{\epsilon}(h) - \epsilon(h) \r) > t^* \r) \leq \alpha.
    \end{align*}
    Multiplying through by $\P_n[h]$, the process on the LHS can be rewritten as
    \begin{align*}
        \P_n[h] \cdot \P_n[(L - \hat{\theta} - \epsilon(h)) \cdot h].
    \end{align*}
    
    We claim that $\P_n[h]$ is uniformly consistent for $\P[h]$. To see this, first recall that \Cref{lma:rkhs_donsker} implies that $\cH_1$ is $P$-Donsker. The Donsker property is preserved for any subset, so if $\cH_1$ is $P$-Donsker, then so is $\cH^+_1$. Uniform consistency follows from $P$-Donsker $\implies$ $P$-Glivenko-Cantelli. 

    To apply Slutsky's lemma and replace $\P_n[h]$ with $\P[h]$, we must also show that $\sqrt{n} \cdot \P_n[(L - \hat{\theta} - \epsilon(h)) \cdot h]$ is $O_P(1)$. To this end, observe that
    \begin{align*}
        \l |\sup_{h \in \cH_1^+} \sqrt{n} \cdot \P_n[(L - \hat{\theta} - \epsilon(h)) \cdot h] \r | &\leq \l |\sup_{h \in \cH_1^+} \sqrt{n} (\P_n - P)[(L - \theta_P - \epsilon(h)) \cdot h]  \r| \\
        &\quad + \l | \sup_{h \in \cH_1^+} \sqrt{n} (\P_n - P)[\psi] \cdot \P_n[h] \r| 
    \end{align*}
    We can show that both of the terms on the RHS are $O_P(1)$. First, we claim that $\tilde{\cH}_1 = \{(L - \epsilon(h))h \mid h \in \cH^+_1\}$ is $P$-Donsker. First, observe that $$\tilde{\cH}_1 \subseteq \{ L \cdot h \mid h \in \cH^+_1\} - \{\epsilon(h) \cdot h \mid h \in \cH^+_1 \}.$$ Recalling that any subset of a $P$-Donsker class is also $P$-Donsker (Theorem~2.10.1 in \cite{vaart1996weak}), we need to show that the RHS of this display is $P$-Donsker. Each function class on the RHS is $P$-Donsker. The first is because $\|L\|_\infty < \infty$ (Example~2.10.10 in \cite{vaart1996weak}). Then, the second is $P$-Donsker because it is a subset of the elementwise product of two uniformly bounded $P$-Donsker classes (Example~2.10.8 in \cite{vaart1996weak}). Last, elementwise addition of two $P$-Donsker classes yields a $P$-Donsker class (Example~2.10.7 in \cite{vaart1996weak}). Thus, we conclude that $\sqrt{n} (\P_n - P)[(L - \epsilon(h))h] = O_P(1)$. We can upper bound the second term by $|\sqrt{n} (\P_n - P)[\psi]|$, so this term is also $O_P(1)$.
    
    We now apply Slutsky's lemma and obtain:
    \begin{align*}
        P[h] \cdot \sqrt{n} \cdot \P_n[(L - \hat{\theta} - \epsilon(h)) h]
        &= \sqrt{n}\cdot (\P_n [(L - \hat{\theta}) \cdot P[h] - P[(L - \theta_P) \cdot h]) h]) \\
        &= P[h] \cdot \sqrt{n}(\P_n[(L - \hat{\theta}) \cdot h] - P[(L - \theta_P) \cdot h])  \\
        &\quad - P[(L - \theta_P) \cdot h] \cdot \sqrt{n}(\P_n - P)[h] 
    \end{align*}

    Next, we show that this process is $P$-Donsker and converges to a tight Gaussian limit. To do so, we linearize the first term of the process:
    \begin{align*}
        P[h] \cdot \sqrt{n}(\P_n[(L - \hat{\theta}) \cdot h] - P[(L - \theta_P) \cdot h]) &= P[h] \cdot \sqrt{n}(\P_n - P)[L \cdot h] \\
        &\quad - P[h] \cdot \sqrt{n}(\hat{\theta} \cdot \P_n[h] - \theta_P \cdot P[h])) \\
        &= P[h] \cdot \sqrt{n}(\P_n - P)[L \cdot h] \\
        &\quad - P[h] \cdot \theta_P \sqrt{n}(\P_n - P)[h] \\
        &\quad - P[h]^2 \cdot \sqrt{n}(\P_n - P)[\psi] \\
        &\quad + \underbrace{\frac{1}{\sqrt{n}} \left(\sqrt{n}(\hat{\theta} - \theta_P) \cdot \sqrt{n}(\P_n - P)[h]\right)}_{o_P(1)}.
    \end{align*}
    Combining both terms, we conclude that the process is equivalent to
    \begin{multline*}
        \sqrt{n}(\P_n - P)[ P[h] \cdot L \cdot h] - \sqrt{n}(\P_n - P)[P[L \cdot h] \cdot h] - \sqrt{n}(\P_n - P)[P[h]^2 \cdot \psi] \\
        = \sqrt{n}(\P_n - P)[P[h] \cdot L \cdot h - P[L \cdot h] \cdot h - P[h]^2 \cdot \psi]
    \end{multline*}

    Observe that the function class indexing this process can be written as a subset of an elementwise sum of three classes,
    \begin{multline*}
        \{P[h] \cdot L \cdot h - P[L \cdot h] \cdot h - P[h]^2 \cdot \psi \mid h \in \cH^+_1\} \\ \subseteq \{P[h] \cdot L \cdot h \mid h \in \cH^+_1\} - \{P[L \cdot h] \cdot h \mid h \in \cH^+_1\} - \{P[h]^2 \cdot \psi \mid h \in \cH^+_1\},
    \end{multline*}
    each of which is Donsker. 
    
    We check the first class is Donsker as an example; the arguments for the other two proofs follow identically. Note that $\{L \cdot h \mid h \in \cH_1^+\}$ is a $P$-Donsker class because $L$ is uniformly bounded and $\{h \mid h \in \cH_1\}$ is a uniformly bounded Donsker class (Example~2.10.10 in \cite{vaart1996weak}). Then, $\{P[h] \cdot L \cdot h \mid h \in \cH_1^+\}$ is a Donsker class because the subset of an elementwise product of two uniformly bounded Donsker classes is a Donsker class (Example~2.10.8 in \cite{vaart1996weak}).

    Thus, by the definition of a Donsker class and the continuous mapping theorem,
    \begin{align*}
        \sup_{h \in \cH_1^+} \sqrt{n}(\P_n - P)[P[h] \cdot L \cdot h - P[L \cdot h] \cdot h - P[h]^2 \cdot \psi]
    \end{align*}
    converges to a tight limit. 
    
    Then, we can upper bound
    \begin{multline*}
        \lim_{n \to \infty} \P \l(\sup_{h \in \cH_1^+} \sqrt{n}(\P_n - P)[P[h] \cdot L \cdot h - P[L \cdot h] \cdot h - P[h_+]^2 \cdot \psi] > t^* \cdot \sqrt{n} \r) \\
        \leq \lim_{n \to \infty} \P \l(\sup_{h \in \cH_1} \sqrt{n}(\P_n - P)[P[h] \cdot L \cdot h - P[L \cdot h] \cdot h - P[h]^2 \cdot \psi] > t^* \cdot \sqrt{n} \r)
    \end{multline*}

    The bootstrap is consistent for $$\sup_{h \in \cH_1} \sqrt{n}(\P_n - P)[P[h] \cdot L \cdot h - P[L \cdot h] \cdot h - P[h]^2 \cdot \psi]$$
    because the function class inde`xing this process is $P$-Donsker. The same argument we used to prove that the process indexed by $h \in \cH^+_1$ is Donsker can be repeated here. 

    Via a derivation that is identical to the one above, we observe that the bootstrap process we sample from in \textsf{\Cref{algo:bootstrap_rkhs_appendix}} (see \Cref{lma:algo_trick}),
    $$
    \sqrt{n} \cdot \left(\P^*_n[(L - \hat{\theta}(\cD^*)) \cdot h] \cdot \P_n[h] - \P_n[(L - \hat{\theta}) \cdot h])] \cdot \P^*_n[h] \right),
    $$
    equals
    $$
    \sqrt{n} (\P^*_n - \P_n)[P[h] \cdot L \cdot h - P[L \cdot h] \cdot h - P[h]^2 \cdot \psi] + o_P(1).
    $$
    Thus, we conclude by \Cref{lma:donsker_bootstrap} and the continuous mapping theorem that the supremum of the bootstrap process is consistent for the distribution of the supremum of the limit process. 

    Last, we need to show that the bootstrap quantile $t^*$ is a consistent estimator of the true limiting quantile. We establish consistency of the bootstrap quantile by verifying that the limiting distribution has a continuous and strictly increasing CDF at its $(1 - \alpha)$-quantile \citep[Lemma 11.2.1(ii)]{lehmann2005testing}. The variance assumption on $L$ guarantees that for at least some $h \in \cH$, the limiting distribution of $$\sqrt{n}\l(\P_n[(L - \hat{\theta}) \cdot h]\cdot P[h] - P[(L - \theta_P) \cdot h] \cdot \P_n[h]\r)$$ is a non-trivial Gaussian, which then implies the desired CDF property. \Cref{cor:quantile_gaussian_2} implies consistency of $t^*$ and, thus, our desired claim regarding simultaneous coverage.
\end{proof}

We might adapt the bootstrap process for the RKHS so that the confidence bound width scales more naturally with the ``complexity'' of the shift chosen. Here we do not consider defining Wald-style confidence bounds, but rather simply aim to adjust the process so that the confidence bound width scales more naturally with the ``complexity'' of the queried shift. For example, the current bound scales as 
\begin{align*}
    \epsilon_{\text{lb}}(h) = \hat{\epsilon}(h) - \frac{C}{ \sqrt{n} \cdot \P_n[h]^2},
\end{align*}
while we might wish the bound to scale as
\begin{align*}
\epsilon_{\text{lb}}(h) = \hat{\epsilon}(h) - \frac{C}{\sqrt{n_{\text{eff}}}}
\end{align*}
where $n_{\text{eff}}$ quantifies the ``effective sample size'' of the reweighted metric. 

We motivate our choice of $n_{\text{eff}} = (\sum_i h(x_i))^2/\sum_i h(x_i)^2$ by observing that the variance of $(\sum_{i = 1}^n z_i \cdot h_i) / (\sum_{i = 1}^n h_i)$ for $z_i \simiid (0, 1)$ scales as $1/\sqrt{n_{\text{eff}}}$. 

To motivate our choice of $\hat{s}(h)$, observe that rescaling the process by $1/\hat{s}(h)$ yields a bound of the form 
\begin{align*}
    \epsilon_{\text{lb}}(h) = \hat{\epsilon}(h) - C \frac{\hat{s}(h)}{\P_n[h]^2 \cdot \sqrt{n}}.
\end{align*}
So, if we solve for $\hat{s}(\cdot)$ that yields the desired $\sqrt{n_{\text{eff}}}$ denominator, we obtain $|\P_n[h]| \cdot \sqrt{\P_n[h^2]}$. Truncating to ensure uniform consistency, we define
\begin{align*}
    \hat{s}(h) \defeq \max\l ( |\P_n[h]| \cdot \sqrt{\P_n[h^2]}, \,h_* \r),
\end{align*}
for some threshold $h_*$. Unlike $p_*$ in \eqref{eqn:rescaling_est}, note that $h_*$ is not interpretable. 

Besides the lack of interpretability, the rescaled RKHS process can no longer be efficiently bootstrapped. Even if we assume that $\theta_P$ is known, we must now compute in line 8 of \textsf{\Cref{algo:bootstrap_rkhs_appendix}},
\begin{align*}
    t^{(b)} = \sup_{h \in \cH_1} \frac{\P^b_n[L \cdot h]\cdot \P_n[h] - \P_n[L \cdot h] \cdot \P^b_n[h]}{\max\l ( |\P_n[h]| \cdot \sqrt{\P_n[h^2]}, \,h_* \r)}.
\end{align*}
While one might hope to mimic our previous approach and reduce this computation to some eigenvalue problem, applying the finite-dimensional representation of the RKHS function only yields
\begin{align*}
    \sup_{\beta : \|\beta\|_2 \leq 1} \frac{\l(\beta^\top \mathbf{K}^{1/2} \l(\mathbf{A} + \mathbf{A}^\top \r) \mathbf{K}^{1/2} \beta \r)}{2 \cdot \max\l ( |\mathbf{1}^\top K^{1/2} \beta| \cdot \sqrt{\beta^\top K \beta}, \,h_* \r)}.
\end{align*}
The $\sqrt{\beta^\top K \beta}$ term in the denominator makes optimizing $\beta$ extremely challenging. Since $K$ is full-rank for a positive definite kernel, we cannot rely on any low-rank structure in $A$ to simplify this problem. At best, a rank-$m$ approximation to $K$ yields an intractable (and inaccurate) optimization problem over the surface of a $(m + 4)$-dimensional hypersphere. If we simply drop the $\sqrt{\beta^\top K \beta}$ term from $\hat{s}(G)$, the bootstrap step can be reduced to an optimization problem over the surface of a $4$-dimensional hypersphere. We can solve that problem via a brute-force search, but, in practice, we find that the resulting confidence bounds are not improved. As a consequence, we recommend against rescaling the RKHS process. 
\end{document}